\documentclass[runningheads,a4paper,envcountsame]{llncs}

\usepackage[T1]{fontenc}
\usepackage{graphicx}
\usepackage{amsmath}
\usepackage{amsfonts}
\usepackage{amssymb}
\usepackage{mathtools}
\usepackage{booktabs}
\usepackage{xspace}
\usepackage[inline]{enumitem}
\usepackage[noend]{algpseudocode}

\newcommand{\BlankLine}{\State}
\usepackage{pifont}
\usepackage{listings}
\usepackage{xcolor}
\usepackage{tikz}
\usetikzlibrary{arrows.meta}
\usepackage[hidelinks]{hyperref}
\usepackage{bookmark}
\usepackage[capitalize,noabbrev]{cleveref}

\hypersetup{hypertexnames=false,bookmarksdepth=3,bookmarksopen=true}
\newif\ifshowcomments
\showcommentsfalse
\newcommand{\authorcomment}[3]{\ifshowcomments{\footnotesize\color{#1}[#2: #3]}\fi}

\newcommand{\giuliano}[1]{\authorcomment{violet}{Giuliano}{#1}}
\newcommand{\daniel}[1]{\authorcomment{blue}{Daniel}{#1}}

\newcommand{\party}[1]{\ensuremath{p_{#1}}\xspace}
\newcommand{\True}{\mathsf{true}}
\newcommand{\False}{\mathsf{false}}
\newcommand{\None}{\mathsf{none}}
\newcommand{\var}[1]{\mathsf{#1}}
\newcommand{\fn}[1]{\mathbf{\mathsf{#1}}}
\newcommand{\sig}[1]{\langle #1 \rangle}
\newcommand{\msgtag}[1]{\textup{\textsc{#1}}}
\newcommand{\Vote}{\msgtag{vote}}
\newcommand{\Ready}{\msgtag{ready}}
\newcommand{\Propose}{\msgtag{propose}}
\newcommand{\CommitMsg}{\msgtag{commit}}
\newcommand{\Candidate}{\msgtag{candidate}}

\newcommand{\Disable}{\msgtag{disable}}
\newcommand{\hash}[1]{\ensuremath{\mathsf{H}(#1)}}
\newcommand{\nrecv}[1]{\#\!\bigl[\,#1\,\bigr]}

\newcommand{\kw}[1]{\textbf{#1}}

\newcommand{\ind}{\hspace{0.3cm}}
\newcommand{\indd}{\hspace{0.6cm}}
\newcommand{\inddd}{\hspace{0.9cm}}

\makeatletter
\newcounter{savedalgline}

\makeatother

\title{Simple and Fast Signature-Free Blockchain Consensus}
\titlerunning{Simple and Fast Signature-Free Blockchain Consensus}

\author{Giuliano Losa\inst{1}\thanks{Authors are listed in alphabetical order. Giuliano Losa is the lead author.} \and Xuechao Wang\inst{2} \and Zhuolun Xiang\inst{3} \and Qianyu Yu\inst{2}}
\authorrunning{}

\institute{%
Stellar Development Foundation\\
\email{giuliano@stellar.org} \and
The Hong Kong University of Science and Technology (Guangzhou)\\
\email{xuechaowang@hkust-gz.edu.cn, qyu100@connect.hkust-gz.edu.cn} \and
Aptos Labs\\
\email{xiangzhuolun@gmail.com}%
}

\begin{document}

\maketitle

\begin{abstract}
Signature-free protocols avoid the cost of post-quantum signatures.
We present two simple signature-free blockchain consensus protocols for eventual synchrony with optimal good-case commit latency (three message delays for \(f<n/3\) and two for \(f<n/5\)), optimistic responsiveness, a block time of only two message delays without speculation, and \(O(n^2)\) communication per view.
They instantiate Generic Simplex, a blockchain consensus construction parameterized by a new abstraction called view agreement.
The same construction also captures Simplex, Minimmit, and a new synchronous signature-free protocol with optimal good-case commit latency of two message delays for \(f<n/4\).
\end{abstract}

\section{Introduction}

State-machine replication (SMR)~\cite{lamport_time_1978} is a popular technique to create fault-tolerant systems, from lock services in datacenters to public blockchains.
The idea is for a set of network nodes, or parties, to each maintain a local copy of a deterministic state machine, initially in the same state, and run a consensus protocol to agree on an ever-growing sequence of commands, called the log, to apply to their local copies of the state machine.
To the outside world, this creates the illusion of a single state machine that can be accessed at any of the parties.

Early practical consensus protocols like Multi-Paxos~\cite{lamport2001paxos} and PBFT~\cite{castro1999practical} are notorious for being hard to understand.
This spurred the development of simpler protocols, and in particular of the family of streamlined blockchain protocols.
Examples include Chained HotStuff~\cite{yin2019hotstuff} and Streamlet~\cite{chan2020streamlet}.

Streamlined protocols proceed in numbered \emph{views}.
In each view, a designated leader proposes a \emph{block}, parties vote for it if it satisfies the protocol's rules, and the block becomes \emph{prepared} (notarized in Simplex) if it receives enough votes.
Except for the special \emph{genesis} block, each block consists of a sequence of commands and a reference (often a hash) to a unique \emph{parent} block, forming a \emph{blockchain}.
Parties try to propose blocks that extend recently prepared blocks, and this results in a tree of prepared blocks.
A block gets committed once a particular pattern of prepared blocks emerges.
For example, in Streamlet, a chain of three blocks prepared in consecutive views commits the block in the middle~\cite{chan2020streamlet}.

Streamlined protocols repeat the same block-preparation procedure in each view, providing a simple structure that pipelines block commitment across views.
However, committing a block depends on votes that prepare its descendants, so faulty leaders in later views can prevent commitment.
Simplex~\cite{chanSimplexConsensusSimple2023} preserves this pipelining while using a separate voting round to commit a prepared block.
Thus, preparing the next block and committing the current block proceed concurrently, without commitment depending on progress in later views.
However, Simplex relies on transferable signatures.
Avoiding such signatures is particularly attractive for post-quantum-resilient systems, where post-quantum signature schemes can impose substantial costs~\cite{aaronsonQuantumComputingBlockchain2026}.
\giuliano{Shortened a bit but retained the Simplex story, added quantum issue. Better?}

In this paper, we tackle the problem of applying the Simplex architecture to obtain Byzantine fault-tolerant blockchain consensus protocols that do not use transferable signatures and nevertheless achieve optimal good-case commit latency.
This is challenging because, without transferable signatures, parties cannot prove to each other what they saw.
Simplex, like many protocols, relies on forwarding proofs that \(n-f\) parties\footnote{As customary, \(f\) is the maximum number of Byzantine parties and \(n\) is the total number of parties.} voted for a block.
Shoup's signature-free variant of Simplex~\cite[Section 7]{shoupSingSongSimplexEprint} replaces these proofs with reliable broadcast, which increases commit latency from three to six message delays.
Even for single-shot eventually synchronous consensus, Forget-IT~\cite{abrahamForgetITOptimalGoodCase2026} (PODC 2026) only recently achieved the optimal good-case commit latency of three message delays and \(O(n^2)\) communication per view under the optimal resilience of~\(f<n/3\).

The first contribution of this paper is to present two simple signature-free blockchain consensus protocols, under eventual synchrony, that achieve the optimal good-case commit latencies~\cite{abraham_good-case_2021,abraham2022good} of three message delays when \(f<n/3\) and two message delays when \(f<n/5\), respectively.
Both are optimistically responsive, have a block time of only two message delays, without speculation\footnote{Correct leaders extend only prepared blocks, so, after GST, their proposals commit regardless of previous proposals.}, and, like Simplex, both allow block commitment to complete independently of progress in later views.
Under maximal protocol resilience (\(n=3f+1\) and \(n=5f+1\), respectively), both send \(O(n^2)\) bits per view, assuming fixed-size blocks and message overhead.

The second contribution, which may be of wider interest, is Generic Simplex, a blockchain consensus protocol parameterized by an implementation of a new abstraction that we call view agreement.
View agreement is specified without timing assumptions, and it captures what parties must agree on in a view: which blocks are prepared or committed, and whether later views are authorized to bypass it.
We obtain the two signature-free protocols by plugging in two different implementations of view agreement.

View agreement supports a \(\fn{vote}(x)\) and a \(\fn{request\_disable}()\) invocation, and exposes a set of prepared values, an optional committed value, and a disabled flag.
In Generic Simplex, the values are blocks, and later views must build on one of a view's prepared blocks unless that view is disabled.
If a correct party commits a value, no correct party can ever prepare or commit a different value, or disable the view.
This is similar to classic abstractions like adopt-commit and its graded variants~\cite{gafni_round-by-round_1998,attiya_multi-valued_2023}.

However, view agreement has two more unusual properties compared to classic abstractions.
First, it guarantees eventual agreement on prepared values and disabled status; this allows correct proposals in later views to be eventually accepted by all correct parties without having to forward proofs.
Second, view agreement may prepare multiple values.
This lets our view agreement protocols prepare a value in just one message delay (enabling a block time of two message delays in Generic Simplex) while still ensuring eventual agreement on prepared values.
In fact, for signature-free implementations with \(n/4<f<n/3\), we prove this impossible if at most one value may be prepared (\Cref{app:va-multiple-prepared}).

With signed view agreement implementations, Generic Simplex recovers the core of Simplex (\Cref{app:simplex}) and of Minimmit~\cite{chouMinimmitFC2026} (\Cref{app:minimmit}).
Minimmit requires multiple prepared values: its prepare quorums need not share a correct party, so an equivocating leader can get several blocks prepared in one view.
Generic Simplex also works, unchanged, in synchrony, where view agreement implementations may rely on timing.
This yields a signature-free protocol for \(f<n/4\) with the optimal good-case commit latency of two message delays~\cite{abraham_good-case_2021}~(\Cref{app:synchronous}).

Generic Simplex follows Abraham’s decomposition of Simplex into an outer protocol that drives instances of an inner per-view protocol~\cite{abrahamDeconstructingSimplex2026}, which Abraham and Neu also apply to chained Simplex~\cite{abrahamNeuChainedSimplex2026}; it keeps the shape of their outer proposal rule and changes the specification of the inner protocol (see \Cref{sec:related-work}).

\Cref{sec:model} presents the model and defines blockchain consensus and its latency metrics, \Cref{sec:blockchain-va} specifies view agreement and presents Generic Simplex, \Cref{sec:signature-free-implementations} presents 1/3-VA, outlines Fast VA, and gives the resulting latency bounds, and \Cref{sec:related-work} discusses related work.
The appendices contain the correctness proofs (\Cref{app:consensus-proofs,app:va-proofs}), a tweak to shorten view timeouts (\Cref{app:strong-unanimity}), the 1/5-resilient Fast VA protocol (\Cref{app:va-fifth}), the synchronous protocol (\Cref{app:synchronous}), and the Minimmit and Simplex instantiations (\Cref{app:minimmit,app:simplex}).

\section{Computation Model and Blockchain Consensus}
\label{sec:model}

We consider a static Byzantine adversary in an asynchronous or eventually synchronous message-passing system with an authenticated channel between every pair of parties.
The model and its assumptions are standard for the field.

There are \(n\) parties, of which at most \(f\) may be Byzantine, and, if a party is not Byzantine, we say it is correct.
Byzantine parties may send arbitrary messages, but only on the channels connected to them.
Messages on channels connected to correct parties cannot be forged, and thus if a party \(p\) receives a message directly from a party \(q\), then \(p\) can be sure that \(q\) sent it.
In practice, such channels can be implemented using symmetric cryptography like HMAC message authentication codes.
In the special case in which Byzantine parties do not send messages, we say that the execution is Byzantine-silent, or eventually Byzantine-silent if they only eventually stop sending messages.

Message delivery is reliable, meaning that every message sent by a correct party to a correct party is eventually delivered.
Under asynchrony, there is no bound on message delay.
Under eventual synchrony (the GST model of partial synchrony), there is an unknown global stabilization time (GST) after which the message delay between correct parties is bounded by a known constant \(\Delta\), and messages sent before GST arrive by \(\text{GST}+\Delta\).
When discussing latency, we use \(\delta\) for the maximum message delay during the period under consideration.

We assume that every party starts executing at time \(0\), and that parties have local clocks that may have arbitrary offset but that advance at real-time rate.
At a timer's deadline, parties process all arriving messages and the resulting local actions before firing the timer.

\Cref{app:extended-model} defines the additional PKI and timing assumptions used by some appendix protocols and lower bounds.

\paragraph{Notation in pseudocode.}
In the protocol figures, \(\nrecv{P_1\text{ or }\cdots\text{ or }P_k}\) denotes the number of distinct parties from which the executing party has received a message matching at least one of the message patterns \(P_1,\ldots,P_k\).
Each sender is counted at most once, even if messages from that sender match multiple patterns.
All such counts are local; any figure-specific rule for handling equivocation applies before counting.

\subsection{Blockchain Consensus}

In the blockchain consensus problem, parties maintain an ever-growing blockchain by repeatedly creating and committing new valid blocks.

A block is either the special, empty genesis block or a pair \(b = \sig{\var{payload}, h}\), where \(\var{payload}\) is an application-specific payload (usually, a list of transactions) and \(h\) is a hash value that determines a unique parent block \(b'\).
We write \(h=\hash{b'}\) and \(\fn{parent}(b)=b'\).
We say that a block \(b'\) is an ancestor of a block \(b\) and, equivalently, that \(b\) is a descendant of \(b'\), when \(b'=b\) or \(b'\) is reachable from \(b\) by following parent links.
We also say that two blocks \(b\) and \(b'\) are compatible when one is an ancestor of the other, and incompatible otherwise.
We assume that following parent links eventually reaches genesis; thus, a block \(b\) determines a unique sequence of blocks starting at genesis and ending at \(b\).

A block may be valid or invalid according to some application-specific rules (this is sometimes called external validity).
We assume a predicate \(\fn{valid}(b)\) that identifies valid blocks, and we assume that the genesis block is valid.
To create new valid blocks, parties can call \(\fn{extend}(b)\), where \(b\) is a valid block, which returns a new valid block \(b'\) such that \(\fn{parent}(b')=b\).
We assume that both \(\fn{valid}\) and \(\fn{extend}\) can be evaluated locally in negligible time.
Parties may commit blocks by calling \(\fn{commit}(b)\), which by convention also commits all ancestors of \(b\).

A blockchain consensus protocol must provide the following guarantees:
\begin{itemize}[noitemsep,leftmargin=*]
    \item[-] \textbf{Consistency.} Correct parties never commit incompatible blocks.
    \item[-] \textbf{Totality.} If a correct party commits a block \(b\), then every correct party eventually commits \(b\).
    \item[-] \textbf{Validity.} Correct parties never commit invalid blocks.
    \item[-] \textbf{Liveness.} Every correct party obtains infinitely many blocks through calls to \(\fn{extend}\) and commits infinitely many of them.
\end{itemize}

The spirit of liveness is that every correct party must repeatedly get the opportunity to create blocks and get its blocks committed.
The concrete protocols we present satisfy the stronger guarantee that every block created by a correct party sufficiently far after GST commits.
To quantify progress guarantees more precisely, we use the following latency metrics.

\begin{definition}[Good-case commit latency \(\ell_g\)]
    The good-case commit latency \(\ell_{g}\) of a blockchain consensus protocol is the maximum time it takes, assuming \(\text{GST}=0\), for a block created by a correct party at time \(0\) to be committed by all correct parties.
\end{definition}
When \(n=3f+1\) (optimal resilience) and \(f\geq2\), every signature-free protocol has \(\ell_{g}\geq3\delta\), even in the synchronous setting of Appendix~\ref{app:extended-model}~\cite{abraham2022good}.
Under eventual synchrony, this lower bound holds even with signatures~\cite{abraham_good-case_2021}.
The bound is tight: Forget-IT~\cite{abrahamForgetITOptimalGoodCase2026} and our 1/3-resilient protocol (\Cref{sec:signature-free-implementations}) attain it.
When \(f<n/5\), our consensus protocol based on Fast VA (\Cref{app:va-fifth}) attains \(\ell_{g}=2\delta\), which is also optimal~\cite{abraham_good-case_2021}.\footnote{The papers referenced prove lower bounds for single-shot broadcast, but they carry over by reduction: the party that creates a block at time~\(0\) can broadcast a value by including it in that block, and every party outputs the value in the first non-genesis block of its committed chain.}

Good-case commit latency is convenient for lower bounds, but it only concerns the beginning of an execution.
Thus we use the following additional metrics:

\begin{definition}[Steady-state commit latency \(\ell_s\)]
    The steady-state commit latency \(\ell_{s}\) is the least bound such that, in every eventually Byzantine-silent execution, there is a time \(T\geq \text{GST}\) such that, after \(T\), every block created by a correct party is committed within~\(\ell_{s}\).
\end{definition}

\begin{definition}[Eventual worst-case commit latency \(\ell_w\)]
    The eventual worst-case commit latency \(\ell_{w}\) is the least bound such that, in every execution, there is a time \(T\geq \text{GST}\) such that, after \(T\), every block created by a correct party is committed within~\(\ell_{w}\).
\end{definition}

\begin{definition}[Block time]
    The block time is the worst-case average time between two consecutive calls to \(\fn{extend}\), after GST and assuming no failures.
\end{definition}

Note that steady-state commit latency is at most the eventual worst-case commit latency.
Together, steady-state commit latency and block time characterize what clients can expect during normal operation.
For example, the 1/3-resilient protocol of \Cref{sec:signature-free-implementations} has a block time of \(2\delta\) and a steady-state commit latency of \(3\delta\).
This means that a client that broadcasts a transaction to all parties can on average expect to wait \(\delta\) for the next block proposal and another \(3\delta\) for it to be committed at all correct parties.
The total end-to-end latency is therefore \(4\delta\) plus the round-trip to and from the client.

\begin{definition}[Optimistic responsiveness]
    A blockchain consensus protocol is optimistically responsive if \(\ell_g\), \(\ell_s\), \(\ell_w\), and its block time admit bounds that vanish as \(\delta\to0\) when \(\Delta\) is fixed.
\end{definition}
In other words, in good conditions, an optimistically responsive protocol progresses at the speed of the actual network delay, and not its upper bound \(\Delta\).

\begin{definition}[\(k\)-gap latency]
    The \(k\)-gap latency is the eventual worst-case time between two consecutive calls to \(\fn{extend}\) by correct parties, assuming at most \(k\) Byzantine parties.
\end{definition}
The \(k\)-gap latency bounds how long \(k\) Byzantine parties can delay the next correct block.
The \(0\)-gap latency bounds the block time.

\section{Blockchain Consensus from View Agreement}%
\label{sec:blockchain-va}

In this section, we specify view agreement and present Generic Simplex through an informal walkthrough of its pseudocode in~\Cref{fig:consensus-va} (line numbers below refer to \Cref{fig:consensus-va}).
Detailed proofs appear in~\Cref{app:consensus-proofs}.

\begin{figure*}[tp]
    \centering
    \small
    \setlength{\fboxsep}{2pt}%
    \fbox{%
    \begin{minipage}{\dimexpr\textwidth-2\fboxsep-2\fboxrule\relax}
    \raggedright
    \begin{algorithmic}[1]
        \State \textbf{State variables and shorthands:}
        \State \ind $\var{curr\_view} \gets 1$;
        \State \ind $\var{va} \gets [v \mapsto \text{a fresh, active view agreement instance}]$;\label{ln:consensus-va-instances}
        \State \ind $\var{proposal} \gets [v \mapsto \None]$;
        \State \ind $\var{voted} \gets \False$;
        \State \ind $\var{disable\_requested} \gets \False$;
        \State \ind $\fn{prepared}(v) =$ \kw{if} $v = 0$ \kw{then} $\{\text{genesis}\}$
            \kw{else} $\var{va}[v].\var{prepared}$;
        \State \ind $\fn{disabled}(v) =$ \kw{if} $v = 0$ \kw{then} $\False$
            \kw{else} $\var{va}[v].\var{disabled}$;
        \State \ind $\fn{committed}(v) =$ \kw{if} $v = 0$ \kw{then} $\text{genesis}$
            \kw{else} $\var{va}[v].\var{committed}$;
        \BlankLine
        \State \kw{upon} startup: \kw{call} $\fn{enter\_view}(1)$;
        \BlankLine
        \State \kw{procedure} $\fn{enter\_view}(v)$:
        \State \ind $\var{curr\_view} \gets v$;
            $\var{voted} \gets \False$; $\var{disable\_requested} \gets \False$;
        \State \ind \kw{call} $\fn{reset\_view\_timer}()$;\label{ln:consensus-reset-timer}
        \State \ind \kw{if} $\fn{leader}(v) = \text{self}$: \kw{call} $\fn{propose}()$;
        \BlankLine
        \State \kw{predicate} $\fn{is\_safe}(b = \sig{\var{payload}, h})$:\label{ln:consensus-is-safe}
        \State \ind \kw{return} $\exists v', b':$
        \State \indd $0 \le v' < \var{curr\_view}$ \kw{and} $\hash{b'} = h$ \kw{and}
            $b' \in \fn{prepared}(v')$
        \State \indd \kw{and} $\forall v'' \in \{v'+1,\ldots,\var{curr\_view}-1\}:
            \fn{disabled}(v'') = \True$;\label{ln:consensus-check-disabled}
        \BlankLine
        \State \kw{procedure} $\fn{propose}()$:\label{ln:consensus-propose}
        \State \ind \kw{with} $\var{vmax} = \max \{v < \var{curr\_view} \mid
            \fn{prepared}(v) \neq \emptyset\}$
        \State \ind \kw{with} $b' \in \fn{prepared}(\var{vmax})$:
        \State \indd $b \gets \fn{extend}(b')$;
        \State \indd \kw{broadcast} $\sig{\Propose, \var{curr\_view}, b}$;\label{ln:consensus-broadcast}
        \BlankLine
        \State \kw{upon} $\sig{\Propose, v, b}$ from the leader of view $v$:
        \State \ind \kw{if} $\var{proposal}[v] = \None$: $\var{proposal}[v] \gets b$;
        \BlankLine
        \State \kw{when} \kw{not} $\var{voted}$ \kw{and} \kw{not} $\var{disable\_requested}$\label{ln:consensus-vote-guard}
        \State \indd \kw{and} $\var{proposal}[\var{curr\_view}] \neq \None$
            \kw{and} $\fn{is\_safe}(\var{proposal}[\var{curr\_view}])$:
        \State \ind \kw{call} $\var{va}[\var{curr\_view}].\fn{vote}(\var{proposal}[\var{curr\_view}])$;\label{ln:consensus-vote}
        \State \ind $\var{voted} \gets \True$;
        \BlankLine
        \State \kw{upon} view timer firing:
        \State \ind \kw{if} \kw{not} $\var{disable\_requested}$:
        \State \indd \kw{call} $\var{va}[\var{curr\_view}].\fn{request\_disable}()$;
            $\var{disable\_requested} \gets \True$;\label{ln:consensus-request-disable}
        \BlankLine
        \State \kw{when} $\fn{prepared}(\var{curr\_view}) \neq \emptyset$ \kw{or} $\fn{disabled}(\var{curr\_view})$:\label{ln:consensus-cleared}
        \State \ind \kw{if} \kw{not} $\var{voted}$ \kw{and} \kw{not} $\var{disable\_requested}$
            \kw{and} $\fn{prepared}(\var{curr\_view}) \neq \emptyset$:
        \State \indd \kw{with} $b \in \fn{prepared}(\var{curr\_view})$:
        \State \inddd \kw{call} $\var{va}[\var{curr\_view}].\fn{vote}(b)$;
            $\var{voted} \gets \True$;\label{ln:consensus-catchup-vote}
        \State \ind \kw{call} $\fn{enter\_view}(\var{curr\_view} + 1)$;\label{ln:consensus-advance}
        \BlankLine
        \State \kw{upon} $\fn{committed}(v)$ becoming $b \neq \None$: \kw{call} $\fn{commit}(b)$;\label{ln:consensus-commit}
    \end{algorithmic}
    \end{minipage}}
    \caption{The Generic Simplex protocol.}%
    \label{fig:consensus-va}
\end{figure*}

\paragraph{Views and the view agreement interface.}
As in Simplex, parties proceed in a sequence of numbered views, starting with view 1, and each party has a (possibly different) current view. %
Each view has a pre-determined leader party, chosen round-robin, that may pick a block \(b\) to extend, create a new block using \(\fn{extend}(b)\), and propose it to the other parties (line~\ref{ln:consensus-broadcast}).
If the proposed block is \emph{safe} in the view (defined shortly), the other parties may vote for it in an instance of view agreement assigned to the view (line~\ref{ln:consensus-vote}); that instance may then prepare and later commit the proposed block.

Each party has a view timer, which it resets when incrementing its current view (line~\ref{ln:consensus-reset-timer}).
The timer is set to expire after a duration \(\Delta_{\mathrm{to}}\) that depends on the view agreement implementation used (more about this below), and if the timer fires before the party has prepared a block in the current view, the party requests to disable its current view (line~\ref{ln:consensus-request-disable}).

To vote, request disabling, and learn which blocks are prepared or committed and whether the view is disabled, parties use the view agreement interface.
It provides two invocations, \(\fn{vote}(x)\) and \(\fn{request\_disable}()\), and a party must vote at most once in an instance and never after requesting disabling.
In return, the interface exposes three output variables: \(\var{prepared}\), an initially empty, growing set of blocks; \(\var{disabled}\), initially \(\False\) and only allowed to change to \(\True\); and \(\var{committed}\), initially \(\None\) and changing at most once, to a block.
We say that the instance is \emph{cleared} at a party when \(\var{prepared}\neq\emptyset\) or \(\var{disabled}\) holds.
Moreover, all view agreement instances are always active at a party regardless of the party's view, so that their liveness guarantees always apply.\footnote{Standard techniques, e.g., lazy instantiation, avoid infinite state.}

\paragraph{View agreement guarantees.}\label{par:va-guarantees}
View agreement guarantees the following properties, where all parties mentioned are assumed correct.
\begin{itemize}[leftmargin=*,itemsep=3pt,parsep=0pt,topsep=3pt]
    \item \textbf{VA-consistency.}
    If a party commits \(x\), no party ever disables the view or prepares or commits \(y\neq x\).
    A party commits \(x\) only after preparing it.
    \item \textbf{VA-validity.}
    A prepared or committed value \(x\) is valid and was previously voted for by some party.
    \item \textbf{VA-unanimity.}
    If all parties vote for the same valid value \(x\) and none requests disabling, all eventually commit~\(x\).
    \item \textbf{VA-totality.}
    If a party prepares \(x\), all eventually prepare \(x\); if a party disables, all eventually disable.
    \item \textbf{VA-fallback-progress.}
    If all parties request disabling, the instance is eventually cleared at all parties.
\end{itemize}

\paragraph{Ensuring consistency and validity.}
Note that VA-consistency ensures consistency within a view, i.e., no two different blocks may be committed by correct parties in a view.
To ensure consistency of committed blocks across views, we define a block \(b\) as safe in a view \(v\) (lines~\ref{ln:consensus-is-safe}--\ref{ln:consensus-check-disabled}) when there exists \(v'<v\) such that:
\begin{enumerate}
    \item The parent of \(b\) is prepared in view \(v'\), or \(v'=0\) and the parent of \(b\) is the genesis block, and
    \item All views strictly between \(v'\) and \(v\) are disabled.
\end{enumerate}
Note that, by VA-totality, a block that is safe in a view at a party eventually becomes so at all correct parties; thus safety of a block in a view is a global property that parties converge on over time.

Now correct parties vote either for a proposal they have checked to be safe (lines~\ref{ln:consensus-vote-guard}--\ref{ln:consensus-vote}) or, at line~\ref{ln:consensus-catchup-vote}, for a block that is already prepared.
By VA-validity, a block is prepared only after some correct party votes for it, so the first correct vote for any block has checked it safe.
Hence all prepared and committed blocks are safe in their view.

Moreover, VA-consistency ensures that a view with a committed block cannot be disabled or have a different prepared block.
Thus, the parent relation (restricted to prepared and committed blocks) cannot ``jump over'' committed blocks, and all committed blocks form a single chain, as illustrated in \Cref{fig:consistency-example}. %

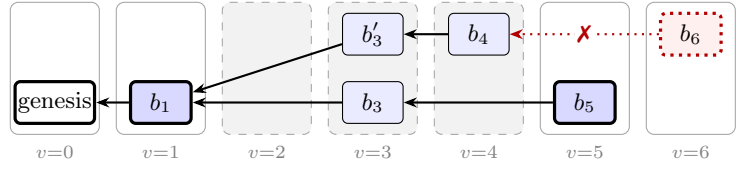
\begin{figure}[!htbp]
  \centering
  \begin{tikzpicture}[
      x=1.4cm,
      view/.style={draw=gray!70, rounded corners=3pt},
      disabled/.style={view, dashed, fill=gray!10},
      block/.style={draw, rounded corners=2pt, minimum width=0.8cm,
        minimum height=0.55cm, font=\small, fill=blue!8},
      committed/.style={block, fill=blue!15, very thick},
      unsafe/.style={draw=red!70!black, dotted, very thick, rounded corners=2pt,
        minimum width=0.8cm, minimum height=0.55cm, font=\small, fill=red!6},
      parent/.style={-{Stealth[length=5pt]}, thick},
      vlabel/.style={font=\scriptsize\sffamily, gray},
    ]
    \foreach \v in {0,1,5,6} {
      \draw[view] (\v-0.42, -0.42) rectangle (\v+0.42, 1.32);
    }
    \foreach \v in {2,3,4} {
      \draw[disabled] (\v-0.42, -0.42) rectangle (\v+0.42, 1.32);
    }
    \foreach \v in {0,1,2,3,4,5,6} {
      \node[vlabel] at (\v, -0.65) {$v{=}\v$};
    }

    \node[block, fill=white, very thick, inner xsep=1pt] (b0) at (0,0) {genesis};
    \node[committed] (b1)  at (1,0)   {$b_1$};
    \node[block]     (b3)  at (3,0)   {$b_3$};
    \node[block]     (b3p) at (3,0.9) {$b_3'$};
    \node[block]     (b4)  at (4,0.9) {$b_4$};
    \node[committed] (b5)  at (5,0)   {$b_5$};
    \node[unsafe]    (b6)  at (6,0.9) {$b_6$}; %

    \draw[parent] (b1) -- (b0);
    \draw[parent] (b3) -- (b1);   %
    \draw[parent] (b3p) -- (b1);  %
    \draw[parent] (b4) -- (b3p);
    \draw[parent] (b5) -- (b3);   %
    \draw[parent, red!70!black, dotted] (b6) --
      node[midway, fill=white, inner sep=1pt, text=red!70!black] {\ding{55}} (b4); %
  \end{tikzpicture}
  \caption{An example execution of Generic Simplex.
Columns are views, and each block sits in the view in which it is prepared (or proposed, for~\(b_6\)), with an arrow to its parent.
Thick blue blocks are committed, light blue blocks are only prepared, and dashed gray views (\(2,3,4\)) are disabled.
View~\(3\) prepares two blocks, which is possible because it commits none.
Every prepared block is safe in its view; e.g., \(b_5\)'s parent link skips the disabled view~\(4\).
In contrast, \(b_6\) is not safe in view~\(6\): its parent link jumps over view~\(5\), which has a committed block and thus cannot be disabled.
Hence every later committed block descends from~\(b_5\), and neither \(b_3'\) nor \(b_4\) will ever be committed.
The committed blocks and their ancestors form the chain \(\text{genesis}, b_1, b_3, b_5\).}%
  \label{fig:consistency-example}
\end{figure}

VA-validity moreover ensures that committed blocks are valid, and this covers the whole safety argument of the protocol.
Next we explain how the consensus protocol ensures that parties keep starting new views and how, in views started after GST, it guarantees that the proposals of correct leaders commit.

\paragraph{Progression from one view to the next.}
When its current view \(v\) becomes cleared, a party enters \(v+1\) (line~\ref{ln:consensus-advance}).
Normally, this happens as soon as the party prepares a block in \(v\), without waiting for it to commit, as in Simplex.
Since a party only votes or requests disabling in its current view, later views, even with faulty leaders, cannot affect whether \(v\) commits (\Cref{lem:consensus-good-view}).
Because VA-unanimity needs every correct vote, a party that prepares a block before voting (e.g., it entered \(v\) late or has not yet learned that the proposal is safe) votes for a prepared block before advancing (line~\ref{ln:consensus-catchup-vote}).
If instead the current view \(v\) fails to prepare a block at a party before its timer fires, the party requests to disable \(v\).
Globally, either some correct party clears \(v\) or all correct parties time out.
In the former case, VA-totality guarantees that all parties eventually clear the view; in the latter case, VA-fallback-progress guarantees that all parties eventually clear the view too.
Thus, every correct party always eventually increments its view.

\paragraph{Ensuring commits after GST.}
To ensure that every proposal made by a correct leader in a view \(v\) started after GST commits, by VA-unanimity it suffices to ensure that all correct parties vote for the leader's proposal and that none requests to disable the view.
To ensure the former, a correct leader makes sure to propose a block that is safe in its current view.
It does so by picking the highest earlier view \(v'\) in which it has a prepared block, choosing a prepared block of that view (there may be multiple), say \(b'\), creating a new block \(b=\fn{extend}(b')\), and proposing \(b\).
Since the leader must have cleared every earlier view, all views between \(v'\) and \(v\) are disabled and \(b\) is safe in view \(v\), which, by VA-totality, all correct parties eventually learn.
Thus, assuming no previous requests to disable, each correct party votes for the leader's proposal either upon receiving it and learning that it is safe (lines~\ref{ln:consensus-vote-guard}--\ref{ln:consensus-vote}), or upon preparing it (line~\ref{ln:consensus-catchup-vote}).

\phantomsection\label{par:va-latency}
Next, we must ensure no correct party requests to disable the view.
To pick a sufficiently large timer value \(\Delta_{\mathrm{to}}\), we use the following \emph{latency parameters} of view agreement.
With message delay bounded by \(\delta\), these are worst-case delay bounds, each assumed at least \(\delta\).
All parties mentioned in these definitions are assumed correct unless stated otherwise.
\begin{itemize}[leftmargin=*,itemsep=3pt,parsep=0pt,topsep=3pt]
    \item \emph{Prepare and commit delays} \(d_p(\delta)\) and \(d_c(\delta)\).
    The delay from the last correct vote until all correct parties prepare \(x\) or commit \(x\), respectively, under the conditions of VA-unanimity.
    \item \emph{Fallback delay} \(d_f(\delta)\).
    The delay from the last correct party's request to disable until the instance is cleared at all parties.
    \item \emph{Totality delay} \(d_t(\delta)\).
    The delay from the first correct preparation of a value \(x\) until all correct parties prepare \(x\), or from the first correct disabling until all correct parties disable, whichever is larger.
    \item \emph{Byzantine-silent totality delay} \(d_t^s(\delta)\leq d_t(\delta)\).
    The totality delay if Byzantine parties send no messages in the instance; we call the instance Byzantine-silent.
\end{itemize}
Starting from the first correct entry into a view with a correct leader after GST, we allow one totality delay for all parties to enter and the leader to propose, another to receive the proposal and learn that it is safe, and one prepare delay before any timer fires.
Thus, it suffices to choose \(\Delta_{\mathrm{to}}\geq 2d_t(\Delta)+d_p(\Delta)\).
The prepare delay is needed because VA-unanimity does not guarantee a commit if some party requests to disable the view after voting.
Some view agreement implementations guarantee it anyway, a property that we call VA-strong-unanimity, and then \(\Delta_{\mathrm{to}}\geq 2d_t(\Delta)\) suffices (see \Cref{app:strong-unanimity}).
Since correct parties keep starting new views and, after GST, views with correct leaders commit, we obtain blockchain consensus liveness.

\paragraph{Latency bounds.}
The view agreement delays give the following bounds for the latency metrics of Generic Simplex:
\begin{align*}
\text{good-case commit latency }\ell_g &\leq \delta+d_c(\delta),\\
\text{steady-state commit latency }\ell_s &\leq d_t^s(\delta)+d_c(\delta),\\
\text{eventual worst-case commit latency }\ell_w &\leq d_t(\delta)+d_c(\delta),\\
\text{block time} &\leq \delta+d_p(\delta),\\
k\text{-gap latency} &\leq d_t(\delta)+d_p(\delta)\\
  &\quad{}+k\bigl(\Delta_{\mathrm{to}}+\max\{d_f(\delta),d_t(\delta)\}\bigr).
\end{align*}

Each commit bound is the time for all correct parties to vote for a correct leader's proposal, plus the commit delay \(d_c(\delta)\).
In the good case, voting takes one message delay.
In later views, parties may also need to enter the view and learn that the proposal is safe, which takes at most \(d_t(\delta)\), or \(d_t^s(\delta)\) if the execution is eventually Byzantine-silent and the view is late enough.
Our two signature-free view agreement implementations (\Cref{sec:signature-free-implementations}) and those of Simplex and Minimmit all have \(d_t^s(\delta)=\delta\), so their steady-state and good-case commit latency bounds coincide.
With 1/3-VA, the next leader proposes after \(2\delta\), before the current block commits at \(3\delta\).
For the \(k\)-gap latency, once a correct leader proposes after GST, all correct parties learn that the proposal is safe, prepare it, and enter the next view within \(d_t(\delta)+d_p(\delta)\).
Leaders rotate round-robin, so at most \(k\) consecutive views have Byzantine leaders.
Each such view delays the next correct proposal by at most one view timeout plus \(\max\{d_f(\delta),d_t(\delta)\}\): once all correct parties have entered the view, they all enter the next one within that time (\Cref{cor:consensus-view-duration}).
\Cref{app:consensus-latency-proofs} proves these bounds.

\subsection{Discussion}

The construction above helps explain the design of view agreement: which guarantees later views need, and which outcomes can remain flexible within a view.

Preparation and disabling give later views permission to extend a block or bypass a view, respectively.
Both may occur in the same view, and several blocks may be prepared.
VA-consistency requires uniqueness of the prepared block, and excludes disabling, only when a block commits in the instance.
The signature-free implementations in~\Cref{sec:signature-free-implementations} use the freedom to prepare several values to achieve one-message-delay preparation together with VA-totality.
Disabling does not erase prepared blocks or end their propagation.

VA-totality lets later leaders make progress without forwarding proofs of their observations.
A correct leader may already have a prepared parent and the disabled statuses needed to justify a proposal, while another party has not yet learned those outputs.
Their propagation eventually lets every correct party verify that the proposal is safe.
This is why both prepared values and disabled statuses must propagate, even from views that never commit.

Committed blocks themselves need not propagate.
If a correct party commits a block \(b\), every later committed block extends \(b\), so a party that misses the original VA commitment eventually commits \(b\) as an ancestor of a later committed block.
A block from a disabled view can likewise become an ancestor of a committed block, even though the disabled view cannot commit it directly.

\section{Signature-Free Implementations with Optimal Latency}%
\label{sec:signature-free-implementations}

We obtain our two signature-free protocols by plugging two view agreement implementations into Generic Simplex: 1/3-VA, for \(f<n/3\), and Fast VA, for \(f<n/5\).
We describe 1/3-VA below and outline Fast VA in~\Cref{sec:fast-va}; its full description and proofs appear in~\Cref{app:va-fifth}.
Under maximal resilience, both send \(O(n^2)\) bits per view, assuming fixed-size blocks and message overhead (\Cref{thm:va-third-complexity} and~\Cref{app:va-fifth}).

\Cref{tab:signature-free-consensus-latency} gives the latency bounds of the two view agreement implementations (top) and the consensus latency bounds obtained by substituting them in the formulas of~\Cref{sec:blockchain-va} (bottom).
The bounds on \(\ell_g\), \(\ell_s\), \(\ell_w\), and block time depend only on \(\delta\), so both protocols are optimistically responsive.
The \(k\)-gap latency also depends on \(\Delta_{\mathrm{to}}\), since correct parties detect a silent Byzantine leader only by timing out.
The timeout condition \(\Delta_{\mathrm{to}}\geq2d_t(\Delta)+d_p(\Delta)\) gives view timeouts of \(7\Delta\) and \(5\Delta\), respectively.
Fast VA satisfies VA-strong-unanimity, and so does a variant of 1/3-VA; this reduces these view timeouts to \(6\Delta\) and \(4\Delta\) (\Cref{app:strong-unanimity}).

\begin{table}[ht]
    \centering
    \small
    \caption{Upper bounds for the two signature-free protocols under eventual synchrony: the delays of the asynchronous view agreement implementations (top; see the \hyperref[par:va-latency]{latency definitions in~\Cref*{sec:blockchain-va}}) and the resulting latency metrics of Generic Simplex (bottom).}
    \label{tab:signature-free-consensus-latency}%
    \label{tab:signature-free-va-latency}
    \setlength{\tabcolsep}{4pt}
    \begin{tabular}{@{}lcc@{}}
        \toprule
        View agreement implementation & 1/3-VA & Fast VA \\
        Protocol resilience & \(f<n/3\) & \(f<n/5\) \\
        \midrule
        \multicolumn{3}{@{}l}{\emph{View agreement delays}} \\
        Prepare delay \(d_p(\delta)\) & \(\delta\) & \(\delta\) \\
        Commit delay \(d_c(\delta)\) & \(2\delta\) & \(\delta\) \\
        Fallback delay \(d_f(\delta)\) & \(3\delta\) & \(3\delta\) \\
        Totality delay \(d_t(\delta)\) & \(3\delta\) & \(2\delta\) \\
        Byzantine-silent totality delay \(d_t^s(\delta)\) & \(\delta\) & \(\delta\) \\
        \midrule
        \multicolumn{3}{@{}l}{\emph{Generic Simplex latency}} \\
        Good-case commit latency \(\ell_g\) & \(3\delta\) & \(2\delta\) \\
        Steady-state commit latency \(\ell_s\) & \(3\delta\) & \(2\delta\) \\
        Eventual worst-case commit latency \(\ell_w\) & \(5\delta\) & \(3\delta\) \\
        Block time & \(2\delta\) & \(2\delta\) \\
        View timeout \(\Delta_{\mathrm{to}}\) & \(\geq7\Delta\) & \(\geq4\Delta\) \\
        \(k\)-gap latency & \(4\delta+k(\Delta_{\mathrm{to}}+3\delta)\) & \(3\delta+k(\Delta_{\mathrm{to}}+3\delta)\) \\
        \bottomrule
    \end{tabular}
\end{table}

\subsection{The 1/3-VA Protocol}

The 1/3-VA protocol (\Cref{fig:va-shared,fig:va-third,fig:va-third-amplification}) implements the \hyperref[par:va-guarantees]{view agreement abstraction} of~\Cref{sec:blockchain-va} under asynchrony for \(f<n/3\) (proofs in \Cref{app:va-proofs}).
Its prepare delay of \(\delta\) and commit delay of \(2\delta\) are what give Generic Simplex a block time of \(2\delta\) and the optimal good-case commit latency of \(3\delta\).

\begin{figure*}[!t]
    \centering
    \small
    \setlength{\fboxsep}{2pt}%
    \fbox{%
    \begin{minipage}{\dimexpr\textwidth-2\fboxsep-2\fboxrule\relax}
    \begin{minipage}[t]{0.485\linewidth}
    \begin{algorithmic}[1]
        \State \textbf{Public state:}
        \State \ind $\var{prepared} \gets \emptyset$;
        \State \ind $\var{disabled} \gets \False$;
        \State \ind $\var{committed} \gets \None$;
        \BlankLine
        \State \kw{procedure} $\fn{clear}(a)$:
        \State \ind \kw{if} $a = \bot$: $\var{disabled} \gets \True$;
        \State \ind \kw{else if} $a \notin \var{prepared}$ \kw{and} $\fn{valid}(a)$:
        \State \indd $\var{prepared} \gets \var{prepared} \cup \{a\}$;
        \algstore{vashared}
    \end{algorithmic}
    \end{minipage}\hfill
    \begin{minipage}[t]{0.485\linewidth}
    \begin{algorithmic}[1]
        \algrestore{vashared}
        \State \kw{procedure} $\fn{do\_commit}(x)$:
        \State \ind \kw{if} $\fn{valid}(x)$:
        \State \indd \kw{call} $\fn{clear}(x)$;
        \State \indd \kw{if} $\var{committed} = \None$:
        \State \inddd $\var{committed} \gets x$;
    \end{algorithmic}
    \end{minipage}
    \end{minipage}}
    \caption{State and helper procedures shared by the view agreement implementations.}%
    \label{fig:va-shared}
    \vspace{\floatsep}
    \fbox{%
    \begin{minipage}{\dimexpr\textwidth-2\fboxsep-2\fboxrule\relax}
    \begin{minipage}[t]{0.485\linewidth}
    \begin{algorithmic}[1]
        \State \textbf{Private state:}
        \State \ind $\var{sent\_commit} \gets \None$;
        \State \ind $\var{sent\_candidate} \gets \emptyset$;
        \BlankLine
        \State \kw{upon} $\fn{vote}(x)$:
        \State \ind \kw{broadcast} $\sig{\Vote,x}$;
        \BlankLine
        \State \kw{upon} $\nrecv{\sig{\Vote,x}} \geq n-f$\label{ln:va-third-vote-quorum}
        \State \indd \kw{and} $\fn{valid}(x)$:
        \State \ind \kw{call} $\fn{clear}(x)$;\label{ln:va-third-prepare}
        \State \ind \kw{if} $\var{sent\_commit} = \None$
        \State \inddd \kw{and} $\var{sent\_candidate} \subseteq \{x\}$:\label{ln:va-third-commit-exclusion}
        \State \indd $\var{sent\_commit} \gets x$;
        \State \indd \kw{broadcast} $\sig{\CommitMsg,x}$;\label{ln:va-third-send-commit}
        \BlankLine
        \State \kw{upon} $\nrecv{\sig{\CommitMsg,x}} \geq n-f$:\label{ln:va-third-commit-quorum}
        \State \ind \kw{call} $\fn{do\_commit}(x)$;
        \algstore{vathird}
    \end{algorithmic}
    \end{minipage}\hfill
    \begin{minipage}[t]{0.485\linewidth}
    \begin{algorithmic}[1]
        \algrestore{vathird}
        \State \kw{procedure} $\fn{send\_candidate}(a)$:
        \State \ind \kw{if} $a\notin\var{sent\_candidate}$
        \State \inddd \kw{and} $\var{sent\_commit}\in\{\None,a\}$:\label{ln:va-third-candidate-exclusion}
        \State \indd $\var{sent\_candidate} \gets$
        \State \inddd $\var{sent\_candidate} \cup \{a\}$;
        \State \indd \kw{broadcast} $\sig{\Candidate,a}$;
        \BlankLine
        \State \kw{upon} $\nrecv{\sig{\Vote,x}} \geq f+1$\label{ln:va-third-candidate-votes}
        \State \indd \kw{and} $\fn{valid}(x)$:
        \State \ind \kw{call} $\fn{send\_candidate}(x)$;
        \BlankLine
        \State \kw{upon} $\fn{request\_disable}()$:\label{ln:va-third-request-disable}
        \State \ind \kw{call} $\fn{send\_candidate}(\bot)$;
        \BlankLine
        \State \kw{upon} $\nrecv{\sig{\Candidate,a}} \geq 2f+1$:\label{ln:va-third-candidate-quorum}
        \State \ind \kw{broadcast} $\sig{\Ready,a}$;
    \end{algorithmic}
    \end{minipage}
    \end{minipage}}
    \caption{1/3-VA, part 1: core protocol.
    Thresholds count distinct senders, and vote counts ignore equivocating voters.}%
    \label{fig:va-third}
    \vspace{\floatsep}
    \fbox{%
    \begin{minipage}{\dimexpr\textwidth-2\fboxsep-2\fboxrule\relax}
    \begin{minipage}[t]{0.485\linewidth}
    \begin{algorithmic}[1]
        \State \kw{upon} $\nrecv{\sig{\Ready,a}} \geq f+1$:\label{ln:va-third-relay-ready}
        \State \ind \kw{broadcast} $\sig{\Ready,a}$;
        \algstore{vathirdamplification}
    \end{algorithmic}
    \end{minipage}\hfill
    \begin{minipage}[t]{0.485\linewidth}
    \begin{algorithmic}[1]
        \algrestore{vathirdamplification}
        \State \kw{upon} $\nrecv{\sig{\Ready,a}} \geq 2f+1$:\label{ln:va-third-ready-quorum}
        \State \ind \kw{call} $\fn{clear}(a)$;
    \end{algorithmic}
    \end{minipage}
    \end{minipage}}
    \caption{1/3-VA, part 2: Bracha-style amplification.}%
    \label{fig:va-third-amplification}
\end{figure*}

1/3-VA consists of a core protocol (\Cref{fig:va-third}) and a Bracha-style amplification mechanism~\cite{bracha1987asynchronous} (\Cref{fig:va-third-amplification}), which run concurrently.
Additionally, it uses the shared state and helper procedures of~\Cref{fig:va-shared} to manipulate the view agreement output interface: \(\fn{clear}(x)\) prepares \(x\) if it is valid, or disables the instance if \(x=\bot\), while \(\fn{do\_commit}(x)\) prepares and commits \(x\) if it is valid.
Here we assume that \(\bot\) is a special value that correct parties never vote for.

The algorithm is best understood through three cases, based on the votes of correct parties.
The first is the fast path, which gives 1/3-VA its prepare and commit delays.
The other two distinguish whether the fast path may have prepared a value, which requires at least \(f+1\) correct voters: if so, all correct parties must eventually prepare that value; if not, nothing can have been committed, and the instance can be safely disabled.
Those three cases are tied together by Bracha-style amplification, which propagates preparations and disabling, and a simple local exclusion rule that ensures consistency across cases.
Unless stated otherwise, line references below refer to \Cref{fig:va-third}.

\paragraph{Case 1: All correct parties vote for the same valid value.}
Suppose all correct parties vote for the same valid value \(x\) and none requests disabling.
Then every correct party receives \(n-f\) votes for \(x\), prepares \(x\), and broadcasts a \(\CommitMsg\) message for \(x\) (lines~\ref{ln:va-third-vote-quorum}--\ref{ln:va-third-send-commit}).
Next, every correct party receives \(n-f\) \(\CommitMsg\) messages for \(x\) and commits it (line~\ref{ln:va-third-commit-quorum}).
This achieves 1/3-VA's prepare delay of \(\delta\) and its commit delay of~\(2\delta\).

\paragraph{Case 2: Some value has at least \(f+1\) correct voters.}
Suppose at least \(f+1\) correct parties vote for the same valid value \(x\).
There may not be enough votes to commit \(x\), but these correct voters intersect every set of \(n-f\) voters, so no other value can gather \(n-f\) votes.
Moreover, every correct party eventually receives their votes.
We use this common support to ensure that all correct parties eventually prepare \(x\):

Every correct party broadcasts a \(\Candidate\) message for \(x\) upon receiving \(f+1\) votes for \(x\) (line~\ref{ln:va-third-candidate-votes}).
The local exclusion rule discussed below cannot prevent this, since no correct party can have sent a \(\CommitMsg\) message for a different value.
Upon receiving \(2f+1\) \(\Candidate\) messages for \(x\), parties broadcast a \(\Ready\) message for \(x\) (line~\ref{ln:va-third-candidate-quorum}).
Upon receiving \(2f+1\) \(\Ready\) messages for \(x\), they prepare \(x\) by calling \(\fn{clear}(x)\) (\Cref{fig:va-third-amplification}, line~\ref{ln:va-third-ready-quorum}).

\paragraph{Case 3: No valid value has \(f+1\) correct voters.}
Finally, suppose at most \(f\) correct parties vote for each valid value.
No correct party can send a \(\CommitMsg\) message or prepare any value at line~\ref{ln:va-third-prepare}, since \(n-f\) votes for one value would include at least \(n-2f\geq f+1\) correct voters.
Nevertheless, we must ensure that, if every correct party calls \(\fn{request\_disable}()\), then eventually all correct parties clear the instance.
For this, a party that requests disabling broadcasts a \(\Candidate\) message for \(\bot\) (line~\ref{ln:va-third-request-disable}).
Since no correct party has sent a \(\CommitMsg\) message, the local exclusion rule discussed below does not prevent this.
Thus, if all correct parties request disabling, they all broadcast \(\Candidate\) messages for \(\bot\), and the same exchange of \(\Ready\) messages as in the previous case eventually disables the instance at all correct parties.

\paragraph{Consistency.}
To maintain VA-consistency, correct parties never broadcast \(\CommitMsg\) and \(\Candidate\) for different values, including \(\bot\) (lines~\ref{ln:va-third-commit-exclusion} and~\ref{ln:va-third-candidate-exclusion}).
A set of \(n-f\) \(\CommitMsg\) messages for one value and a set of \(2f+1\) \(\Candidate\) messages for another would have a correct sender in common, contradicting this exclusion.
Moreover, any two sets of \(n-f\) votes intersect in a correct party, so a commitment to \(x\) also excludes preparing a different value directly upon votes.

Notice that several values may each have \(f+1\) correct voters, in which case the candidate and ready exchanges prepare all valid ones.
This is okay because view agreement allows preparing multiple values if none commits. %

\paragraph{Totality.}
If a correct party prepares \(x\) upon receiving \(n-f\) votes, then at least \(f+1\) correct parties voted for \(x\), so all correct parties eventually prepare \(x\) as in the second case.
Preparations and disabling triggered by \(\Ready\) messages propagate through standard Bracha-style amplification, using the \(f+1\) relay rule (\Cref{fig:va-third-amplification}, line~\ref{ln:va-third-relay-ready}).

\paragraph{Note on requests to disable after voting.}
By the exclusion rule, a party that has sent a \(\Candidate\) message for \(\bot\) can no longer send a \(\CommitMsg\) message.
Hence, in the first case, a party that requests disabling after voting but before sending its \(\CommitMsg\) message may prevent the commit, although every correct party still prepares \(x\) upon \(n-f\) votes.
This is why VA-unanimity requires that no party requests disabling, and why the timeout condition of Generic Simplex includes the prepare delay.
\Cref{app:strong-unanimity} presents a variant of 1/3-VA in which such requests cannot prevent the commit, reducing the view timeout to~\(6\Delta\).

\subsection{Fast VA}
\label{sec:fast-va}

Fast VA assumes \(f<n/5\) and commits directly upon \(n-f\) votes, eliminating the separate \(\CommitMsg\) exchange and reducing the blockchain good-case latency to \(2\delta\).
Its fallback uses \(2f+1\) votes to trigger ready messages and Bracha-style amplification to propagate preparations and disabling.
To protect a possible commitment, a voter sends \(\Disable\) only after observing \(2f+1\) parties that voted for other values or sent \(\Disable\), rather than merely upon requesting disabling.
The stronger resilience assumption \(n>5f\) ensures that fallback can still progress under these restrictions.
The full protocol and proofs appear in~\Cref{app:va-fifth}.

\section{Related Work}%
\label{sec:related-work}
\daniel{We probably need a comparison table for these protocols in the main paper or appendix.}

\paragraph{Signature-free consensus.}
Among earlier responsive signature-free protocols with optimal resilience, Castro's MAC-based PBFT commits in \(3\delta\) but sends \(O(n^3)\) bits per view~\cite{castro_thesis}, while IT-HS~\cite{abrahamInformationTheoreticHotStuff2021} and TetraBFT~\cite{yuTetraBFTReducingLatency2024b} send \(O(n^2)\) bits per view but commit in \(6\delta\) and \(5\delta\), respectively.
Forget-IT~\cite{abrahamForgetITOptimalGoodCase2026}, which is closely related, solves single-shot consensus with \(n=3f+1\), \(3\delta\) good-case latency, \(O(n^2)\) communication per view, and constant persistent storage.
Relative to Forget-IT, we add a pipelined blockchain protocol with a block time of \(2\delta\) and commitment independent of later views, obtained from a modular, generic construction.
1/3-VA's core is inspired by Forget-IT and by IT-Kuplex, a signature-free protocol described in a blog post~\cite{abrahamInformationTheoreticKuplex2026}; in particular, its exclusion between commit and candidate messages resembles IT-Kuplex's final lock.
Fast TetraBFT~\cite{fernandez-pintoFastTetraBFTOptimizing2026} is another single-shot protocol deciding in \(3\delta\).

\paragraph{Signature-free blockchain protocols.}
Multi-shot TetraBFT~\cite{yuTetraBFTReducingLatency2024b} proposes a block every \(\delta\), but, as in Chained HotStuff, a block commits only once blocks in the next three slots are notarized.
In the full version of his paper, Shoup sketches signature-free variants of Simplex~\cite[Section~7]{shoupSingSongSimplexEprint} that replace threshold signatures with Bracha's echo/ready exchange~\cite{bracha1987asynchronous}; this doubles their latencies to \(6\delta\) from proposal to commit and \(4\delta\) between consecutive proposals, for a block time of about \(3\delta\) with rotating leaders (about \(2\delta\) with a stable leader).
Simple-IT~\cite{yuSimpleIT2026} notarizes Simplex blocks with reliable broadcast; it commits in \(4\delta\) (\(3\delta\) if about \(5n/6\) parties are correct) and reaches a block time of \(\delta\) with speculative proposals.
Our protocols combine a \(2\delta\) block time, with rotating leaders and no speculation, with commitment independent of later views and optimal good-case latency.

\paragraph{Decompositions of Simplex.}
Abraham decomposes single-shot Simplex into an outer protocol and a per-view Certifying Graded Broadcast (CGB) instance~\cite{abrahamDeconstructingSimplex2026}, and Abraham and Neu apply this decomposition to chained Simplex~\cite{abrahamNeuChainedSimplex2026}.
Generic Simplex keeps the shape of their outer proposal rule but changes how parties establish proposal safety.
CGB receives the leader's input and an external-validity predicate encoding the outer voting rule, which their construction checks using signed certificates carried with the proposal.
With view agreement, the outer protocol at each party waits until its local state shows a proposal to be safe and then votes.
VA-totality ensures that a proposal safe at one correct party eventually becomes safe at all, so proposals need not carry transferable certificates.
Generic Simplex also controls requests to disable, keeping the view timer outside view agreement.
The view agreement guarantees involve no timing assumptions; CGB manages its timer internally and assumes synchrony for its good-case guarantee.
CGB requires totality for all certificates, including decision certificates, whereas view agreement requires it only for preparation and disabling;
Generic Simplex instead obtains commitment totality through later committed descendants.
IT-Kuplex~\cite{abrahamInformationTheoreticKuplex2026} already runs a signature-free instance in each view and, like Generic Simplex, holds a proposal until directly received quorums show it to be safe; it is single-shot, however, and states no modular per-view specification.

\paragraph{Other related work.}
View agreement is also a per-view abstraction in the vein of adopt-commit~\cite{gafni_round-by-round_1998}, crusader agreement~\cite{dolevByzantineGeneralsStrike1982}, graded broadcast (gradecast)~\cite{feldmanOptimalProbabilisticProtocol1997}, and their generalizations~\cite{attiya_multi-valued_2023}.
Unlike these abstractions, view agreement is not a single-input/single-output abstraction.

Casper FFG~\cite{buterinCasperFriendlyFinality2017} is an early example of committing blocks through votes on their descendants, as Chained HotStuff~\cite{yin2019hotstuff} and Streamlet~\cite{chan2020streamlet} also do.
The latter two protocols need four and five consecutive correct leaders, respectively, to guarantee commitment, a requirement that BeeGees~\cite{giridharanBeeGeesStayinAlive2023} and Simplex~\cite{chanSimplexConsensusSimple2023} avoid.
For \(n\geq5f+1\), two-step commitment goes back to FaB Paxos~\cite{martinFastByzantineConsensus2006a} (\(n\geq5f-1\) with signatures~\cite{kuznetsovRevisitingOptimalResilience2021}), and Minimmit~\cite{chouMinimmitFC2026} is a signature-based Simplex-style protocol committing in \(2\delta\).

\clearpage
\bibliographystyle{splncs04}
\bibliography{main}

@article{lamport_time_1978,
  title = {Time, Clocks, and the Ordering of Events in a Distributed System},
  author = {Lamport, Leslie},
  year = {1978},
  journal = {Communications of the ACM},
  volume = {21},
  number = {7},
  pages = {558--565},
  urldate = {2016-11-17}
}

@article{lamport2001paxos,
  title={{Paxos} made simple},
  author={Lamport, Leslie},
  journal={ACM SIGACT News},
  volume={32},
  number={4},
  pages={51--58},
  year={2001}
}

@misc{abrahamInformationTheoreticKuplex2026,
  title = {Information-Theoretic {Kuplex}},
  author = {Abraham, Ittai and Das, Sourav and Efron, Yuval and Komatovic, Jovan and Stern, Gilad},
  year = 2026,
  month = jun,
  howpublished = {Decentralized Thoughts},
  url = {https://decentralizedthoughts.github.io/2026-06-05-IT-Kuplex/},
  note = {Blog post, 5 June 2026}
}

@inproceedings{abrahamForgetITOptimalGoodCase2026,
  title = {Forget-{{IT}}: {{Optimal Good-Case Latency For Information-Theoretic BFT}}},
  shorttitle = {Forget-{{IT}}},
  booktitle = {Proceedings of the {ACM} Symposium on Principles of Distributed Computing},
  author = {Abraham, Ittai and Das, Sourav and Efron, Yuval and Komatovic, Jovan},
  year = 2026,
  month = jul,
  series = {{PODC} '26},
  pages = {478--488},
  publisher = {ACM},
  url = {https://doi.org/10.1145/3796701.3815965},
  urldate = {2026-09-09}
}

@inproceedings{yin2019hotstuff,
  title={{HotStuff}: {BFT} consensus with linearity and responsiveness},
  author={Yin, Maofan and Malkhi, Dahlia and Reiter, Michael K and Gueta, Guy Golan and Abraham, Ittai},
  booktitle={Proceedings of the 2019 ACM Symposium on Principles of Distributed Computing},
  pages={347--356},
  year={2019}
}

@inproceedings{castro1999practical,
  title={Practical {Byzantine} fault tolerance},
  author={Castro, Miguel and Liskov, Barbara},
  booktitle={3rd Symposium on Operating Systems Design and Implementation ({OSDI}~'99)},
  publisher={USENIX Association},
  pages={173--186},
  year={1999}
}

@inproceedings{yuTetraBFTReducingLatency2024b,
  title = {{{TetraBFT}}: {{Reducing Latency}} of {{Unauthenticated}}, {{Responsive BFT Consensus}}},
  shorttitle = {{{TetraBFT}}},
  booktitle = {Proceedings of the 43rd {{ACM Symposium}} on {{Principles}} of {{Distributed Computing}}},
  author = {Yu, Qianyu and Losa, Giuliano and Wang, Xuechao},
  year = 2024,
  month = jun,
  series = {{{PODC}} '24},
  pages = {257--267},
  publisher = {Association for Computing Machinery},
  address = {New York, NY, USA},
  doi = {10.1145/3662158.3662783},
  urldate = {2025-12-10},
  isbn = {979-8-4007-0668-4}
}

@inproceedings{abraham2022good,
  title={Good-Case and Bad-Case Latency of Unauthenticated {Byzantine} Broadcast: {A} Complete Categorization},
  author={Abraham, Ittai and Ren, Ling and Xiang, Zhuolun},
  booktitle={25th International Conference on Principles of Distributed Systems ({OPODIS} 2021)},
  series={Leibniz International Proceedings in Informatics (LIPIcs)},
  volume={217},
  pages={5:1--5:20},
  publisher={Schloss Dagstuhl -- Leibniz-Zentrum f{\"u}r Informatik},
  year={2022},
  doi={10.4230/LIPIcs.OPODIS.2021.5}
}

@inproceedings{kuznetsovRevisitingOptimalResilience2021,
  title = {Revisiting {{Optimal Resilience}} of {{Fast Byzantine Consensus}}},
  booktitle = {Proceedings of the 2021 {{ACM Symposium}} on {{Principles}} of {{Distributed Computing}}},
  author = {Kuznetsov, Petr and Tonkikh, Andrei and Zhang, Yan X},
  year = {2021},
  month = jul,
  series = {{{PODC}}'21},
  pages = {343--353},
  publisher = {Association for Computing Machinery},
  address = {New York, NY, USA},
  doi = {10.1145/3465084.3467924},
  urldate = {2022-11-02},
  isbn = {978-1-4503-8548-0}
}

@article{martinFastByzantineConsensus2006a,
  title = {Fast {{Byzantine Consensus}}},
  author = {Martin, J.-P. and Alvisi, L.},
  year = {2006},
  month = jul,
  journal = {IEEE Transactions on Dependable and Secure Computing},
  volume = {3},
  number = {3},
  pages = {202--215},
  issn = {1941-0018},
  doi = {10.1109/TDSC.2006.35},
  urldate = {2025-03-14}
}

@inproceedings{chan2020streamlet,
  title={Streamlet: Textbook streamlined blockchains},
  author={Chan, Benjamin Y and Shi, Elaine},
  booktitle={Proceedings of the 2nd ACM Conference on Advances in Financial Technologies},
  pages={1--11},
  year={2020}
}

@article{bracha1987asynchronous,
  title={Asynchronous {Byzantine} agreement protocols},
  author={Bracha, Gabriel},
  journal={Information and Computation},
  volume={75},
  number={2},
  pages={130--143},
  year={1987},
  publisher={Elsevier}
}

@inproceedings{abrahamInformationTheoreticHotStuff2021,
  title = {Information {{Theoretic HotStuff}}},
  booktitle = {24th {{International Conference}} on {{Principles}} of {{Distributed Systems}} ({{OPODIS}} 2020)},
  author = {Abraham, Ittai and Stern, Gilad},
  editor = {Bramas, Quentin and Oshman, Rotem and Romano, Paolo},
  year = 2021,
  series = {Leibniz {{International Proceedings}} in {{Informatics}} ({{LIPIcs}})},
  volume = {184},
  pages = {11:1--11:16},
  publisher = {Schloss Dagstuhl -- Leibniz-Zentrum f\"ur Informatik},
  address = {Dagstuhl, Germany},
  issn = {1868-8969},
  doi = {10.4230/LIPIcs.OPODIS.2020.11},
  urldate = {2026-05-28},
  isbn = {978-3-95977-176-4}
}

@phdthesis{castro_thesis,
	type = {Ph.{D}. thesis},
	title = {Practical {Byzantine} {Fault} {Tolerance}},
	school = {MIT},
	author = {Castro, Miguel},
	month = jan,
	year = {2001},
        howpublished={\url{https://www.microsoft.com/en-us/research/wp-content/uploads/2017/01/thesis-mcastro.pdf}}
}

@inproceedings{abraham_good-case_2021,
	address = {New York, NY, USA},
	series = {{PODC}'21},
	title = {Good-case {Latency} of {Byzantine} {Broadcast}: {A} {Complete} {Categorization}},
	isbn = {978-1-4503-8548-0},
	shorttitle = {Good-case {Latency} of {Byzantine} {Broadcast}},
	url = {https://doi.org/10.1145/3465084.3467899},
	urldate = {2024-02-11},
	booktitle = {Proceedings of the 2021 {ACM} {Symposium} on {Principles} of {Distributed} {Computing}},
	publisher = {Association for Computing Machinery},
	author = {Abraham, Ittai and Nayak, Kartik and Ren, Ling and Xiang, Zhuolun},
	month = jul,
	year = {2021},
	pages = {331--341},
}

@inproceedings{attiya_multi-valued_2023,
	title = {Multi-{Valued} {Connected} {Consensus}: {A} {New} {Perspective} on {Crusader} {Agreement} and {Adopt}-{Commit}},
	shorttitle = {Multi-{Valued} {Connected} {Consensus}},
	booktitle = {27th International Conference on Principles of Distributed Systems ({OPODIS} 2023)},
	author = {Attiya, Hagit and Welch, Jennifer L.},
	year = {2024},
	series = {Leibniz International Proceedings in Informatics ({LIPIcs})},
	volume = {286},
	pages = {6:1--6:23},
	publisher = {Schloss Dagstuhl -- Leibniz-Zentrum f\"ur Informatik},
	doi = {10.4230/LIPIcs.OPODIS.2023.6},
}

@inproceedings{chanSimplexConsensusSimple2023,
  title = {Simplex {{Consensus}}: {{A Simple}} and~{{Fast Consensus Protocol}}},
  shorttitle = {Simplex {{Consensus}}},
  booktitle = {Theory of {{Cryptography}}},
  author = {Chan, Benjamin Y. and Pass, Rafael},
  editor = {Rothblum, Guy and Wee, Hoeteck},
  year = 2023,
  pages = {452--479},
  publisher = {Springer Nature Switzerland},
  address = {Cham},
  doi = {10.1007/978-3-031-48624-1_17},
  isbn = {978-3-031-48624-1},
  langid = {english}
}

@misc{shoupSingSongSimplexEprint,
  author = {Shoup, Victor},
  title = {Sing a Song of {Simplex}},
  howpublished = {Cryptology {ePrint} Archive, Paper 2023/1916, \url{https://eprint.iacr.org/2023/1916}},
  year = {2023},
  note = {Full version of~\cite{shoupSingSongSimplex2024}},
}

@inproceedings{gafni_round-by-round_1998,
  title = {Round-by-Round {{Fault Detectors}} ({{Extended Abstract}}): {{Unifying Synchrony}} and {{Asynchrony}}},
  shorttitle = {Round-by-Round {{Fault Detectors}} ({{Extended Abstract}})},
  booktitle = {Proceedings of the {{Seventeenth Annual ACM Symposium}} on {{Principles}} of {{Distributed Computing}}},
  author = {Gafni, Eli},
  year = 1998,
  series = {{{PODC}} '98},
  pages = {143--152},
  publisher = {ACM},
  address = {New York, NY, USA},
  doi = {10.1145/277697.277724},
  urldate = {2016-11-16},
  isbn = {978-0-89791-977-7}
}

@misc{chouMinimmitFastFinality2026,
  title = {Minimmit: {{Fast Finality}} with {{Even Faster Blocks}}},
  shorttitle = {Minimmit},
  author = {Chou, Brendan Kobayashi and {Lewis-Pye}, Andrew and O'Grady, Patrick},
  year = 2026,
  month = jan,
  number = {arXiv:2508.10862},
  eprint = {2508.10862},
  primaryclass = {cs.DC},
  publisher = {arXiv},
  doi = {10.48550/arXiv.2508.10862},
  urldate = {2026-05-18},
  archiveprefix = {arXiv},
  note = {Full version of~\cite{chouMinimmitFC2026}}
}

@inproceedings{chouMinimmitFC2026,
  title = {Minimmit: {{Fast Finality}} with {{Even Faster Blocks}}},
  author = {Chou, Brendan Kobayashi and {Lewis-Pye}, Andrew and O'Grady, Patrick},
  booktitle = {Financial Cryptography and Data Security ({FC} 2026)},
  series = {Lecture Notes in Computer Science},
  publisher = {Springer},
  year = {2026}
}

@misc{yuSimpleIT2026,
  title = {{Simple-IT}: Practical Low-Latency Signature-Free {BFT} Consensus},
  author = {Yu, Qianyu and Villacis, Juan and Losa, Giuliano and Xiang, Zhuolun and Wang, Xuechao},
  year = {2026},
  eprint = {2606.14404},
  archiveprefix = {arXiv},
  primaryclass = {cs.DC},
  url = {https://arxiv.org/abs/2606.14404}
}

@misc{buterinCasperFriendlyFinality2017,
  title = {{Casper} the {Friendly} {Finality} {Gadget}},
  author = {Buterin, Vitalik and Griffith, Virgil},
  year = {2017},
  eprint = {1710.09437},
  archiveprefix = {arXiv},
  primaryclass = {cs.CR},
  doi = {10.48550/arXiv.1710.09437}
}

@misc{abrahamDeconstructingSimplex2026,
  title = {Deconstructing {Simplex}},
  author = {Abraham, Ittai},
  year = {2026},
  month = apr,
  howpublished = {Decentralized Thoughts},
  note = {Blog post, 23 April 2026},
  url = {https://decentralizedthoughts.github.io/2026-04-23-deconstructing-simplex/}
}

@misc{abrahamNeuChainedSimplex2026,
  title = {From Single-Shot {Simplex} to Chained {Simplex}},
  author = {Abraham, Ittai and Neu, Joachim},
  year = {2026},
  month = may,
  howpublished = {Decentralized Thoughts},
  note = {Blog post, 20 May 2026},
  url = {https://decentralizedthoughts.github.io/2026-05-20-from-single-shot-to-chained/}
}

@techreport{aaronsonQuantumComputingBlockchain2026,
  title = {Quantum Computing \& Blockchain},
  author = {Aaronson, Scott and Boneh, Dan and Drake, Justin and Kannan, Sreeram and Lindell, Yehuda and Malkhi, Dahlia},
  institution = {Coinbase Independent Advisory Board on Quantum Computing and Blockchain},
  type = {Position paper},
  year = {2026},
  month = apr,
  note = {21 April 2026},
  url = {https://assets.ctfassets.net/sygt3q11s4a9/6EjYavuGdtJDYCqaJrASj9/9f464a8bf26f44bd6c85710fe7e4a29f/Quantum_Computing_and_Blockchain_v10.3_15April2026.pdf}
}

@inproceedings{giridharanBeeGeesStayinAlive2023,
  title = {{BeeGees}: Stayin' Alive in Chained {BFT}},
  booktitle = {Proceedings of the 2023 {ACM} Symposium on Principles of Distributed Computing},
  author = {Giridharan, Neil and Suri-Payer, Florian and Ding, Matthew and Howard, Heidi and Abraham, Ittai and Crooks, Natacha},
  year = {2023},
  series = {{PODC} '23},
  pages = {233--243},
  publisher = {ACM},
  doi = {10.1145/3583668.3594572}
}

@misc{fernandez-pintoFastTetraBFTOptimizing2026,
  title = {Fast {TetraBFT}: Optimizing Latency Where It Matters},
  author = {Fern{\'a}ndez-Pinto, Antonio J. and Bravo, Manuel and Chockler, Gregory and Gotsman, Alexey},
  year = {2026},
  eprint = {2606.03754},
  archiveprefix = {arXiv},
  primaryclass = {cs.DC},
  url = {https://arxiv.org/abs/2606.03754}
}

@article{dolevByzantineGeneralsStrike1982,
  title = {The {Byzantine} Generals Strike Again},
  author = {Dolev, Danny},
  journal = {Journal of Algorithms},
  year = {1982},
  volume = {3},
  number = {1},
  pages = {14--30},
  doi = {10.1016/0196-6774(82)90004-9}
}

@article{feldmanOptimalProbabilisticProtocol1997,
  title = {An Optimal Probabilistic Protocol for Synchronous {Byzantine} Agreement},
  author = {Feldman, Pesech and Micali, Silvio},
  journal = {SIAM Journal on Computing},
  year = {1997},
  volume = {26},
  number = {4},
  pages = {873--933},
  doi = {10.1137/S0097539790187084}
}

\clearpage
\appendix
\crefalias{section}{appendix}
\crefalias{subsection}{subappendix}
\crefalias{subsubsection}{subsubappendix}
\counterwithin*{theorem}{section}
\renewcommand{\thetheorem}{\thesection.\arabic{theorem}}
\renewcommand{\thelemma}{\thetheorem}
\renewcommand{\thecorollary}{\thetheorem}
\renewcommand{\thedefinition}{\thetheorem}
\section{Extended Model}
\label{app:extended-model}

The main text assumes eventual synchrony and authenticated channels without transferable signatures.
Some results in the appendix use the following extensions.

\paragraph{PKI setting.}
In the PKI setting, all parties have registered public keys and can create publicly verifiable authenticated messages.
We assume perfect authentication, with no probability of error.
Unlike channel authentication, a signature remains verifiable when the signed message is forwarded.

\paragraph{Synchronous executions.}
In a synchronous execution, all correct parties start at time \(0\), GST is \(0\), and correct-to-correct messages have delay at most \(\Delta\) throughout the execution.
The parties' local clocks indicate real time.
As in the main model, \(\delta\leq\Delta\) denotes the maximum message delay during a period whose latency we analyze.

\section{Generic Simplex: Correctness Proofs}%
\label{app:consensus-proofs}

We prove the properties of Generic Simplex, presented in~\Cref{sec:blockchain-va}, using its \hyperref[par:va-guarantees]{view agreement guarantees}, with its latency parameters read as in~\Cref{app:consensus-setting}.
Line numbers refer to~\Cref{fig:consensus-va}.

\subsection{Setting}
\label{app:consensus-setting}

\paragraph{Conventions.}
All correct parties start at time~\(0\) by entering view~\(1\).
Local steps take no time, a \kw{when} rule fires as soon as its guard holds, and only finitely many events occur in any bounded time interval.
Local clocks advance at the rate of real time, so a view timer reset at time \(s\) fires at time \(s+\Delta_{\mathrm{to}}\) unless it is reset earlier.
Leaders rotate round-robin: the leader of view \(v\) is \(\party{((v-1)\bmod n)+1}\).
We use only two consequences of this rule: every party leads infinitely many views, and if at most \(k<n\) parties are Byzantine, then at most \(k\) consecutive views have Byzantine leaders.

\paragraph{Timing.}
An execution is \emph{\(\delta\)-bounded from time \(T_0\)} if every message that a correct party sends to a correct party at a time \(s\) is received by time \(\max\{s,T_0\}+\delta\).
We read eventual synchrony as saying that every execution is \(\Delta\)-bounded from \(\mathrm{GST}\).
The latency bounds concern executions that are \(\delta\)-bounded from some \(T_0\geq\mathrm{GST}\), where \(\delta\leq\Delta\); totality and liveness use \(\delta=\Delta\) and \(T_0=\mathrm{GST}\).

We read the latency parameters of view agreement as follows, for an instance in an execution that is \(\delta\)-bounded from \(T_0\) and in which correct parties respect the interface.
\begin{itemize}
    \item If every correct party votes for the same valid value \(x\), no correct party calls \(\fn{request\_disable}()\) (or, if the implementation satisfies VA-strong-unanimity (\Cref{def:va-strong-unanimity}), none does so before voting), and the last correct vote occurs at time \(s\), then every correct party prepares \(x\) by time \(\max\{s,T_0\}+d_p(\delta)\) and commits \(x\) by time \(\max\{s,T_0\}+d_c(\delta)\).
    \item If a correct party prepares \(x\), or disables, at time \(s\), then every correct party prepares \(x\), or disables, respectively, by time \(\max\{s,T_0\}+d_t(\delta)\), and by time \(\max\{s,T_0\}+d_t^s(\delta)\) if the instance is Byzantine-silent.
    \item If every correct party calls \(\fn{request\_disable}()\), the last one at time \(s\), then the instance is cleared at every correct party by time \(\max\{s,T_0\}+d_f(\delta)\).
\end{itemize}
An execution that is \(\delta\)-bounded from \(T_0\) is also \(\delta'\)-bounded from \(T_0\) for every \(\delta'\geq\delta\), so the least valid bounds are nondecreasing in \(\delta\), and we assume that every parameter is.
As in the \hyperref[par:va-latency]{latency definitions in~\Cref*{sec:blockchain-va}}, every parameter is at least \(\delta\), and \(d_t^s(\delta)\leq d_t(\delta)\).
Finally, the guarantees and delay bounds of view agreement hold for every pattern and timing of input calls that respects the interface.

\paragraph{Notation.}
For a correct party \(p\) and a view \(v\), \(\fn{prepared}_p(v)\), \(\fn{disabled}_p(v)\), and \(\fn{committed}_p(v)\) denote the outputs of view \(v\) at \(p\), as defined in~\Cref{fig:consensus-va}; thus \(\fn{prepared}_p(0)=\{\text{genesis}\}\), \(\fn{committed}_p(0)=\text{genesis}\), and view~\(0\) is never disabled.
View \(v\) is \emph{cleared} at \(p\) when \(\fn{prepared}_p(v)\neq\emptyset\) or \(\fn{disabled}_p(v)\) holds.
A block is \emph{prepared} (\emph{committed}) \emph{in view \(v\)}, and view \(v\) is \emph{disabled}, if this holds at some correct party at some time.
A block \(b\) is \emph{safe for view \(v\) at \(p\)} if \(b\) is not genesis and there is a view \(v'<v\) such that \(\fn{parent}(b)\in\fn{prepared}_p(v')\) and \(\fn{disabled}_p(v'')\) holds for every \(v''\) with \(v'<v''<v\).
This is \(\fn{is\_safe}(b)\) evaluated with \(\var{curr\_view}=v\); since outputs only grow, it remains true once true.
By VA-totality, a block that is safe for view \(v\) at some correct party eventually becomes safe for \(v\) at every correct party.

We write \(e_p(v)\) for the time at which \(p\) enters view \(v\), and \(t_f(v)\) for the first time at which a correct party enters view \(v\).
If the leader of view \(v\) is correct, \(b_v\) denotes the block that it creates in view \(v\).

\begin{lemma}
\label{lem:consensus-basic}
Every correct party \(p\) satisfies the following.
\begin{enumerate}[label=(\alph*)]
    \item Party \(p\) enters views \(1,2,3,\ldots\) in order, each at most once, and its current view is always the smallest view that is not cleared at \(p\).
    Hence \(p\) has entered view \(v\) by time \(t\) if and only if every view below \(v\) is cleared at \(p\) by time \(t\).
    \item Party \(p\) calls the instance of view \(v\) only while in view \(v\), and it respects the interface of the instance.
    \item Party \(p\) calls \(\fn{request\_disable}()\) on the instance of view \(v\) if and only if it is still in view \(v\) at time \(e_p(v)+\Delta_{\mathrm{to}}\), and then does so at that time.
    Hence no correct party calls \(\fn{request\_disable}()\) on this instance before time \(t_f(v)+\Delta_{\mathrm{to}}\).
    \item If \(p\) leads view \(v\), then, upon entering \(v\), it creates \(b_v\) by its only call to \(\fn{extend}\) in view \(v\) and broadcasts it.
    The parent of \(b_v\) is in \(\fn{prepared}_p(\var{vmax})\), where \(\var{vmax}\) is the largest view below \(v\) with a nonempty prepared set at \(p\) at time \(e_p(v)\), and every view strictly between \(\var{vmax}\) and \(v\) is disabled at \(p\) at that time.
    The block \(b_v\) is valid, and it is safe for view \(v\) at \(p\) from time \(e_p(v)\) on.
    \item A vote by \(p\) in view \(v\) through the vote rule (line~\ref{ln:consensus-vote}) is for \(\var{proposal}[v]\), which is then safe for view \(v\) at \(p\); a catch-up vote (line~\ref{ln:consensus-catchup-vote}) is for a block in \(\fn{prepared}_p(v)\).
    If the leader of \(v\) is correct, then \(\var{proposal}[v]\) is \(\None\) or \(b_v\).
\end{enumerate}
\end{lemma}

\begin{proof}
(a)
Party \(p\) enters view \(1\) at startup and then only enters the next view, through the advance rule (line~\ref{ln:consensus-advance}), which fires as soon as the current view is cleared; view \(0\) is always cleared, and outputs only grow.

(b)
Every call goes to \(\var{va}[\var{curr\_view}]\), \(p\) never returns to a view, and the flags \(\var{voted}\) and \(\var{disable\_requested}\) are reset only upon entering a view.

(c)
Party \(p\) resets its view timer upon entering each view (line~\ref{ln:consensus-reset-timer}) and calls \(\fn{request\_disable}()\) only when the timer fires, on the instance of its current view; \(\var{disable\_requested}\) is reset only upon entering a view, and \(e_p(v)\geq t_f(v)\).

(d)
Procedure \(\fn{propose}\) is called only upon entering a view that \(p\) leads, \(\var{vmax}\) exists since \(\fn{prepared}_p(0)\neq\emptyset\), and the parent is valid by VA-validity or because genesis is valid.
By (a), every view strictly between \(\var{vmax}\) and \(v\) is cleared at \(p\) at time \(e_p(v)\), and its prepared set is empty, so it is disabled.

(e)
The variable \(\var{proposal}[v]\) holds the first proposal for view \(v\) received from its leader.
\end{proof}

\subsection{Safety}
\label{app:consensus-safety}

\begin{lemma}
\label{lem:consensus-justification}
If a block \(b\) is prepared or committed in a view \(v>0\), then \(b\) is valid and there is a view \(v'<v\) such that \(\fn{parent}(b)\) is prepared in view \(v'\) and every view strictly between \(v'\) and \(v\) is disabled.
\end{lemma}

\begin{proof}
By VA-validity, \(b\) is valid and some correct party voted for \(b\) in view \(v\).
The first such vote is not a catch-up vote, since a catch-up vote is for a block that is already prepared and hence, by VA-validity, already voted for by a correct party.
By~\Cref{lem:consensus-basic}(e), \(b\) was therefore safe for view \(v\) at the voter, which provides the view \(v'\).
\end{proof}

\begin{lemma}
\label{lem:consensus-committed-prefix}
If \(b\) is prepared or committed in view \(v\), and \(b'\) is committed in view \(v'\leq v\), then \(b'\) is an ancestor of \(b\).
\end{lemma}

\begin{proof}
We fix \(v'\) and \(b'\), and proceed by strong induction on \(v\).
If \(v=v'\), then \(b=b'\): for \(v=0\), both are genesis, and for \(v>0\), this follows from VA-consistency.
Suppose that \(v>v'\).
By~\Cref{lem:consensus-justification}, there is a view \(v_{prev}<v\) such that \(\fn{parent}(b)\) is prepared in view \(v_{prev}\) and every view strictly between \(v_{prev}\) and \(v\) is disabled.
View \(v'\) is not disabled: view \(0\) never is, and for \(v'>0\), VA-consistency excludes it because \(b'\) is committed in \(v'\).
Hence \(v'\) is not strictly between \(v_{prev}\) and \(v\), that is, \(v_{prev}\geq v'\).
By the induction hypothesis, \(b'\) is an ancestor of \(\fn{parent}(b)\), and hence of \(b\).
\end{proof}

\begin{theorem}[Consistency]
\label{thm:blockchain-consistency}
Correct parties never commit incompatible blocks.
\end{theorem}

\begin{proof}
A correct party calls \(\fn{commit}(c)\) only when \(c\) is committed in some view at that party (line~\ref{ln:consensus-commit}).
Thus, for two blocks \(b\) and \(b'\) committed by correct parties, possibly as ancestors, there are blocks \(c\) and \(c'\) committed in views \(v\) and \(v'\), respectively, such that \(b\) is an ancestor of \(c\) and \(b'\) is an ancestor of \(c'\).
Without loss of generality, let \(v\leq v'\).
By~\Cref{lem:consensus-committed-prefix}, \(c\) is an ancestor of \(c'\).
Thus \(b\) and \(b'\) lie on the same chain ending at \(c'\), so they are compatible.
\end{proof}

\begin{theorem}[Validity]
\label{thm:blockchain-validity}
Correct parties never commit invalid blocks.
\end{theorem}

\begin{proof}
By strong induction on \(v\), every ancestor of a block prepared or committed in view \(v\) is valid: for \(v=0\), the block is genesis, which is valid, and for \(v>0\), \Cref{lem:consensus-justification} shows that the block is valid and that its parent is prepared in a lower view.
A correct party calls \(\fn{commit}(c)\) only for a block \(c\) committed in some view, so every block that it commits, \(c\) or an ancestor of \(c\), is valid.
\end{proof}

\subsection{Views}
\label{app:consensus-views}

\begin{lemma}
\label{lem:consensus-view-progress}
Every correct party enters every view.
\end{lemma}

\begin{proof}
All correct parties enter view~\(1\).
Suppose that all correct parties enter view \(v\); by~\Cref{lem:consensus-basic}(a), it suffices to show that \(v\) is eventually cleared at every correct party.
If \(v\) is cleared at some correct party, then VA-totality clears it at every correct party.
Otherwise, no correct party ever leaves \(v\), so, by~\Cref{lem:consensus-basic}(c), every correct party calls \(\fn{request\_disable}()\) on the instance of view \(v\), and VA-fallback-progress clears \(v\) at every correct party, a contradiction.
\end{proof}

Hence \(e_p(v)\) and \(t_f(v)\) are defined for every correct party \(p\) and every view \(v\).

For the rest of this appendix, we consider an execution that is \(\delta\)-bounded from some \(T_0\geq\mathrm{GST}\), where \(\delta\leq\Delta\).

\begin{lemma}
\label{lem:consensus-view-synchronization}
If a correct party enters view \(v\) at time \(t\), then every correct party enters view \(v\) by time \(\max\{t,T_0\}+d_t(\delta)\).
In particular, if \(t_f(v)\geq T_0\), then every correct party enters view \(v\) by time \(t_f(v)+d_t(\delta)\).
\end{lemma}

\begin{proof}
Let \(p\) be a correct party that enters view \(v\) at time \(t\).
By~\Cref{lem:consensus-basic}(a), every view \(v'<v\) is cleared at \(p\) by time \(t\): \(p\) has prepared some block or disabled in view \(v'\).
By the totality delay, every correct party prepares that block, or disables, in view \(v'\) by time \(\max\{t,T_0\}+d_t(\delta)\).
Thus every view below \(v\) is cleared at every correct party by that time, and~\Cref{lem:consensus-basic}(a) gives the claim.
\end{proof}

\begin{corollary}
\label{cor:consensus-view-duration}
If every correct party has entered view \(v\) by time \(s\), and \(s+\Delta_{\mathrm{to}}\geq T_0\), then every correct party enters view \(v+1\) by time
\[
  s+\Delta_{\mathrm{to}}+\max\{d_f(\delta),d_t(\delta)\}.
\]
\end{corollary}

\begin{proof}
Let \(T=s+\Delta_{\mathrm{to}}\).
If view \(v\) is cleared at some correct party by time \(T\), then the totality delay makes it cleared at every correct party by time \(T+d_t(\delta)\), since \(T\geq T_0\).
Otherwise, no correct party leaves \(v\) by time \(T\), so, by~\Cref{lem:consensus-basic}(c), every correct party \(p\) calls \(\fn{request\_disable}()\) on the instance of view \(v\) at time \(e_p(v)+\Delta_{\mathrm{to}}\leq T\).
The fallback delay then makes \(v\) cleared at every correct party by time \(T+d_f(\delta)\).
In both cases, \Cref{lem:consensus-basic}(a) gives the claim.
\end{proof}

\begin{lemma}
\label{lem:consensus-proposal-propagation}
If the leader \(L\) of view \(v\) is correct, then, by time \(\max\{e_L(v),T_0\}+d_t(\delta)\), every correct party has entered view \(v\) and received \(b_v\), and \(b_v\) is safe for view \(v\) at every correct party.
\end{lemma}

\begin{proof}
Every correct party enters view \(v\) by that time by~\Cref{lem:consensus-view-synchronization}.
It also receives \(b_v\), which \(L\) broadcasts at time \(e_L(v)\), by time \(\max\{e_L(v),T_0\}+\delta\), and \(\delta\leq d_t(\delta)\).
By~\Cref{lem:consensus-basic}(d), at time \(e_L(v)\), there is a view \(v'<v\) such that \(\fn{parent}(b_v)\in\fn{prepared}_L(v')\) and \(\fn{disabled}_L(v'')\) holds for every \(v''\) with \(v'<v''<v\).
By the totality delay, these outputs hold at every correct party by time \(\max\{e_L(v),T_0\}+d_t(\delta)\), so \(b_v\) is then safe for view \(v\) at every correct party.
\end{proof}

\subsection{Views with Correct Leaders}
\label{app:consensus-good-views}

We first apply VA-unanimity and the prepare delay to execution prefixes that satisfy their conditions so far, even if the rest of the execution might not.

\begin{lemma}
\label{lem:va-continuation}
Consider a view agreement instance and a finite prefix \(\pi\) of the execution in which no correct party calls \(\fn{request\_disable}()\) on the instance and every vote of a correct party is for the same valid value \(x\).
\begin{enumerate}
    \item No correct party disables in~\(\pi\).
    \item If every correct party votes in \(\pi\) by time \(s\), and \(\pi\) extends through time \(\max\{s,T_0\}+d_p(\delta)\), then every correct party prepares \(x\) in~\(\pi\).
\end{enumerate}
\end{lemma}

\begin{proof}
Since the guarantees of view agreement hold for every pattern of input calls that respects its interface, we may extend \(\pi\) to an execution \(\sigma\) of the instance, still \(\delta\)-bounded from \(T_0\), in which every correct party that has not voted in \(\pi\) calls \(\fn{vote}(x)\) and no correct party ever calls \(\fn{request\_disable}()\).
In \(\sigma\), each correct party votes at most once and never after requesting disabling, as the interface requires.
Moreover, all correct parties vote for the valid value \(x\) and none requests disabling, so, by VA-unanimity, every correct party eventually commits \(x\) in \(\sigma\).
By VA-consistency, no correct party ever disables in \(\sigma\), and in particular not in its prefix~\(\pi\).
For part 2, the last correct vote in \(\sigma\) occurs by time \(s\), so the prepare delay makes every correct party prepare \(x\) by time \(\max\{s,T_0\}+d_p(\delta)\); these preparations occur in \(\pi\), which extends through that time.
\end{proof}

From now on, we assume the \emph{timeout condition}
\begin{equation}
\label{eq:consensus-timeout}
\Delta_{\mathrm{to}}\geq
\begin{cases}
2d_t(\Delta) & \text{with VA-strong-unanimity},\\
2d_t(\Delta)+d_p(\Delta) & \text{otherwise},
\end{cases}
\end{equation}
where VA-strong-unanimity is defined in~\Cref{app:strong-unanimity}.
This condition in particular requires \(d_t(\Delta)\) to be finite, and \(d_p(\Delta)\) as well in the second case.
The following lemma is the core of the progress and latency arguments.

\begin{lemma}
\label{lem:consensus-good-view}
Let view \(v\) have a correct leader, with \(t_f(v)\geq T_0\), and let \(t_1\leq t_f(v)+2d_t(\delta)\) be a time by which every correct party has entered view \(v\) and received \(b_v\), and \(b_v\) is safe for view \(v\) at every correct party.
Then:
\begin{enumerate}[label=(\alph*)]
    \item every correct vote in view \(v\) is for \(b_v\), and every correct party votes in view \(v\) by time \(t_1\);
    \item every correct party prepares \(b_v\), and enters view \(v+1\), by time \(t_1+d_p(\delta)\);
    \item no correct party calls \(\fn{request\_disable}()\) on the instance of view \(v\) before voting in it, and view \(v\) is never disabled;
    \item every correct party commits \(b_v\), and does so by time \(t_1+d_c(\delta)\).
\end{enumerate}
\end{lemma}

\begin{proof}
Let \(b=b_v\), which is valid by~\Cref{lem:consensus-basic}(d).
Since every correct party has entered \(v\) by time \(t_1\), we have \(t_1\geq t_f(v)\geq T_0\).

First, every correct vote in view \(v\) is for \(b\).
Otherwise, consider the first correct vote in view \(v\) for a block \(b''\neq b\).
By~\Cref{lem:consensus-basic}(e), it is a catch-up vote for a block prepared in view \(v\), so, by VA-validity, a correct party voted for \(b''\) in view \(v\) earlier, a contradiction.

Second, by~\Cref{lem:consensus-basic}(c), the timeout condition, and the monotonicity of the delay parameters, no correct party calls \(\fn{request\_disable}()\) on the instance of view \(v\) before time
\[
  t_f(v)+\Delta_{\mathrm{to}}\geq t_f(v)+2d_t(\delta)\geq t_1.
\]
Let \(\pi\) be the prefix of the execution through time \(t_1\), excluding timer firings at that time.
By part 1 of~\Cref{lem:va-continuation}, no correct party disables view \(v\) in~\(\pi\).

Third, every correct party \(p\) votes for \(b\) by time \(t_1\), which completes (a).
If \(p\) is in view \(v\) at time \(t_1\) and has not voted, then the guard of the vote rule holds, since \(\var{proposal}[v]=b\) is safe for view \(v\) at \(p\) and \(p\) has not requested disabling.
If \(p\) has left view \(v\) by time \(t_1\), then \(v\) was cleared at \(p\) in \(\pi\), hence not by disabling, so \(\fn{prepared}_p(v)\neq\emptyset\); by VA-validity and the first step, \(\fn{prepared}_p(v)=\{b\}\).
Upon leaving, \(p\) had therefore voted or cast a catch-up vote for \(b\).

Fourth, no correct party calls \(\fn{request\_disable}()\) on the instance of view \(v\) before voting in it.
With VA-strong-unanimity, this follows from the second and third steps.
Otherwise, the timeout condition gives \(t_f(v)+\Delta_{\mathrm{to}}\geq t_1+d_p(\delta)\), so no correct party calls \(\fn{request\_disable}()\) on the instance before time \(t_1+d_p(\delta)\).
Part 2 of~\Cref{lem:va-continuation}, applied with \(s=t_1\) to the prefix through time \(t_1+d_p(\delta)\), excluding timer firings at that time, shows that every correct party \(p\) prepares \(b\), and hence leaves view \(v\), by time \(t_1+d_p(\delta)\leq t_f(v)+\Delta_{\mathrm{to}}\leq e_p(v)+\Delta_{\mathrm{to}}\).
By~\Cref{lem:consensus-basic}(c), no correct party ever calls \(\fn{request\_disable}()\) on the instance.

Hence the conditions of VA-strong-unanimity, or of VA-unanimity in the second case, hold for the whole execution, with the last correct vote by time \(t_1\) and \(\max\{t_1,T_0\}=t_1\).
Every correct party therefore prepares \(b\), and enters view \(v+1\), by time \(t_1+d_p(\delta)\), which proves (b).
It also commits \(b\) by time \(t_1+d_c(\delta)\), and then calls \(\fn{commit}(b)\) (line~\ref{ln:consensus-commit}).
By VA-consistency, view \(v\) is never disabled, which completes (c) and~(d).
\end{proof}

The commitment in (d) depends only on the instance of view \(v\): it completes regardless of the progress of later views.
The bound on \(t_1\) guarantees that no view timer fires too early, so each use of~\Cref{lem:consensus-good-view} only has to exhibit a suitable time \(t_1\).

\begin{corollary}
\label{lem:consensus-correct-leader}
Let view \(v\) have a correct leader \(L\), with \(t_f(v)\geq T_0\).
Then view \(v\) is never disabled, \(b_v\) is the only block prepared in view \(v\), and every correct party votes for \(b_v\) by time \(e_L(v)+d_t(\delta)\) and enters view \(v+1\) by time \(e_L(v)+d_t(\delta)+d_p(\delta)\).
Moreover, every correct party commits \(b_v\), and does so by time
\[
  e_L(v)+d_t(\delta)+d_c(\delta)\leq t_f(v)+2d_t(\delta)+d_c(\delta).
\]
\end{corollary}

\begin{proof}
Let \(t_1=e_L(v)+d_t(\delta)\).
Since \(e_L(v)\geq t_f(v)\geq T_0\), \Cref{lem:consensus-proposal-propagation} shows that \(t_1\) satisfies the conditions of~\Cref{lem:consensus-good-view}, where \(t_1\leq t_f(v)+2d_t(\delta)\) because \(e_L(v)\leq t_f(v)+d_t(\delta)\) by~\Cref{lem:consensus-view-synchronization}.
\Cref{lem:consensus-good-view} gives the claims, and \(b_v\) is the only block prepared in view \(v\) by VA-consistency.
\end{proof}

\begin{lemma}
\label{lem:consensus-fresh-views}
Let \(V_0\) be the largest view that some correct party has entered by time \(T_0\), which exists because finitely many events occur by then.
\begin{enumerate}[label=(\alph*)]
    \item Every view \(v>V_0\) has \(t_f(v)>T_0\).
    \item Every correct party enters view \(V_0\) by time \(T_0+d_t(\delta)\); hence a correct party that calls \(\fn{extend}\) after time \(T_0+d_t(\delta)\) does so in a view above \(V_0\).
    \item Every correct party leads, and calls \(\fn{extend}\) in, infinitely many views above \(V_0\).
\end{enumerate}
\end{lemma}

\begin{proof}
Part (a) holds by the definition of \(V_0\).
For (b), some correct party enters view \(V_0\) by time \(T_0\), so~\Cref{lem:consensus-view-synchronization} applies; moreover, a correct party \(p\) calls \(\fn{extend}\) in a view \(w\) only at time \(e_p(w)\) (\Cref{lem:consensus-basic}(d)), which exceeds \(e_p(V_0)\) only if \(w>V_0\).
Part (c) follows from round-robin leader rotation, \Cref{lem:consensus-view-progress}, and~\Cref{lem:consensus-basic}(d).
\end{proof}

\subsection{Totality and Liveness}
\label{app:consensus-liveness}

By eventual synchrony, every execution is \(\Delta\)-bounded from \(\mathrm{GST}\), so the results above apply with \(\delta=\Delta\) and \(T_0=\mathrm{GST}\); we use them in this form here.

\begin{theorem}[Totality]
\label{thm:consensus-totality}
Suppose that the timeout condition~\eqref{eq:consensus-timeout} holds.
If a correct party commits a block \(b\), then every correct party eventually commits \(b\).
\end{theorem}

\begin{proof}
The block \(b\) is an ancestor of a block \(c\) committed in some view \(v\).
By~\Cref{lem:consensus-fresh-views}, some view \(w>\max\{v,V_0\}\) has a correct leader and satisfies \(t_f(w)>\mathrm{GST}\).
By~\Cref{lem:consensus-correct-leader}, every correct party commits \(b_w\), and, by~\Cref{lem:consensus-committed-prefix}, \(c\) is an ancestor of \(b_w\).
Hence every correct party commits \(c\) and its ancestor \(b\).
\end{proof}

\begin{theorem}[Liveness]
\label{thm:consensus-liveness}
Suppose that the timeout condition~\eqref{eq:consensus-timeout} holds.
Every correct party obtains infinitely many blocks through calls to \(\fn{extend}\) and commits infinitely many of them.
\end{theorem}

\begin{proof}
By~\Cref{lem:consensus-fresh-views}, every correct party calls \(\fn{extend}\) in infinitely many views \(w>V_0\), each with \(t_f(w)>\mathrm{GST}\), and, by~\Cref{lem:consensus-correct-leader}, every correct party commits the resulting blocks \(b_w\).
These blocks are distinct.
Indeed, let \(V_0<w<w'\) be views with correct leaders.
When the leader of \(w'\) enters \(w'\), view \(w\) is cleared at it and never disabled, so the leader has prepared \(b_w\), the only block prepared in \(w\) (\Cref{lem:consensus-correct-leader}).
Hence, by~\Cref{lem:consensus-basic}(d), \(\fn{parent}(b_{w'})\) is prepared in a view \(\var{vmax}\geq w\), and, by~\Cref{lem:consensus-committed-prefix}, \(b_w\) is an ancestor of \(\fn{parent}(b_{w'})\) and thus a proper ancestor of \(b_{w'}\).
\end{proof}

\begin{corollary}
\label{cor:consensus-eventual-commit}
Suppose that the timeout condition~\eqref{eq:consensus-timeout} holds.
Every block created by a correct party after time \(\mathrm{GST}+d_t(\Delta)\) is eventually committed by every correct party.
\end{corollary}

\begin{proof}
By~\Cref{lem:consensus-fresh-views}, such a block is \(b_w\) for a view \(w>V_0\), so \(t_f(w)>\mathrm{GST}\), and~\Cref{lem:consensus-correct-leader} applies.
\end{proof}

\subsection{Latency Bounds}
\label{app:consensus-latency-proofs}

We now prove the latency bounds of~\Cref{sec:blockchain-va}, assuming that all delay parameters are finite.
A block created by a call to \(\fn{extend}\) at time \(t\) is \emph{committed within}~\(\ell\) if every correct party commits it by time \(t+\ell\).
We read the metrics of~\Cref{sec:model} as follows, where every execution is \(\delta\)-bounded from some \(T_0\geq\mathrm{GST}\).
\begin{itemize}
    \item \(\ell_g\leq\ell\): if \(\mathrm{GST}=T_0=0\), then the block created by a correct party at time~\(0\) is committed within~\(\ell\).
    \item \(\ell_w\leq\ell\): every execution has a time \(T\geq T_0\) such that every block created by a correct party after time \(T\) is committed within~\(\ell\).
    For \(\ell_s\leq\ell\), the same holds for every eventually Byzantine-silent execution.
    \item The block time is at most \(\ell\): in every failure-free execution, the time \(c_N\) of the \(N\)th call to \(\fn{extend}\) satisfies \(\limsup_{N\to\infty}c_N/N\leq\ell\).
    \item The \(k\)-gap latency is at most \(\ell\): every execution with at most \(k\leq f\) Byzantine parties has a time \(T\geq T_0\) such that, whenever a correct party calls \(\fn{extend}\) at a time \(t>T\), another call to \(\fn{extend}\) by a correct party occurs at a time in \([t,t+\ell]\).
\end{itemize}

\begin{theorem}
\label{thm:consensus-latency}
Suppose that all delay parameters are finite and that the timeout condition~\eqref{eq:consensus-timeout} holds.
Then Generic Simplex satisfies
\begin{align*}
\ell_g &\leq \delta+d_c(\delta),\\
\ell_s &\leq d_t^s(\delta)+d_c(\delta),\\
\ell_w &\leq d_t(\delta)+d_c(\delta),\\
\text{block time} &\leq \delta+d_p(\delta),\\
k\text{-gap latency} &\leq d_t(\delta)+d_p(\delta)
  +k\bigl(\Delta_{\mathrm{to}}+\max\{d_f(\delta),d_t(\delta)\}\bigr).
\end{align*}
\end{theorem}

\begin{proof}
The lemmas below prove the five bounds one at a time, in the reading above: good-case commit latency (\Cref{lem:latency-good-case}), eventual worst-case commit latency (\Cref{lem:latency-eventual}), steady-state commit latency (\Cref{lem:latency-steady}), block time (\Cref{lem:latency-block-time}), and \(k\)-gap latency (\Cref{lem:latency-gap}).
\end{proof}

In the rest of this subsection, we assume that all delay parameters are finite and that the timeout condition~\eqref{eq:consensus-timeout} holds.
We consider an execution that is \(\delta\)-bounded from \(T_0\geq\mathrm{GST}\), and \(V_0\) is as in~\Cref{lem:consensus-fresh-views}.

\begin{lemma}[Good-case commit latency]
\label{lem:latency-good-case}
If \(\mathrm{GST}=T_0=0\), then the block created by a correct party at time~\(0\) is committed within \(\delta+d_c(\delta)\).
\end{lemma}

\begin{proof}
The block created at time~\(0\) is \(b_1\), created by the correct leader of view~\(1\), and \(t_f(1)=0\).
By time \(\delta\), every correct party has entered view~\(1\) and received \(b_1\), and \(b_1\) is safe for view~\(1\) at every correct party, since its parent is genesis.
Since \(\delta\leq2d_t(\delta)\), \Cref{lem:consensus-good-view} applies with \(t_1=\delta\) and gives commitment by time \(\delta+d_c(\delta)\).
\end{proof}

\begin{lemma}[Eventual worst-case commit latency]
\label{lem:latency-eventual}
Every block created by a correct party after time \(T_0+d_t(\delta)\) is committed within \(d_t(\delta)+d_c(\delta)\).
\end{lemma}

\begin{proof}
Let \(L\) be a correct party that creates a block at a time \(t>T_0+d_t(\delta)\).
By~\Cref{lem:consensus-fresh-views}, the block is \(b_v\) for a view \(v>V_0\) led by \(L\), with \(e_L(v)=t\) and \(t_f(v)>T_0\).
By~\Cref{lem:consensus-correct-leader}, every correct party commits \(b_v\) by time \(t+d_t(\delta)+d_c(\delta)\).
\end{proof}

\begin{lemma}[Steady-state commit latency]
\label{lem:latency-steady}
If the execution is eventually Byzantine-silent, then there is a time \(T\geq T_0\) such that every block created by a correct party after time \(T\) is committed within \(d_t^s(\delta)+d_c(\delta)\).
\end{lemma}

\begin{proof}
Byzantine parties send no messages after some time \(T_s\).
They send finitely many messages by time \(T_s\), each belonging to at most one instance, so there is a view \(w>V_0\) such that the instances of all views \(v\geq w\) are Byzantine-silent.
By~\Cref{lem:consensus-fresh-views}, some view \(u\geq w\) has a correct leader, and \(t_f(u)>T_0\); by~\Cref{lem:consensus-correct-leader}, view \(u\) is never disabled, and \(b_u\) is the only block prepared in it.
Let \(T\) be a time by which every correct party has entered view \(u+1\); then \(T\geq t_f(u)>T_0\), and every correct party has prepared \(b_u\) by time \(T\).

Let a correct party \(L\) create a block \(b\) at a time \(t>T\); then \(b=b_v\) for a view \(v>u\) with \(e_L(v)=t\), and \(t_f(v)>T_0\) by~\Cref{lem:consensus-fresh-views}.
We show that \(t_1=t+d_t^s(\delta)\) satisfies the conditions of~\Cref{lem:consensus-good-view}.
Every view \(v'\) with \(u<v'<v\) is cleared at \(L\) by time \(t\), and its instance is Byzantine-silent, so every output of view \(v'\) at \(L\) at time \(t\) holds at every correct party by time \(t_1\).
Since the views up to \(u\) are cleared at every correct party by time \(T\), every correct party enters view \(v\) by time \(t_1\) (\Cref{lem:consensus-basic}(a)).
Every correct party also receives \(b\) by time \(t+\delta\leq t_1\).
By~\Cref{lem:consensus-basic}(d), at time \(t\), \(\fn{parent}(b)\in\fn{prepared}_L(\var{vmax})\) and the views strictly between \(\var{vmax}\) and \(v\) are disabled at \(L\), where \(\var{vmax}\geq u\) because \(\fn{prepared}_L(u)\neq\emptyset\).
These views lie strictly between \(u\) and \(v\), and so does \(\var{vmax}\) unless \(\var{vmax}=u\), in which case \(\fn{parent}(b)=b_u\).
Hence \(b\) is safe for view \(v\) at every correct party by time \(t_1\).
Finally, \(t\leq t_f(v)+d_t(\delta)\) by~\Cref{lem:consensus-view-synchronization}, and \(d_t^s(\delta)\leq d_t(\delta)\), so \(t_1\leq t_f(v)+2d_t(\delta)\), and~\Cref{lem:consensus-good-view} gives commitment by time \(t_1+d_c(\delta)=t+d_t^s(\delta)+d_c(\delta)\).
\end{proof}

\begin{lemma}[Block time]
\label{lem:latency-block-time}
Suppose that the execution is failure-free, and let \(e_v\) be the time at which the last party enters view \(v\).
Then \(e_{v+1}\leq e_v+\delta+d_p(\delta)\) for every view \(v\geq V_0+2\), and the block time is at most \(\delta+d_p(\delta)\).
\end{lemma}

\begin{proof}
Every leader is correct, so, by~\Cref{lem:consensus-fresh-views,lem:consensus-correct-leader}, every view \(v>V_0\) has \(t_f(v)>T_0\), is never disabled, and has \(b_v\) as its only prepared block.
Fix a view \(v\geq V_0+2\).
When a party enters view \(v\), view \(v-1\) is cleared at it, so it has prepared \(b_{v-1}\).
Thus the leader of view \(v\) creates \(b_v\) with parent \(b_{v-1}\) (\Cref{lem:consensus-basic}(d)), and \(b_v\) is safe for view \(v\) at every party that has entered \(v\).
The leader creates \(b_v\) by time \(e_v\), so, by time \(e_v+\delta\), every party has entered view \(v\) and received \(b_v\), and \(b_v\) is safe for view \(v\) at every party.
Since \(e_v\leq t_f(v)+d_t(\delta)\) by~\Cref{lem:consensus-view-synchronization}, and \(\delta\leq d_t(\delta)\), \Cref{lem:consensus-good-view} applies with \(t_1=e_v+\delta\) and gives \(e_{v+1}\leq e_v+\delta+d_p(\delta)\).

Moreover, \(e_v\) is nondecreasing in \(v\), and the leader of each view \(v\) calls \(\fn{extend}\) by time \(e_v\), so at least \(N\) calls occur by time \(e_N\).
Hence \(c_N\leq e_N\leq e_{V_0+2}+(N-V_0-2)(\delta+d_p(\delta))\) for every \(N\geq V_0+2\), and \(\limsup_{N\to\infty}c_N/N\leq\delta+d_p(\delta)\).
\end{proof}

\begin{lemma}[\(k\)-gap latency]
\label{lem:latency-gap}
Suppose that at most \(k\leq f\) parties are Byzantine.
If a correct party calls \(\fn{extend}\) at a time \(t>T_0+d_t(\delta)\), then another call to \(\fn{extend}\) by a correct party occurs at a time in \([t,t+\ell]\), where
\[
  \ell=d_t(\delta)+d_p(\delta)+k\bigl(\Delta_{\mathrm{to}}+\max\{d_f(\delta),d_t(\delta)\}\bigr).
\]
\end{lemma}

\begin{proof}
Let \(L\) be the correct party that calls \(\fn{extend}\) at time \(t\).
By~\Cref{lem:consensus-fresh-views}, it does so in a view \(v>V_0\), with \(e_L(v)=t\) and \(t_f(v)>T_0\), and, by~\Cref{lem:consensus-correct-leader}, every correct party enters view \(v+1\) by time \(t+d_t(\delta)+d_p(\delta)\).
Let \(v+j+1\) be the first view after \(v\) with a correct leader; round-robin rotation gives \(j\leq k\).
Applying~\Cref{cor:consensus-view-duration} to views \(v+1,\ldots,v+j\) in turn shows that every correct party enters view \(v+j+1\) by time
\[
  t+d_t(\delta)+d_p(\delta)+j\bigl(\Delta_{\mathrm{to}}+\max\{d_f(\delta),d_t(\delta)\}\bigr)\leq t+\ell,
\]
and the leader of view \(v+j+1\) calls \(\fn{extend}\) upon entering it.
This call occurs no earlier than time \(t\): since view \(v\) is never disabled, its leader enters view \(v+1\) only after preparing \(b_v\), which, by VA-validity, follows a correct vote for \(b_v\); the first such vote is not a catch-up vote, so it follows the receipt of \(b_v\) (\Cref{lem:consensus-basic}(e)), which \(L\) broadcasts at time \(t\).
\end{proof}

\paragraph{Instantiations.}
1/3-VA has \(d_p(\delta)=\delta\), \(d_c(\delta)=2\delta\), \(d_f(\delta)=3\delta\), \(d_t(\delta)\leq3\delta\), and \(d_t^s(\delta)=\delta\) (\Cref{thm:va-third}), and Fast VA has \(d_p(\delta)=d_c(\delta)=\delta\), \(d_f(\delta)=3\delta\), \(d_t(\delta)=2\delta\), and \(d_t^s(\delta)=\delta\) (\Cref{thm:va-fifth}).
Fast VA satisfies VA-strong-unanimity, and so does strongly unanimous 1/3-VA (\Cref{app:va-third-strong}), which has the same delays as 1/3-VA (\Cref{thm:va-third-strong}).
These proofs assume that every message delay is at most \(\delta\), but they also establish the bounds in the form of~\Cref{app:consensus-setting}: each is a chain of message deliveries that starts from messages sent by correct parties by the time \(s\) of the triggering event, and, in an execution that is \(\delta\)-bounded from \(T_0\), the \(i\)th delivery of the chain occurs by time \(\max\{s,T_0\}+i\delta\).
The bounds are linear in \(\delta\), hence nondecreasing and at least \(\delta\).
Substituting them into the timeout condition~\eqref{eq:consensus-timeout} gives the view timeouts \(7\Delta\) and \(5\Delta\), respectively, using only VA-unanimity, and \(6\Delta\) and \(4\Delta\) with VA-strong-unanimity; substituting them into~\Cref{thm:consensus-latency} gives the bounds of~\Cref{tab:signature-free-consensus-latency}.
With \(\delta=\Delta\), \Cref{lem:consensus-correct-leader} gives commitment by time \(t_f(v)+8\Delta\) with 1/3-VA and by time \(t_f(v)+5\Delta\) with Fast VA, in every view \(v\) with a correct leader and \(t_f(v)\geq\mathrm{GST}\).

\clearpage
\section{1/3-VA: Correctness and Complexity}%
\label{app:va-proofs}

This appendix proves that 1/3-VA (\Cref{fig:va-shared,fig:va-third,fig:va-third-amplification}) implements view agreement with the stated delays (\Cref{app:va-correctness}) and bounds its communication complexity (\Cref{app:va-complexity}).
Finally, it shows that, without signatures, a prepare delay of \(\delta\) requires allowing several prepared values when \(n/4<f<n/3\) (\Cref{app:va-multiple-prepared}).

\subsection{Correctness and Latency}
\label{app:va-correctness}

In all implementations, procedure \(\fn{do\_commit}\) prepares a value before committing it (\Cref{fig:va-shared}).
Both \(\fn{do\_commit}\) and \(\fn{clear}\) check \(\fn{valid}\), so every prepared or committed value is valid.
A correct party broadcasts at most one vote, for the value of its first \(\fn{vote}\) call, without checking validity at that point.

\begin{theorem}
\label{thm:va-third}
If \(n\geq3f+1\), then 1/3-VA (\Cref{fig:va-shared,fig:va-third,fig:va-third-amplification}) implements view agreement.
Moreover, when message delay is bounded by \(\delta\), its prepare delay is
\(\delta\), its commit delay is \(2\delta\), its fallback delay is
\(3\delta\), its totality delay is at most \(3\delta\), and its Byzantine-silent
totality delay is \(\delta\).
\end{theorem}

\begin{proof}
The lemmas below prove the properties one at a time: VA-validity (\Cref{lem:va-third-validity}), VA-consistency (\Cref{lem:va-third-consistency}), VA-unanimity with the prepare and commit delays (\Cref{lem:va-third-unanimity}), VA-totality with the totality delays (\Cref{lem:va-third-totality}), and VA-fallback-progress with the fallback delay (\Cref{lem:va-third-fallback}).
\end{proof}

In the rest of this subsection, we assume \(n\geq3f+1\) and let \(t\leq f\) be the actual number of Byzantine parties.
A \emph{vote core} for \(x\) at a party is a set of \(n-f\) parties whose votes for \(x\) it has counted; recall that vote counts ignore equivocating voters.
The time bounds in the lemmas below assume that message delay is bounded by \(\delta\).
Without this assumption, the same arguments show that the stated events eventually happen, since correct parties eventually receive every message sent to them by correct parties.

\begin{lemma}
\label{lem:va-third-basic}
Correct parties \(p\) and \(q\) satisfy the following.
\begin{enumerate}[label=(\alph*)]
    \item If \(p\) and \(q\) have vote cores for \(x\) and \(y\), respectively, then \(x=y\).
    Every vote core contains at least \(n-f-t\geq f+1\) correct parties.
    \item If \(p\) commits \(x\), then some correct party sent \(\sig{\CommitMsg,x}\) upon a vote core for \(x\), which is valid.
    \item If \(p\) sends \(\sig{\Candidate,x}\) with \(x\neq\bot\), then \(x\) is valid and \(p\) received votes for \(x\) from \(f+1\) parties, at least one of them correct.
    Party \(p\) sends \(\sig{\Candidate,\bot}\) only upon \(\fn{request\_disable}()\).
    \item Party \(p\) does not send both \(\sig{\CommitMsg,x}\) and \(\sig{\Candidate,a}\) with \(a\neq x\).
    \item If a correct party sends \(\sig{\Ready,a}\), then some correct party received \(\sig{\Candidate,a}\) from \(2f+1\) parties, at least \(2f+1-t\geq f+1\) of them correct.
\end{enumerate}
\end{lemma}

\begin{proof}
(a)~Two sets of \(n-f\) parties intersect in at least \(n-2f\geq f+1\) parties, hence in a correct party, which votes only once.
A vote core contains at most \(t\) Byzantine parties, and \(n-f-t\geq n-2f\geq f+1\).

(b)~Among the \(n-f\) commit messages that \(p\) received, at least \(n-f-t\geq1\) were sent by correct parties.
A correct party sends \(\sig{\CommitMsg,x}\) only upon a vote core for \(x\) and after checking \(\fn{valid}(x)\).

(c)~This follows from the calls to \(\fn{send\_candidate}\) in~\Cref{fig:va-third}, since at most \(t\leq f\) parties are Byzantine.

(d)~If \(p\) sends \(\sig{\Candidate,a}\) before \(\sig{\CommitMsg,x}\), then \(\var{sent\_candidate}\not\subseteq\{x\}\) prevents the commit message; if it tries afterward, \(\var{sent\_commit}=x\notin\{\None,a\}\) prevents the candidate.

(e)~Consider the first correct party to send \(\sig{\Ready,a}\).
It cannot use the ready-relay rule, since \(f+1\) ready messages include one sent earlier by a correct party.
It therefore uses the \((2f+1)\)-candidate rule, and at most \(t\) of the candidate senders are Byzantine.
\end{proof}

By~\Cref{lem:va-third-basic}(c,e), a correct party sends \(\sig{\Ready,a}\) only if \(a\) is \(\bot\) or a valid value.
We say that a party \emph{clears the instance with} \(a\) when it calls \(\fn{clear}(a)\) and thereby disables, if \(a=\bot\), or prepares \(a\), if \(a\) is a valid value.

\begin{lemma}
\label{lem:va-third-amplification}
Bracha-style amplification (\Cref{fig:va-third-amplification}) satisfies the following.
\begin{enumerate}[label=(\alph*)]
    \item If every correct party has sent \(\sig{\Candidate,a}\) by time \(s\), then every correct party clears the instance with \(a\) by time \(s+2\delta\).
    \item If a correct party clears the instance with \(a\) upon \(2f+1\) ready messages at time \(s\), then every correct party clears the instance with \(a\) by time \(s+2\delta\).
\end{enumerate}
\end{lemma}

\begin{proof}
(a)~Here \(a\) is \(\bot\) or a valid value by~\Cref{lem:va-third-basic}(c).
By \(s+\delta\), every correct party has received at least \(n-t\geq2f+1\) correct candidate messages for \(a\) and sent \(\sig{\Ready,a}\).
By \(s+2\delta\), every correct party has received at least \(n-t\geq2f+1\) correct ready messages for \(a\) and cleared the instance with \(a\).

(b)~At least \(f+1\) of the \(2f+1\) ready messages were sent by correct parties by time \(s\).
By \(s+\delta\), every correct party has received them and relayed \(\sig{\Ready,a}\).
By \(s+2\delta\), every correct party has received at least \(n-t\geq2f+1\) correct ready messages for \(a\) and cleared the instance with \(a\).
\end{proof}

\begin{lemma}[VA-validity]
\label{lem:va-third-validity}
If a correct party prepares or commits \(x\), then \(x\) is valid and some correct party called \(\fn{vote}(x)\).
\end{lemma}

\begin{proof}
Every prepared or committed value is valid.
If a correct party prepares \(x\) upon a vote core, or commits \(x\), then some correct party has a vote core for \(x\), by~\Cref{lem:va-third-basic}(b) in the latter case, and this vote core contains a correct party that called \(\fn{vote}(x)\) (\Cref{lem:va-third-basic}(a)).
Otherwise, it prepares \(x\) upon \(2f+1\) ready messages, one of which was sent by a correct party.
By~\Cref{lem:va-third-basic}(e), some correct party sent \(\sig{\Candidate,x}\), after receiving a correct vote for \(x\) (\Cref{lem:va-third-basic}(c)).
\end{proof}

\begin{lemma}[VA-consistency]
\label{lem:va-third-consistency}
If a correct party commits \(x\), then no correct party prepares or commits a value other than \(x\), or disables.
\end{lemma}

\begin{proof}
Suppose that a correct party commits \(x\), and let \(C_x\) be the correct parties among the \(n-f\) commit senders that it observed.
Then \(|C_x|\geq n-f-t\geq n-2f\), and every party in \(C_x\) sent \(\sig{\CommitMsg,x}\) upon a vote core for \(x\).

By~\Cref{lem:va-third-basic}(a), every correct party that has a vote core has one for \(x\).
Hence no correct party prepares a value \(x'\neq x\) upon a vote core, and, by~\Cref{lem:va-third-basic}(b), no correct party commits \(x'\).

It remains to exclude clearing the instance with \(a\neq x\), a value or \(\bot\), upon ready messages.
One of these ready messages would be sent by a correct party, so, by~\Cref{lem:va-third-basic}(e), some correct party would receive \(\sig{\Candidate,a}\) from a set of \(2f+1\) parties.
This set intersects \(C_x\), because
\[
  |C_x|+(2f+1)\geq(n-2f)+(2f+1)>n.
\]
Yet a party in \(C_x\) sent \(\sig{\CommitMsg,x}\) and therefore cannot send \(\sig{\Candidate,a}\) (\Cref{lem:va-third-basic}(d)).
\end{proof}

\begin{lemma}[VA-unanimity]
\label{lem:va-third-unanimity}
Suppose that all correct parties vote for the same valid value \(x\) and none calls \(\fn{request\_disable}()\), and let \(t_0\) be the time of the last correct vote.
Then every correct party prepares \(x\) by time \(t_0+\delta\) and commits \(x\) by time \(t_0+2\delta\).
\end{lemma}

\begin{proof}
No correct party sends \(\sig{\Candidate,\bot}\), and a candidate for a value \(y\neq x\) requires a correct vote for \(y\) (\Cref{lem:va-third-basic}(c)), so correct parties send candidates only for \(x\).
Moreover, a correct party sends \(\sig{\CommitMsg,y}\) only upon a vote core for \(y\), which contains a correct vote, so only for \(y=x\).

By \(t_0+\delta\), every correct party has received at least \(n-t\geq n-f\) correct votes for \(x\) and prepared \(x\); since its candidate set is a subset of \(\{x\}\), it has also sent \(\sig{\CommitMsg,x}\).
By \(t_0+2\delta\), every correct party has received at least \(n-f\) correct commit messages and committed \(x\).
\end{proof}

\begin{lemma}[VA-totality]
\label{lem:va-third-totality}
If a correct party prepares \(x\), or disables, at time \(t_0\), then every correct party prepares \(x\), or disables, respectively, by time \(t_0+3\delta\).
If the instance is Byzantine-silent, it does so by time \(t_0+\delta\).
\end{lemma}

\begin{proof}
A correct party prepares upon a vote core, upon ready messages, or by committing, and disables only upon ready messages.

First, suppose that a correct party has a vote core for a valid value \(x\) at time \(t_0\), as when it prepares \(x\) upon a vote core.
The vote core contains at least \(f+1\) correct parties (\Cref{lem:va-third-basic}(a)).
A correct party that sent a commit message did so upon a vote core, hence for \(x\) (\Cref{lem:va-third-basic}(a)), so the check of \(\var{sent\_commit}\) in \(\fn{send\_candidate}(x)\) passes at every correct party.
By \(t_0+\delta\), every correct party has received the \(f+1\) correct votes for \(x\) and sent \(\sig{\Candidate,x}\).
By~\Cref{lem:va-third-amplification}(a), every correct party prepares \(x\) by \(t_0+3\delta\).

Second, if a correct party clears the instance with \(a\) upon \(2f+1\) ready messages at time \(t_0\), every correct party does so by \(t_0+2\delta\) (\Cref{lem:va-third-amplification}(b)).

Third, if it prepares \(x\) by committing \(x\) at time \(t_0\), some correct party had a vote core for the valid value \(x\) at some time \(s\leq t_0\) (\Cref{lem:va-third-basic}(b)).
By the first case, every correct party prepares \(x\) by \(s+3\delta\leq t_0+3\delta\).

For the Byzantine-silent totality delay, suppose that Byzantine parties send no messages in the instance.
A correct party prepares or disables at time \(t_0\) upon \(n-f\) votes, \(n-f\) commit messages, or \(2f+1\) ready messages, all broadcast by correct parties by time \(t_0\).
By \(t_0+\delta\), every correct party has received the same messages.
Correct parties do not equivocate, so the vote-counting thresholds count all of them, and the corresponding rules have no precondition other than validity, which every party evaluates alike.
Every correct party therefore clears the instance with the same value or \(\bot\) by \(t_0+\delta\).
\end{proof}

\begin{lemma}[VA-fallback-progress]
\label{lem:va-third-fallback}
Suppose that every correct party calls \(\fn{request\_disable}()\), and let \(\tau\) be the time of the last such call.
Then the instance is cleared at every correct party by time \(\tau+3\delta\), and by time \(\tau+2\delta\) if it is not cleared at any correct party by time \(\tau\).
The bound \(3\delta\) is attained when \(n=3f+1\) and \(t=f\geq1\).
\end{lemma}

\begin{proof}
If the instance is cleared at some correct party by time \(\tau\), \Cref{lem:va-third-totality} ensures that it is cleared at every correct party by time \(\tau+3\delta\).
Otherwise, no correct party has sent a commit message by time \(\tau\), since it would first have prepared a value.
Hence the check of \(\var{sent\_commit}\) in \(\fn{send\_candidate}(\bot)\) passes, and every correct party has sent \(\sig{\Candidate,\bot}\) upon its call to \(\fn{request\_disable}()\), by time \(\tau\).
By~\Cref{lem:va-third-amplification}(a), every correct party disables by \(\tau+2\delta\).

The bound \(3\delta\) is attained when \(n=3f+1\) and \(t=f\geq1\): let \(f+1\) correct parties vote for a valid value \(x\), let the other \(f\) correct parties call \(\fn{request\_disable}()\) without voting, and let the \(f\) Byzantine parties send votes for \(x\) only to one of the correct voters, \(p\), and send no other messages.
Party \(p\) casts the last correct vote, at time \(\tau\), which completes a vote core at \(p\), so \(p\) prepares \(x\) and sends \(\sig{\CommitMsg,x}\).
It then calls \(\fn{request\_disable}()\), which has no effect because \(p\) has already sent a commit message.
The other \(2f\) correct parties have already called \(\fn{request\_disable}()\), after voting if applicable, and sent \(\sig{\Candidate,\bot}\), but the instance is not yet cleared at any of them.
Since \(p\) sends no \(\bot\) candidate, only \(2f\) parties send \(\bot\) candidates, and no correct party sends \(\sig{\Ready,\bot}\).
Hence the instance is cleared at the other correct parties only through \(x\): they receive their \((f+1)\)st vote for \(x\), from \(p\), at \(\tau+\delta\), and then send \(\sig{\Candidate,x}\), send \(\sig{\Ready,x}\) at \(\tau+2\delta\), and prepare \(x\) at \(\tau+3\delta\).
\end{proof}

\subsection{Communication Complexity}
\label{app:va-complexity}

We count messages sent by correct parties throughout one view agreement
instance. Each threshold rule executes at most once per value, and detected
equivocating voters are permanently excluded from local vote counts.
Let \(L\) be the maximum value size and \(\kappa\) the per-message
overhead, both measured in bits.

\begin{theorem}
\label{thm:va-third-complexity}
Under these conventions, if \(n=3f+1\), 1/3-VA sends \(O(n^2)\) messages and
\(O(n^2(L+\kappa))\) bits per instance.
\end{theorem}

\begin{proof}
The case \(n=1\) is immediate.
Suppose \(f\geq1\), and let \(t\leq f\) be the actual number of Byzantine parties.
Sending \(\sig{\Candidate,x}\) for a value \(x\) requires at least \(f+1\) votes for \(x\).
At a fixed correct party, the sets of voters counted for different candidate values are disjoint: a second vote for a different value reveals equivocation, excluding that voter from subsequent counts.
Thus each correct party sends \(\Candidate\) for at most \(\lfloor(3f+1)/(f+1)\rfloor=2\) values other than \(\bot\).

Fix an argument \(a\neq\bot\) for which some correct party sends \(\sig{\Ready,a}\).
By~\Cref{lem:va-third-basic}(e), at least \(2f+1-t\) correct parties send \(\sig{\Candidate,a}\).
There are at most \(2(n-t)\) pairs of a correct party and a value other than \(\bot\) for which it sends \(\Candidate\), so there are at most three such arguments \(a\), since
\[
  \frac{2(n-t)}{2f+1-t}
  =2+\frac{2f}{2f+1-t}
  \leq2+\frac{2f}{f+1}<4.
\]
Including \(\bot\), correct parties therefore send \(\Ready\) for at most four distinct arguments.

For each argument, a correct party broadcasts \(\Ready\) at most twice: once under the \((2f+1)\)-candidate rule and once under the \((f+1)\)-ready rule.
Thus each correct party broadcasts at most one \(\Vote\), one \(\CommitMsg\), three \(\Candidate\) messages, and eight \(\Ready\) messages.
Hence there are at most \(13(n-t)(n-1)\leq13n(n-1)\) point-to-point messages.
Each carries \(O(L+\kappa)\) bits, giving the stated bounds.
\end{proof}

For fixed-size values and overhead, the bit complexity is \(O(n^2)\);
in a blockchain consensus instantiation, \(L\) includes the block payload.

\subsection{Multiple Prepared Values Are Necessary}
\label{app:va-multiple-prepared}

1/3-VA has a prepare delay of \(\delta\) and satisfies VA-totality, but it may prepare several values in an instance in which no value is committed (\Cref{sec:signature-free-implementations}).
We show that, without signatures, allowing several prepared values is necessary for a prepare delay of \(\delta\) when \(n/4<f<n/3\).

\begin{theorem}
\label{thm:va-multiple-prepared}
Let \(n/4<f<n/3\), and consider a signature-free asynchronous implementation of view agreement with prepare delay \(d_p(\delta)=\delta\).
Then, in some execution, correct parties prepare two different values.
\end{theorem}

\begin{proof}
Suppose, for contradiction, that correct parties never prepare two different values.
We use the implementation, with a validity predicate that accepts every value, to build a Byzantine reliable broadcast protocol~\cite{bracha1987asynchronous} with a designated sender~\(s\).
The sender sends its input \(m\) to all parties at time~\(0\).
Every party calls \(\fn{vote}(m')\) upon receiving the first value \(m'\) from \(s\), never calls \(\fn{request\_disable}()\), and delivers the first value it prepares.
Correct parties thus respect the interface of view agreement.

By assumption, correct parties deliver at most one value in total (agreement), and, by VA-totality, if a correct party delivers a value, then every correct party eventually delivers it (totality).
Now suppose that \(s\) is correct.
Every correct party then votes for \(m\) and none requests disabling, which meets the conditions of VA-unanimity.
By VA-unanimity, every correct party eventually commits \(m\), hence prepares \(m\) (VA-consistency) and delivers it (validity).
Moreover, if message delay is bounded by \(\delta\), the last correct vote occurs by time \(\delta\), so every correct party prepares \(m\), and delivers it, by time \(\delta+d_p(\delta)=2\delta\).

We have obtained a signature-free Byzantine reliable broadcast protocol with a good-case latency of two message delays, that is, of two asynchronous rounds.
This contradicts the lower bound of Abraham, Ren, and Xiang: without signatures (in their terms, unauthenticated), reliable broadcast needs three rounds when \(3f+1\leq n\leq4f-1\)~\cite[Theorem~10]{abraham2022good}.
\end{proof}

The argument does not apply when \(n\geq4f\), since signature-free reliable broadcast then has a good-case latency of two rounds~\cite{abraham2022good}.
With signatures, a prepare delay of \(\delta\) is compatible with at most one prepared value: the Simplex implementation of~\Cref{app:simplex} achieves both, and obtains VA-totality by forwarding the signed votes that justify each preparation.

\clearpage
\section{VA-Strong-Unanimity and Shorter View Timeouts}%
\label{app:strong-unanimity}

In Generic Simplex, the view timeout must leave enough time for the proposal of a correct leader to be prepared, not merely voted for, before any correct party's view timer fires (hence the prepare delay in the timeout condition of~\Cref{sec:blockchain-va}).
The reason is that VA-unanimity guarantees nothing once some party requests disabling, even after voting.
The following strengthening rules this out.

\begin{definition}[VA-strong-unanimity]
\label{def:va-strong-unanimity}
If all correct parties vote for the same valid value \(x\) and none requests disabling before voting, all eventually commit~\(x\).
\end{definition}

The conditions of VA-unanimity imply those of VA-strong-unanimity, so VA-strong-unanimity implies VA-unanimity.
For an implementation that satisfies VA-strong-unanimity, we measure the prepare and commit delays under its conditions (\Cref{app:consensus-setting}).

\paragraph{Shorter view timeouts.}
With VA-strong-unanimity, a correct party whose view timer fires after it has voted for the leader's proposal can no longer prevent the proposal from committing.
The prepare delay is then unnecessary, and the timeout condition~\eqref{eq:consensus-timeout} becomes \(\Delta_{\mathrm{to}}\geq2d_t(\Delta)\).
\Cref{lem:consensus-good-view} establishes progress and latency under either form of the timeout condition, so the latency bounds of~\Cref{sec:blockchain-va} continue to hold.

\paragraph{Implementations.}
Fast VA (\Cref{app:va-fifth}), the synchronous implementation (\Cref{app:synchronous-va}), and the Minimmit implementation (\Cref{app:minimmit}) satisfy VA-strong-unanimity, whereas the Simplex implementation (\Cref{app:simplex}) does not.
Neither does 1/3-VA: a party that requests disabling after voting, but before sending a commit message, sends \(\sig{\Candidate,\bot}\), which prevents it from sending a commit message later.
When \(n=3f+1\) and the \(f\) Byzantine parties are silent, a commit requires commit messages from all correct parties, so a single such request prevents it.
We now present strongly unanimous 1/3-VA, a variant of 1/3-VA that satisfies VA-strong-unanimity with the same delays.
With this variant, Generic Simplex needs a view timeout of only \(6\Delta\) instead of \(7\Delta\); with Fast VA, VA-strong-unanimity similarly reduces it from \(5\Delta\) to \(4\Delta\) (\Cref{app:consensus-latency-proofs}).

\subsection{Strongly Unanimous 1/3-VA}%
\label{app:va-third-strong}

\begin{figure*}[!t]
    \centering
    \small
    \setlength{\fboxsep}{2pt}%
    \fbox{%
    \begin{minipage}{\dimexpr\textwidth-2\fboxsep-2\fboxrule\relax}
    \begin{minipage}[t]{0.485\linewidth}
    \begin{algorithmic}[1]
        \State \textbf{Private state:}
        \State \ind $\var{vote\_value} \gets \None$;\hfill$\star$\label{ln:va-third-strong-vote-value}
        \State \ind $\var{sent\_commit} \gets \None$;
        \State \ind $\var{sent\_candidate} \gets \emptyset$;
        \BlankLine
        \State \kw{upon} $\fn{vote}(x)$:
        \State \ind $\var{vote\_value} \gets x$;\hfill$\star$\label{ln:va-third-strong-record-vote}
        \State \ind \kw{broadcast} $\sig{\Vote,x}$;
        \BlankLine
        \State \kw{upon} $\nrecv{\sig{\Vote,x}} \geq n-f$
        \State \indd \kw{and} $\fn{valid}(x)$:
        \State \ind \kw{call} $\fn{clear}(x)$;
        \State \ind \kw{if} $\var{sent\_commit} = \None$
        \State \inddd \kw{and} $\var{sent\_candidate} \subseteq \{x\}$:
        \State \indd $\var{sent\_commit} \gets x$;
        \State \indd \kw{broadcast} $\sig{\CommitMsg,x}$;
        \BlankLine
        \State \kw{upon} $\nrecv{\sig{\CommitMsg,x}} \geq n-f$:
        \State \ind \kw{call} $\fn{do\_commit}(x)$;
        \algstore{vathirdstrong}
    \end{algorithmic}
    \end{minipage}\hfill
    \begin{minipage}[t]{0.485\linewidth}
    \begin{algorithmic}[1]
        \algrestore{vathirdstrong}
        \State \kw{procedure} $\fn{send\_candidate}(a)$:
        \State \ind \kw{if} $a\notin\var{sent\_candidate}$
        \State \inddd \kw{and} $\var{sent\_commit}\in\{\None,a\}$:
        \State \indd $\var{sent\_candidate} \gets$
        \State \inddd $\var{sent\_candidate} \cup \{a\}$;
        \State \indd \kw{broadcast} $\sig{\Candidate,a}$;
        \BlankLine
        \State \kw{upon} $\nrecv{\sig{\Vote,x}} \geq f+1$:
        \State \ind \kw{if} $\fn{valid}(x)$:
        \State \indd \kw{call} $\fn{send\_candidate}(x)$;
        \State \ind \kw{else}: \kw{call} $\fn{send\_candidate}(\bot)$;\hfill$\star$\label{ln:va-third-strong-invalid-votes}
        \BlankLine
        \State \kw{upon} $\fn{request\_disable}()$:
        \State \ind \kw{if} $\var{vote\_value} = \None$:\hfill$\star$\label{ln:va-third-strong-request-disable}
        \State \indd \kw{call} $\fn{send\_candidate}(\bot)$;
        \BlankLine
        \State \kw{upon} $\var{vote\_value} = x \neq \None$\hfill$\star$\label{ln:va-third-strong-conflicting-votes}
        \State \indd \kw{and} $\#\!\bigl[\,\sig{\Vote,y}: y \neq x$
        \State \inddd \kw{or} $\sig{\Candidate,\bot}\,\bigr] \geq f+1$:
        \State \ind \kw{call} $\fn{send\_candidate}(\bot)$;
        \BlankLine
        \State \kw{upon} $\nrecv{\sig{\Candidate,a}} \geq 2f+1$:
        \State \ind \kw{broadcast} $\sig{\Ready,a}$;
    \end{algorithmic}
    \end{minipage}
    \end{minipage}}
    \caption{Strongly unanimous 1/3-VA: core protocol, replacing~\Cref{fig:va-third}; the shared state and helpers (\Cref{fig:va-shared}) and the amplification (\Cref{fig:va-third-amplification}) are unchanged.
    Starred lines differ from~\Cref{fig:va-third}.
    Thresholds count distinct senders, and vote counts ignore equivocating voters.}%
    \label{fig:va-third-strong}
\end{figure*}

The variant in~\Cref{fig:va-third-strong} differs from 1/3-VA only in when a party sends \(\sig{\Candidate,\bot}\).
A party ignores \(\fn{request\_disable}()\) once it has voted (line~\ref{ln:va-third-strong-request-disable}), since a \(\bot\) candidate would prevent it from sending a commit message, which a commit may require.
To still guarantee VA-fallback-progress, a party that voted for \(x\) sends \(\sig{\Candidate,\bot}\) once it observes \(f+1\) parties that voted for values other than \(x\) or sent \(\sig{\Candidate,\bot}\) (line~\ref{ln:va-third-strong-conflicting-votes}).
This cannot happen under the conditions of VA-strong-unanimity, since one of these parties is correct.
Finally, a party sends \(\sig{\Candidate,\bot}\) upon \(f+1\) votes for an invalid value (line~\ref{ln:va-third-strong-invalid-votes}), which covers the case in which \(f+1\) correct parties vote for the same invalid value.

\begin{theorem}
\label{thm:va-third-strong}
If \(n\geq3f+1\), then strongly unanimous 1/3-VA (\Cref{fig:va-shared,fig:va-third-strong,fig:va-third-amplification}) implements view agreement and satisfies VA-strong-unanimity, with the same delays as 1/3-VA (\Cref{thm:va-third}): when message delay is bounded by \(\delta\), its prepare delay is \(\delta\), its commit delay is \(2\delta\), its fallback delay is \(3\delta\), its totality delay is at most \(3\delta\), and its Byzantine-silent totality delay is \(\delta\).
If \(n=3f+1\), it sends \(O(n^2)\) messages and \(O(n^2(L+\kappa))\) bits per instance.
\end{theorem}

\begin{proof}
\Cref{lem:va-third-basic,lem:va-third-amplification,lem:va-third-validity,lem:va-third-consistency,lem:va-third-totality} and~\Cref{thm:va-third-complexity} hold for the variant, with the same proofs, except for the second sentence of~\Cref{lem:va-third-basic}(c), which \Cref{lem:va-third-strong-bot} below replaces.
None of these proofs uses that sentence; they rely only on the checks in \(\fn{send\_candidate}\) and on the rules that prepare, commit, and send ready messages, all of which are unchanged.
This gives VA-validity, VA-consistency, VA-totality with the totality delays, and the complexity bounds.
\Cref{lem:va-third-strong-unanimity} gives VA-strong-unanimity, and hence VA-unanimity, with the prepare and commit delays, and \Cref{lem:va-third-strong-fallback} gives VA-fallback-progress with the fallback delay.
\end{proof}

We keep the assumptions and notation of~\Cref{app:va-correctness}.

\begin{lemma}
\label{lem:va-third-strong-bot}
A correct party sends \(\sig{\Candidate,\bot}\) in three cases only: upon \(\fn{request\_disable}()\) if it has not voted, upon \(f+1\) votes for an invalid value, and, if it voted for some \(x\), upon observing \(f+1\) parties that voted for values other than \(x\) or sent \(\sig{\Candidate,\bot}\).
\end{lemma}

\begin{proof}
These are the calls to \(\fn{send\_candidate}(\bot)\) in~\Cref{fig:va-third-strong}.
\end{proof}

\begin{lemma}[VA-strong-unanimity]
\label{lem:va-third-strong-unanimity}
Suppose that all correct parties vote for the same valid value \(x\) and none calls \(\fn{request\_disable}()\) before voting, and let \(t_0\) be the time of the last correct vote.
Then every correct party prepares \(x\) by time \(t_0+\delta\) and commits \(x\) by time \(t_0+2\delta\).
\end{lemma}

\begin{proof}
We first show that correct parties send candidates only for \(x\).
A candidate for a value \(y\neq x\) requires a correct vote for \(y\) (\Cref{lem:va-third-basic}(c)), and so does a \(\bot\) candidate upon \(f+1\) votes for an invalid value \(y\), while every correct vote is for the valid value \(x\).
A correct party ignores its call to \(\fn{request\_disable}()\), since it has voted by then.
Hence, by~\Cref{lem:va-third-strong-bot}, the first correct party to send \(\sig{\Candidate,\bot}\), if any, does so upon observing \(f+1\) parties that voted for values other than its own vote \(x\) or sent \(\sig{\Candidate,\bot}\).
One of these parties is correct, but a correct party votes only for \(x\) and had not sent \(\sig{\Candidate,\bot}\) by then, a contradiction.
Moreover, a correct party sends \(\sig{\CommitMsg,y}\) only upon a vote core for \(y\), which contains a correct vote, so only for \(y=x\).

By \(t_0+\delta\), every correct party has received at least \(n-t\geq n-f\) correct votes for \(x\) and prepared \(x\); since its candidate set is a subset of \(\{x\}\), it has also sent \(\sig{\CommitMsg,x}\).
By \(t_0+2\delta\), every correct party has received at least \(n-f\) correct commit messages and committed \(x\).
\end{proof}

\begin{lemma}[VA-fallback-progress]
\label{lem:va-third-strong-fallback}
Suppose that every correct party calls \(\fn{request\_disable}()\), and let \(\tau\) be the time of the last such call.
Then the instance is cleared at every correct party by time \(\tau+3\delta\).
This bound is attained when \(n=3f+1\) and \(t=f\geq1\).
\end{lemma}

\begin{proof}
If the instance is cleared at some correct party by time \(\tau\), \Cref{lem:va-third-totality} ensures that it is cleared at every correct party by time \(\tau+3\delta\).
Otherwise, no correct party has sent a commit message by time \(\tau\), since it would first have prepared a value.
A correct party votes only before calling \(\fn{request\_disable}()\), so every correct vote is cast by time \(\tau\), and every correct party that has not voted has sent \(\sig{\Candidate,\bot}\) upon its call, by time \(\tau\).

First suppose that \(f+1\) correct parties voted for the same value \(x\), and let \(a=x\) if \(x\) is valid and \(a=\bot\) otherwise.
A correct party that sends \(\sig{\CommitMsg,y}\) received \(n-f\) votes for the valid value \(y\), at least \(n-f-t\) of them correct, and \((n-f-t)+(f+1)>n-t\), so some correct party voted for both \(x\) and \(y\); hence \(y=x\) and \(a=x\).
Therefore the check of \(\var{sent\_commit}\) in \(\fn{send\_candidate}(a)\) passes at every correct party, and, by \(\tau+\delta\), every correct party has received the \(f+1\) votes for \(x\) and sent \(\sig{\Candidate,a}\).
By~\Cref{lem:va-third-amplification}(a), every correct party clears the instance with \(a\) by \(\tau+3\delta\).

Now suppose that at most \(f\) correct parties voted for each value.
Then no correct party ever sends a commit message, since a vote core contains at least \(f+1\) correct parties (\Cref{lem:va-third-basic}(a)).
By \(\tau+\delta\), every correct party \(p\) that voted for some \(x\) has received, from each of the at least \(n-t-f\geq f+1\) correct parties that did not vote for \(x\), a vote for another value or \(\sig{\Candidate,\bot}\), and has therefore sent \(\sig{\Candidate,\bot}\).
Thus every correct party has sent \(\sig{\Candidate,\bot}\) by \(\tau+\delta\), and, by~\Cref{lem:va-third-amplification}(a), every correct party disables by \(\tau+3\delta\).

The execution in the proof of~\Cref{lem:va-third-fallback} attains this bound here as well: \(p\) ignores its call to \(\fn{request\_disable}()\), since it has voted, and each other correct voter observes no vote for another value and only the \(f\) \(\bot\) candidates of the correct parties that have not voted.
Hence fewer than \(2f+1\) parties send \(\bot\) candidates, and the correct parties other than \(p\) clear the instance only through \(x\), at \(\tau+3\delta\).
\end{proof}

\clearpage
\section{Fast VA: A 1/5-Resilient Implementation}
\label{app:va-fifth}

\Cref{fig:va-fifth} presents Fast VA, our view agreement implementation for \(f<n/5\).
Like 1/3-VA, it uses the shared state and helper procedures of~\Cref{fig:va-shared}.
A valid value is committed, and thus prepared, directly from \(n-f\) votes.
After receiving \(2f+1\) votes for a value \(x\), a party broadcasts \(\sig{\Ready,x}\) if \(\fn{valid}(x)\), and \(\sig{\Disable}\) otherwise.
A party that has not voted also broadcasts \(\sig{\Disable}\) upon \(\fn{request\_disable}()\), and a party that has voted does so when it observes \(2f+1\) parties that voted for other values or sent \(\sig{\Disable}\).
A voter never sends \(\sig{\Disable}\) merely because it called \(\fn{request\_disable}()\): its vote may already belong to a commit quorum.
This gives VA-strong-unanimity (\Cref{def:va-strong-unanimity}).
The prepare and commit delays of Fast VA are \(\delta\), its fallback delay is \(3\delta\), its totality delay is \(2\delta\), and its Byzantine-silent totality delay is \(\delta\).

\begin{figure*}[!ht]
    \centering
    \small
    \setlength{\fboxsep}{2pt}%
    \fbox{%
    \begin{minipage}{\dimexpr\textwidth-2\fboxsep-2\fboxrule\relax}
    \begin{minipage}[t]{0.485\linewidth}
    \begin{algorithmic}[1]
        \State \kw{import} Fig.~\ref{fig:va-shared};
        \BlankLine
        \State \textbf{Private state:}
        \State \ind $\var{vote\_value} \gets \None$;
        \BlankLine
        \State \kw{upon} $\fn{vote}(x)$:
        \State \ind \kw{if} $\var{vote\_value} = \None$:
        \State \indd $\var{vote\_value} \gets x$;
        \State \indd \kw{broadcast} $\sig{\Vote, x}$;
        \BlankLine
        \State \kw{upon} $\fn{request\_disable}()$:
        \State \ind \kw{if} $\var{vote\_value} = \None$:
        \State \indd \kw{broadcast} $\sig{\Disable}$;
        \BlankLine
        \State \kw{upon} $\nrecv{\sig{\Vote, x}} \geq n-f$:
        \State \ind \kw{call} $\fn{do\_commit}(x)$;
        \BlankLine
        \State \kw{upon} $\nrecv{\sig{\Vote, x}} \geq 2f+1$:
        \State \ind \kw{if} $\fn{valid}(x)$:
        \State \indd \kw{broadcast} $\sig{\Ready, x}$;
        \State \ind \kw{else}:
        \State \indd \kw{broadcast} $\sig{\Disable}$;
        \algstore{vafifth}
    \end{algorithmic}
    \end{minipage}\hfill
    \begin{minipage}[t]{0.485\linewidth}
    \begin{algorithmic}[1]
        \algrestore{vafifth}
        \State \kw{upon} $\var{vote\_value} = x \neq \None$
        \State \indd and $\#\!\bigl[\,\sig{\Vote, y}: y \neq x$
        \State \inddd or $\sig{\Disable}\,\bigr] \geq 2f+1$:
        \State \ind \kw{broadcast} $\sig{\Disable}$;
        \BlankLine
        \State \kw{upon} $\nrecv{\sig{\Disable}} \geq 2f+1$:
        \State \ind \kw{broadcast} $\sig{\Ready, \bot}$;
        \BlankLine
        \State \kw{upon} $\nrecv{\sig{\Ready, a}} \geq f+1$:
        \State \ind \kw{broadcast} $\sig{\Ready, a}$;
        \BlankLine
        \State \kw{upon} $\nrecv{\sig{\Ready, a}} \geq 2f+1$:
        \State \ind \kw{call} $\fn{clear}(a)$;
    \end{algorithmic}
    \end{minipage}
    \end{minipage}}
    \caption{Fast VA: a view agreement implementation for \(f<n/5\).
    All thresholds count distinct senders.
    \(\sig{\Vote,\cdot}\) messages carry values and \(\sig{\Ready,\cdot}\) messages carry a value or \(\bot\); malformed messages are ignored.}%
    \label{fig:va-fifth}
\end{figure*}

\begin{theorem}
\label{thm:va-fifth}
If \(n\geq 5f+1\), then Fast VA (\Cref{fig:va-fifth}) implements view agreement and satisfies VA-strong-unanimity.
Moreover, when message delay is bounded by \(\delta\), its prepare and commit delays are \(\delta\), its fallback delay is \(3\delta\), its totality delay is \(2\delta\), and its Byzantine-silent totality delay is \(\delta\).
\end{theorem}

\begin{proof}
The lemmas below prove the properties one at a time: VA-validity (\Cref{lem:fifth-va-validity}), VA-consistency (\Cref{lem:fifth-va-consistency}), VA-strong-unanimity with the prepare and commit delays (\Cref{lem:fifth-va-unanimity}), VA-totality with the totality delays (\Cref{lem:fifth-va-totality}), and VA-fallback-progress with the fallback delay (\Cref{lem:fifth-va-fallback}).
The conditions of VA-unanimity imply those of VA-strong-unanimity, so VA-unanimity holds as well.
\end{proof}

In the rest of this appendix, we assume \(n\geq5f+1\) and let \(t\leq f\) be the actual number of Byzantine parties, so that \(n-t\geq4f+1\).
The time bounds in the lemmas below assume that message delay is bounded by \(\delta\).
Without this assumption, the same arguments show that the stated events eventually happen, since correct parties eventually receive every message sent to them by correct parties.

\begin{lemma}
\label{lem:fifth-va-ready}
If a correct party sends \(\sig{\Ready,a}\), then some correct party received, from \(2f+1\) parties, either votes for \(a\), if \(a\neq\bot\), in which case \(a\) is valid, or \(\sig{\Disable}\), if \(a=\bot\).
At least \(2f+1-t\geq f+1\) of these parties are correct.
\end{lemma}

\begin{proof}
Consider the first correct party to send \(\sig{\Ready,a}\).
It cannot use the relay rule, since \(f+1\) ready messages include one sent earlier by a correct party.
It therefore uses the rule for \(2f+1\) votes, which checks \(\fn{valid}(a)\), if \(a\neq\bot\), and the rule for \(2f+1\) \(\Disable\) messages if \(a=\bot\).
At most \(t\) of the senders are Byzantine.
\end{proof}

By~\Cref{lem:fifth-va-ready}, a correct party sends \(\sig{\Ready,a}\) only if \(a\) is \(\bot\) or a valid value.
As in~\Cref{app:va-proofs}, a party \emph{clears the instance with} \(a\) when it calls \(\fn{clear}(a)\) and thereby disables, if \(a=\bot\), or prepares \(a\), if \(a\) is a valid value.

\begin{lemma}
\label{lem:fifth-va-amplification}
Let \(a\) be \(\bot\) or a valid value.
\begin{enumerate}[label=(\alph*)]
    \item If every correct party has sent \(\sig{\Ready,a}\) by time \(s\), then every correct party clears the instance with \(a\) by time \(s+\delta\).
    \item If at least \(f+1\) correct parties have sent \(\sig{\Ready,a}\) by time \(s\), then every correct party clears the instance with \(a\) by time \(s+2\delta\).
    In particular, this holds if a correct party clears the instance with \(a\) upon \(2f+1\) ready messages at time \(s\).
\end{enumerate}
\end{lemma}

\begin{proof}
(a)~By \(s+\delta\), every correct party has received at least \(n-t\geq2f+1\) correct ready messages for \(a\) and called \(\fn{clear}(a)\).

(b)~By \(s+\delta\), every correct party has received the \(f+1\) correct ready messages and has sent \(\sig{\Ready,a}\), by the relay rule if not earlier, so (a) applies at time \(s+\delta\).
Of \(2f+1\) ready messages received by time \(s\), at least \(2f+1-t\geq f+1\) were sent by correct parties.
\end{proof}

\begin{lemma}[VA-validity]
\label{lem:fifth-va-validity}
If a correct party prepares or commits \(x\), then \(x\) is valid and some correct party called \(\fn{vote}(x)\).
\end{lemma}

\begin{proof}
Every prepared or committed value is valid, since \(\fn{do\_commit}\) and \(\fn{clear}\) check \(\fn{valid}\) (\Cref{fig:va-shared}).
A correct party sends its unique vote only upon a call \(\fn{vote}(x)\).
A correct party commits, and prepares, \(x\) upon \(n-f\) votes for \(x\), at least \(n-f-t\geq1\) of them correct, or prepares \(x\) upon \(2f+1\) messages \(\sig{\Ready,x}\), one of them sent by a correct party.
In the second case, by~\Cref{lem:fifth-va-ready}, some correct party received at least \(f+1\) correct votes for \(x\).
\end{proof}

\begin{lemma}[VA-consistency]
\label{lem:fifth-va-consistency}
If a correct party commits \(x\), then no correct party prepares or commits a value other than \(x\), or disables.
\end{lemma}

\begin{proof}
Suppose that a correct party commits \(x\), and let \(Q_x\) be the set of \(n-f\) parties whose votes for \(x\) it counted.
A set of \(2f+1\) parties intersects \(Q_x\) in at least \((n-f)+(2f+1)-n=f+1\) parties, hence in a correct party, which voted for \(x\) and for no other value.
Since \(n-f\geq2f+1\), no correct party receives \(n-f\), or even \(2f+1\), votes for a value other than \(x\).

Hence no correct party commits a value \(x'\neq x\), and, by~\Cref{lem:fifth-va-ready}, no correct party sends \(\sig{\Ready,x'}\).
The Byzantine parties alone cannot supply \(2f+1\) ready messages, so no correct party prepares \(x'\) upon ready messages either.
Moreover, \(x\) is valid, so no correct party sends \(\sig{\Disable}\) upon \(2f+1\) votes for an invalid value.

It remains to show that no correct party disables.
By~\Cref{lem:fifth-va-ready}, and since the Byzantine parties alone cannot supply \(2f+1\) messages \(\sig{\Ready,\bot}\), it suffices to show that no correct party receives \(\sig{\Disable}\) from \(2f+1\) parties.
We first show that no correct member of \(Q_x\) sends \(\sig{\Disable}\).
Otherwise, let \(q\) be the first correct member of \(Q_x\) to do so.
Party \(q\) voted, and a party never votes after requesting disabling, so \(q\) does not send \(\sig{\Disable}\) upon \(\fn{request\_disable}()\).
By the above, \(q\) therefore observed \(2f+1\) parties that voted for values other than \(x\) or sent \(\sig{\Disable}\).
The correct ones among them lie outside \(Q_x\), since correct members of \(Q_x\) voted only for \(x\) and, by the choice of \(q\), had not sent \(\sig{\Disable}\).
At most \(f\) correct parties lie outside \(Q_x\), and at most \(f\) parties are Byzantine, so there are at most \(2f<2f+1\) such parties, a contradiction.
Hence only the at most \(f\) correct parties outside \(Q_x\) and the at most \(f\) Byzantine parties send \(\sig{\Disable}\), fewer than \(2f+1\).
\end{proof}

\begin{lemma}[VA-strong-unanimity]
\label{lem:fifth-va-unanimity}
Suppose that all correct parties vote for the same valid value \(x\) and none calls \(\fn{request\_disable}()\) before voting, and let \(t_0\) be the time of the last correct vote.
Then every correct party commits, and thus prepares, \(x\) by time \(t_0+\delta\).
\end{lemma}

\begin{proof}
By \(t_0+\delta\), every correct party has received at least \(n-t\geq n-f\) correct votes for \(x\) and called \(\fn{do\_commit}(x)\), which prepares and commits the valid value \(x\).
This argument does not use the condition on \(\fn{request\_disable}()\).
\end{proof}

\begin{lemma}[VA-totality]
\label{lem:fifth-va-totality}
If a correct party prepares \(x\), or disables, at time \(t_0\), then every correct party prepares \(x\), or disables, respectively, by time \(t_0+2\delta\).
If the instance is Byzantine-silent, it does so by time \(t_0+\delta\).
\end{lemma}

\begin{proof}
A correct party disables only upon \(2f+1\) ready messages, and prepares a value either upon \(2f+1\) ready messages or by committing it upon \(n-f\) votes.

First, if a correct party clears the instance with \(a\) upon \(2f+1\) ready messages at time \(t_0\), every correct party does so by \(t_0+2\delta\) (\Cref{lem:fifth-va-amplification}(b)).

Second, suppose that a correct party commits \(x\) at time \(t_0\).
Its \(n-f\) votes for \(x\) include at least \(n-f-t\geq n-2f\geq2f+1\) correct ones.
By \(t_0+\delta\), every correct party has received these votes and sent \(\sig{\Ready,x}\), since \(x\) is valid, and, by~\Cref{lem:fifth-va-amplification}(a), it prepares \(x\) by \(t_0+2\delta\).

For the Byzantine-silent totality delay, suppose that Byzantine parties send no messages in the instance.
A correct party prepares or disables at time \(t_0\) upon \(n-f\) votes or \(2f+1\) ready messages, all broadcast by correct parties by time \(t_0\).
By \(t_0+\delta\), every correct party has received the same messages.
The corresponding rules have no precondition other than validity, which every party evaluates alike, so every correct party clears the instance with the same value or \(\bot\) by \(t_0+\delta\).
\end{proof}

\begin{lemma}[VA-fallback-progress]
\label{lem:fifth-va-fallback}
Suppose that every correct party calls \(\fn{request\_disable}()\), and let \(\tau\) be the time of the last such call.
Then the instance is cleared at every correct party by time \(\tau+3\delta\).
This bound is attained whenever \(f\geq1\).
\end{lemma}

\begin{proof}
A party never votes after requesting disabling, so every correct vote is cast by time \(\tau\), and every correct party that has not voted has sent \(\sig{\Disable}\) upon its call, by time \(\tau\).
By \(\tau+\delta\), every correct party has received all of these messages.
We distinguish three cases.

If at least \(2f+1\) correct parties voted for the same valid value \(x\), every correct party has sent \(\sig{\Ready,x}\) by \(\tau+\delta\) and, by~\Cref{lem:fifth-va-amplification}(a), prepares \(x\) by \(\tau+2\delta\).

If at least \(2f+1\) correct parties voted for the same invalid value, every correct party has sent \(\sig{\Disable}\) by \(\tau+\delta\).

Otherwise, at most \(2f\) correct parties voted for each value.
Let \(p\) be a correct party that voted for some \(x\).
At least \((n-t)-2f\geq n-3f\geq2f+1\) correct parties voted for other values or did not vote and sent \(\sig{\Disable}\), so \(p\) has sent \(\sig{\Disable}\) by \(\tau+\delta\).
Hence, in this case too, every correct party has sent \(\sig{\Disable}\) by \(\tau+\delta\).

In the last two cases, by \(\tau+2\delta\), every correct party has received at least \(n-t\geq2f+1\) \(\Disable\) messages and sent \(\sig{\Ready,\bot}\), and, by~\Cref{lem:fifth-va-amplification}(a), it disables by \(\tau+3\delta\).

The bound \(3\delta\) is attained when \(f\geq1\): let all correct parties vote for distinct values and call \(\fn{request\_disable}()\) at time \(\tau\), let the Byzantine parties send no messages, and let every message be received after exactly \(\delta\).
No value receives \(2f+1\) votes, and no party sends \(\sig{\Disable}\) before \(\tau+\delta\), when every correct party observes \(n-t-1\geq4f\geq2f+1\) votes for other values and sends \(\sig{\Disable}\).
Hence no correct party sends a ready message before \(\tau+2\delta\), so the instance is cleared at no correct party before \(\tau+3\delta\).
\end{proof}

\subsubsection{Consensus guarantees.}
Instantiating Generic Simplex (\Cref{fig:consensus-va}) with Fast VA gives a 1/5-resilient consensus protocol.
Since Fast VA satisfies VA-strong-unanimity, applying the timeout condition~\eqref{eq:consensus-timeout} to the delay bounds of \Cref{thm:va-fifth} shows that a view timeout \(\Delta_{\mathrm{to}}\geq4\Delta\) suffices (\Cref{app:strong-unanimity}).
Using \(d_t^s(\delta)=\delta\), the good-case, steady-state, and eventual worst-case commit latencies, block time, and \(k\)-gap latency satisfy
\begin{gather*}
\ell_g\leq2\delta,\qquad \ell_s\leq2\delta,\qquad \ell_w\leq3\delta,\qquad \text{block time}\leq2\delta,\\
k\text{-gap latency}\leq3\delta+k\bigl(\Delta_{\mathrm{to}}+3\delta\bigr).
\end{gather*}

\subsubsection{Communication complexity.}
Suppose \(n=5f+1\), and count messages sent by correct parties, with each threshold rule executing at most once per value.
Every value receiving \(2f+1\) votes at a correct party has at least \(f+1\) correct voters, so there are at most \(\lfloor n/(f+1)\rfloor<5\) such values, including invalid ones, since correct parties vote only once.
By~\Cref{lem:fifth-va-ready}, these are also the only non-\(\bot\) arguments of correct ready messages; hence each correct party broadcasts only \(O(1)\) \(\Vote\), \(\Disable\), and \(\Ready\) messages.
The instance therefore sends \(O(n^2)\) messages and \(O(n^2(L+\kappa))\) bits, with \(L\) and \(\kappa\) as in~\Cref{app:va-complexity}.
Generic Simplex adds only a leader proposal broadcast per view, preserving these bounds.

\clearpage
\section{Synchronous Blockchain Consensus}
\label{app:synchronous}

This appendix uses the synchronous setting of Appendix~\ref{app:extended-model}.
\Cref{app:synchronous-consensus} shows that Generic Simplex (\Cref{fig:consensus-va}) remains correct in this setting with view agreement implementations that need votes not to lag too far behind disable requests, and states its guarantees in terms of the delay parameters.
\Cref{app:synchronous-va} presents such an implementation for \(f<n/4\), which gives the same commit latencies as Generic Simplex with Fast VA (\Cref{app:va-fifth}).

\subsection{Generic Simplex in Synchrony}
\label{app:synchronous-consensus}

In the synchronous setting, \(\mathrm{GST}=0\) and every execution is \(\Delta\)-bounded from time~\(0\), in the sense of~\Cref{app:consensus-setting}.
As in the eventually synchronous case, a party keeps processing the messages of every view agreement instance after leaving the corresponding view.
Write \(D=\Delta_{\mathrm{to}}\) for the view timeout.

A synchronous view agreement implementation may guarantee VA-consistency only when votes do not lag too far behind disable requests, as does the implementation of~\Cref{app:synchronous-va}.
Such an implementation does not satisfy the last assumption of~\Cref{app:consensus-setting}, that the guarantees and delay bounds of view agreement hold for every pattern and timing of input calls that respects the interface.
We replace that assumption with the following contract.

\begin{definition}[Lag bound]
\label{def:lag-bound}
A view agreement implementation has \emph{lag bound} \(T\geq0\) if its guarantees and delay bounds, read as in~\Cref{app:consensus-setting}, hold in every synchronous execution in which correct parties respect the interface and no correct party calls \(\fn{vote}\) on the instance more than \(T\) after a correct party calls \(\fn{request\_disable}\) on it.
\end{definition}

An implementation that satisfies the assumptions of~\Cref{app:consensus-setting}, such as those of~\Cref{tab:signature-free-va-latency}, has every lag bound.
The condition constrains only calls that occur, so parties that leave an already-cleared view without calling either function do not affect it.

\begin{lemma}
\label{lem:sync-lag}
Consider a synchronous execution of Generic Simplex with an implementation of lag bound \(T\geq d_t(\Delta)\).
Then no correct party calls \(\fn{vote}\) on any instance more than \(T\) after a correct party calls \(\fn{request\_disable}\) on it.
Hence the guarantees and delay bounds of view agreement hold for every instance in the execution.
\end{lemma}

\begin{proof}
Correct parties respect the interface of every instance by~\Cref{lem:consensus-basic}(b).
Fix a view \(v\), let \(u\) be a correct party that votes in its instance, and let \(w\) be a correct party that calls \(\fn{request\_disable}()\) on it.
By~\Cref{lem:consensus-basic}(c), \(w\) does so at time \(e_w(v)+D\).
By~\Cref{lem:consensus-basic}(b,c), \(u\) votes while in view \(v\), and not after requesting disabling, which it does at time \(e_u(v)+D\) if it is still in view \(v\) then.
Hence \(u\) votes by time \(e_u(v)+D\), at most \(e_u(v)-e_w(v)\) after the request of \(w\).

It remains to show that \(e_u(v)-e_w(v)\leq d_t(\Delta)\), since \(d_t(\Delta)\leq T\).
We proceed by strong induction on \(v\).
All correct parties enter view~\(1\) at time~\(0\).
For \(v>1\), the induction hypothesis gives the condition of~\Cref{def:lag-bound} for every view below \(v\), so the totality delay of those instances holds with \(\delta=\Delta\) and \(T_0=0\).
The proof of~\Cref{lem:consensus-view-synchronization} uses only these bounds, so every correct party enters view \(v\) by time \(t_f(v)+d_t(\Delta)\).
Since \(e_w(v)\geq t_f(v)\), this bounds \(e_u(v)-e_w(v)\).
\end{proof}

The view timeout \(D\) cancels in this argument because the lag is measured from disable requests, which a correct party makes exactly \(D\) after entering the view.
A contract that measured it from a party's first call on the instance would need \(T\geq D+d_t(\Delta)\), since a party may vote as soon as it enters the view.
If VA-totality and its delay bound hold without the condition of~\Cref{def:lag-bound}, as for the implementation of~\Cref{app:synchronous-va}, the induction is unnecessary: \Cref{lem:consensus-view-synchronization} applies directly.

\begin{lemma}
\label{lem:sync-continuation}
\Cref{lem:va-continuation} holds for an implementation with a lag bound in synchronous executions.
\end{lemma}

\begin{proof}
In the proof of \Cref{lem:va-continuation}, choose the extension \(\sigma\) of \(\pi\) to be \(\Delta\)-bounded from~\(0\) as well as \(\delta\)-bounded from \(T_0\).
This is possible because \(\pi\) is a prefix of an execution with both properties, so every message still in transit at the end of \(\pi\) can be delivered by both deadlines.
No correct party calls \(\fn{request\_disable}()\) on the instance in \(\sigma\), so the condition of \Cref{def:lag-bound} holds vacuously, however late the missing votes are supplied.
The guarantees and delay bounds of view agreement therefore apply to \(\sigma\), and the rest of the proof is unchanged.
\end{proof}

\begin{theorem}
\label{thm:sync-consensus}
In the synchronous setting, suppose that Generic Simplex uses a view agreement implementation with lag bound \(T\geq d_t(\Delta)\), and that the timeout condition~\eqref{eq:consensus-timeout} holds.
Then the results of~\Cref{app:consensus-proofs} hold with \(\mathrm{GST}=0\), including \Cref{thm:consensus-totality,thm:consensus-liveness,cor:consensus-eventual-commit,thm:consensus-latency}, the view-duration bound of~\Cref{cor:consensus-view-duration}, and the commit bound of~\Cref{lem:consensus-correct-leader} for views with correct leaders.
Moreover, \Cref{thm:blockchain-consistency,thm:blockchain-validity} hold for every view timeout.
\end{theorem}

\begin{proof}
By \Cref{lem:sync-lag}, the guarantees and delay bounds of view agreement hold for every instance in the execution, which is \(\Delta\)-bounded from \(\mathrm{GST}=0\).
\Cref{app:consensus-proofs} uses the assumption that they hold for every pattern and timing of input calls only in~\Cref{lem:va-continuation}, which reasons about an alternative execution; everywhere else, it uses the guarantees of the instances in the execution itself.
\Cref{lem:va-continuation} holds here by~\Cref{lem:sync-continuation}, so the results of~\Cref{app:consensus-proofs} carry over.
\Cref{lem:sync-lag} does not depend on the view timeout, and the safety proofs of~\Cref{app:consensus-safety} use only VA-consistency and VA-validity, so Consistency and Validity hold for every view timeout.
\end{proof}

Thus, as in the eventually synchronous case, the guarantees of Generic Simplex depend only on the delay parameters of the implementation, now together with its lag bound.
Consider an execution that is \(\delta\)-bounded from \(T_0\geq0\), where \(\delta\leq\Delta\).
If every correct party has entered a view by time \(s\geq T_0\), then all correct parties enter the next view by time \(s+D+\max\{d_f(\delta),d_t(\delta)\}\) (\Cref{cor:consensus-view-duration}).
A view \(v\) with a correct leader and \(t_f(v)\geq T_0\) commits its block at every correct party by time \(t_f(v)+2d_t(\delta)+d_c(\delta)\) (\Cref{lem:consensus-correct-leader}).
The latency metrics satisfy the bounds of~\Cref{thm:consensus-latency}.

\subsection{A 1/4-Resilient View Agreement Implementation}
\label{app:synchronous-va}

\Cref{fig:va-sync-quarter} presents a synchronous implementation for \(f<n/4\), which uses the shared state and helper procedures of~\Cref{fig:va-shared}.
Its internal timer starts only upon \(\fn{request\_disable}()\).
Hence \(T\) only has to cover votes cast after some correct party requests disabling, as in~\Cref{def:lag-bound}, and progress from split votes may require later disable requests.
A value is committed directly from \(n-f\) votes, regardless of disable requests, which gives VA-strong-unanimity (\Cref{def:va-strong-unanimity}).
When the actual message delay is bounded by \(\delta\leq\Delta\), its prepare and commit delays are \(\delta\), its fallback delay is \(T+\Delta+2\delta\), its totality delay is \(2\delta\), and its Byzantine-silent totality delay is \(\delta\).
\Cref{tab:sync-va-latency} summarizes these bounds.

\begin{figure*}[htp]
    \centering
    \small
    \setlength{\fboxsep}{2pt}%
    \fbox{%
    \begin{minipage}{\dimexpr\textwidth-2\fboxsep-2\fboxrule\relax}
    \begin{minipage}[t]{0.485\linewidth}
    \begin{algorithmic}[1]
        \State \kw{import} Fig.~\ref{fig:va-shared};
        \BlankLine
        \State \textbf{Private state:}
        \State \ind a timer, initially unset;
        \BlankLine
        \State \kw{upon} $\fn{vote}(x)$:
        \State \ind \kw{broadcast} $\sig{\Vote, x}$;
        \BlankLine
        \State \kw{upon} $\fn{request\_disable}()$:
        \State \ind \kw{if} the timer has not been set:
        \State \indd set the timer for $T+\Delta$;
        \algstore{vasyncquarter}
    \end{algorithmic}
    \end{minipage}\hfill
    \begin{minipage}[t]{0.485\linewidth}
    \begin{algorithmic}[1]
        \algrestore{vasyncquarter}
        \State \kw{upon} $\nrecv{\sig{\Vote, x}} \geq n-f$:
        \State \ind \kw{call} $\fn{do\_commit}(x)$;
        \BlankLine
        \State \kw{upon} $\nrecv{\sig{\Vote, x}} \geq n-2f$
        \State \indd and $\fn{valid}(x)$:
        \State \ind \kw{broadcast} $\sig{\Ready, x}$;
        \BlankLine
        \State \kw{upon} the timer expiring:
        \State \ind \kw{if} no $\sig{\Ready, \cdot}$
        \State \inddd has been broadcast:
        \State \indd \kw{broadcast} $\sig{\Ready, \bot}$;
        \BlankLine
        \State \kw{upon} $\nrecv{\sig{\Ready, a}} \geq f+1$:
        \State \ind \kw{broadcast} $\sig{\Ready, a}$;
        \BlankLine
        \State \kw{upon} $\nrecv{\sig{\Ready, a}} \geq 2f+1$:
        \State \ind \kw{call} $\fn{clear}(a)$;
    \end{algorithmic}
    \end{minipage}
    \end{minipage}}
    \caption{A synchronous view agreement implementation for \(f<n/4\),
    assuming that no correct party calls \(\fn{vote}\) more than \(T\) after a
    correct party calls \(\fn{request\_disable}\) and that message delay is
    bounded by \(\Delta\). Vote-counting thresholds ignore equivocating voters.}%
    \label{fig:va-sync-quarter}
\end{figure*}

\begin{theorem}
\label{thm:va-sync-quarter}
Assume \(n\geq4f+1\) and the synchrony conditions stated
in~\Cref{fig:va-sync-quarter}.
Then the protocol in~\Cref{fig:va-sync-quarter} implements view agreement and satisfies VA-strong-unanimity.
Moreover, if the actual message delay is bounded by \(\delta\leq\Delta\), its
prepare and commit delays are \(\delta\), its fallback delay is
at most \(T+\Delta+2\delta\), its totality delay is
at most \(2\delta\), and its Byzantine-silent totality delay is \(\delta\).
Only VA-consistency relies on the bound \(T\) relating votes to disable
requests; the other properties and the latency bounds hold without it.
\end{theorem}

\begin{proof}
A correct party \(w\) makes its first \(\fn{request\_disable}()\) call at time \(r_w\) and sets its timer to expire \(T+\Delta\) later.
By assumption, every correct vote is sent by \(r_w+T\), so it reaches \(w\) by \(r_w+T+\Delta\).
At that deadline, \(w\) processes arriving messages before firing its timer (\Cref{sec:model}), so it has processed every correct vote before its timer expires.
Only the proof of VA-consistency uses this fact.

\noindent\textbf{VA-validity.}
A set of \(n-f\) or \(n-2f\) votes includes a vote from a correct party, which called \(\fn{vote}\) with that value.
The commit procedure checks \(\fn{valid}\).
The first correct sender of READY for a value \(x\) must have received \(n-2f\) votes for \(x\): the timer sends READY only for \(\bot\), and \(f+1\) READY messages would include an earlier correct sender.
That sender also checked \(\fn{valid}(x)\).
Since \(2f+1\) READY messages include a correct sender, preparing from READY messages is also valid.

\noindent\textbf{VA-consistency.}
Suppose \(x\) is committed. Its \(n-f\) voters and any \(n-2f\) voters for
\(x'\neq x\) would share at least \(n-3f>f\) parties, forcing a correct
party to vote twice. Thus \(x'\) cannot be committed, prepared from votes,
or sent in READY after receiving votes.

As shown above, every correct party processes the \(n-2f\) correct votes for \(x\) before its timer expires.
Each correct party has therefore sent READY for \(x\), so its timer sends
nothing. The first correct party to send READY for \(x'\neq x\), a value or
\(\bot\), cannot use the \((f+1)\)-READY rule for \(x'\), because that rule
requires a READY message already sent by a correct party.
Hence no correct party sends READY for \(x'\), and the Byzantine parties alone cannot make a correct party clear the instance with \(x'\).
Therefore no correct party prepares a value other than \(x\) or disables.

\noindent\textbf{VA-strong-unanimity.}
If all correct parties vote for the same valid value \(x\), the
last at time \(t\), all receive \(n-f\) matching votes and commit, and thus
prepare, \(x\) by \(t+\delta\).
This argument does not depend on requests to disable, so it establishes VA-strong-unanimity, and hence VA-unanimity.

\noindent\textbf{VA-totality.}
If a correct party clears the instance with \(a\), a value or \(\bot\), at time \(t\), then either \(n-2f\) correct parties have sent votes for \(a\), or \(f+1\) correct parties have sent READY for \(a\).
In either case, all correct parties receive those messages and send READY for \(a\) by \(t+\delta\), then clear the instance with \(a\) by \(t+2\delta\).
These steps do not depend on call times.

\noindent\textbf{VA-fallback-progress.}
Fix a set \(H\) of \(n-f\) correct parties. Consider a value
\(x\) sent in READY by a correct party. As shown in VA-validity,
the first correct sender received \(n-2f\) votes for \(x\). At most
\(f\) senders are outside \(H\), so at least \(n-3f\) parties in
\(H\) voted for \(x\). Three distinct such values would require
\(3(n-3f)>n-f\) voters in \(H\), since a correct party votes only
once. Hence correct READY messages carry at most two distinct values,
and at most three distinct arguments in total, counting \(\bot\).

Suppose that every correct party calls \(\fn{request\_disable}()\) or the instance
is cleared at that party, and let \(\tau\) be the first time at which this
holds for every correct party. If the instance is cleared at some correct
party by \(\tau\), VA-totality ensures that it is cleared at every correct party by
\(\tau+2\delta\).
Otherwise, every correct party has requested disabling and set its timer by \(\tau\), so by \(d=\tau+T+\Delta\) every correct party has sent a READY message: if it has sent none before its timer expires, the timer sends READY for \(\bot\).
Choose one message from each party in \(H\). There are \(n-f\geq3f+1\)
messages and at most three arguments, so at least \(f+1\) messages carry
the same argument \(a\).

By \(d+\delta\), every correct party receives those \(f+1\) READY
messages and sends READY for \(a\).
By \(d+2\delta\), every correct party receives at least \(n-f\geq2f+1\) READY messages for \(a\) and clears the instance with \(a\).
Thus the instance is cleared at every correct party, and the fallback delay is at most \(T+\Delta+2\delta\).
This also proves VA-fallback-progress, whose
premise implies that every correct party calls \(\fn{request\_disable}()\) or the
instance is cleared at that party.

\noindent\textbf{Latency.}
The prepare and commit delays of \(\delta\) follow from the proof of
VA-strong-unanimity, the totality delay of at most \(2\delta\) from the proof of
VA-totality, and the fallback delay from the proof of VA-fallback-progress.

For the Byzantine-silent totality delay, suppose that Byzantine parties send no messages in the instance and that a correct party prepares or disables at time \(t\).
It does so upon \(n-f\) votes or \(2f+1\) READY messages, all broadcast by correct parties by time \(t\).
By \(t+\delta\), every correct party has received the same messages and, since the corresponding rules have no precondition other than validity, clears the instance with the same value or \(\bot\).
Thus the Byzantine-silent totality delay is \(\delta\).

These arguments assume that every message delay is at most \(\delta\), but, as for the implementations of~\Cref{tab:signature-free-va-latency}, they also establish the bounds in the form of~\Cref{app:consensus-setting}.
Each bound is a chain of message deliveries that starts from messages sent by correct parties by the time \(s\) of the triggering event, preceded in the fallback case by timers that expire by time \(s+T+\Delta\).
In an execution that is \(\delta\)-bounded from \(T_0\), the \(i\)th delivery of the chain therefore occurs by time \(\max\{s,T_0\}+i\delta\), or by \(\max\{s,T_0\}+T+\Delta+i\delta\) after the timers.
These bounds are nondecreasing in \(\delta\) and at least \(\delta\), as~\Cref{app:consensus-setting} requires.
\end{proof}

\begin{table}[tbp]
    \centering
    \small
    \caption{Upper bounds on the latency parameters of the synchronous view agreement implementation of~\Cref{fig:va-sync-quarter}. See the \hyperref[par:va-latency]{latency definitions in~\Cref*{sec:blockchain-va}} for the meaning of the bounds. The implementation satisfies VA-strong-unanimity.}
    \label{tab:sync-va-latency}
    \setlength{\tabcolsep}{8pt}
    \begin{tabular}{@{}lc@{}}
        \toprule
        Protocol resilience & \(f<n/4\) (synchronous) \\
        \midrule
        Prepare delay \(d_p(\delta)\) & \(\delta\) \\
        Commit delay \(d_c(\delta)\) & \(\delta\) \\
        Fallback delay \(d_f(\delta)\) & \(T+\Delta+2\delta\) \\
        Totality delay \(d_t(\delta)\) & \(2\delta\) \\
        Byzantine-silent totality delay \(d_t^s(\delta)\) & \(\delta\) \\
        \bottomrule
    \end{tabular}
\end{table}

\subsubsection{Consensus guarantees.}
By~\Cref{thm:va-sync-quarter}, if \(n\geq4f+1\), the implementation of~\Cref{fig:va-sync-quarter} with parameter \(T\) has lag bound \(T\) and satisfies VA-strong-unanimity, with \(d_t(\Delta)=2\Delta\).
\Cref{thm:sync-consensus} therefore applies when \(T\geq2\Delta\) and \(D\geq2d_t(\Delta)=4\Delta\) (with only VA-unanimity, the timeout condition~\eqref{eq:consensus-timeout} would require \(5\Delta\)), so we may take
\[
 D=4\Delta,
 \qquad T=2\Delta.
\]
The internal VA timer then expires \(T+\Delta=3\Delta\) after the local disable request, and the fallback delay is at most \(T+\Delta+2\delta=3\Delta+2\delta\leq5\Delta\).
By the bounds of~\Cref{app:synchronous-consensus}, once every correct party has entered a view, all correct parties enter the next view within \(D+\max\{d_f(\delta),d_t(\delta)\}=7\Delta+2\delta\leq9\Delta\), even though \(d_f(\delta)\) may exceed \(D\).
A view \(v\) with a correct leader and \(t_f(v)\geq T_0\) commits its block by time \(t_f(v)+2d_t(\delta)+d_c(\delta)=t_f(v)+5\delta\).
Increasing \(T\) beyond \(2\Delta\) only lengthens the fallback path and leaves the fast path unchanged.
Unlike with the eventually synchronous implementations, Consistency depends on the timing assumptions, through the lag bound.

Substituting the bounds of \Cref{tab:sync-va-latency} into~\Cref{thm:consensus-latency} gives \Cref{tab:sync-consensus-latency}.
Compared with Generic Simplex using Fast VA (\Cref{tab:signature-free-consensus-latency}), the commit latencies, block time, and view timeout are unchanged at a higher resilience.
The price is the fallback path: each faulty-led view costs \(\Delta_{\mathrm{to}}+T+\Delta+2\delta\) instead of \(\Delta_{\mathrm{to}}+3\delta\), since disabling waits for the VA timer.

\begin{table}[tbp]
    \centering
    \small
    \caption{Upper bounds for the five blockchain consensus latency metrics of Generic Simplex with the synchronous \(f<n/4\) view agreement implementation, with sufficient choices of the view timeout and of \(T\).
    For the view timeout, the first value relies on VA-strong-unanimity and the second only on VA-unanimity.}
    \label{tab:sync-consensus-latency}
    \setlength{\tabcolsep}{4pt}
    \begin{tabular}{@{}lc@{}}
        \toprule
        Protocol resilience & \(f<n/4\) (synchronous) \\
        \midrule
        Good-case commit latency \(\ell_g\) & \(2\delta\) \\
        Steady-state commit latency \(\ell_s\) & \(2\delta\) \\
        Eventual worst-case commit latency \(\ell_w\) & \(3\delta\) \\
        Block time & \(2\delta\) \\
        View timeout \(\Delta_{\mathrm{to}}\) & \(\geq4\Delta\,/\,5\Delta\) \\
        VA timing bound \(T\) & \(\geq2\Delta\) \\
        \(k\)-gap latency & \(3\delta+k(\Delta_{\mathrm{to}}+T+\Delta+2\delta)\) \\
        \bottomrule
    \end{tabular}
\end{table}

\clearpage
\section{Minimmit in the View Agreement Framework}
\label{app:minimmit}

We recover the core of Minimmit~\cite[Section~4, Algorithm~1]{chouMinimmitFastFinality2026} by instantiating Generic Simplex (\Cref{fig:consensus-va}) with a signed implementation of view agreement.
We assume the PKI setting of Appendix~\ref{app:extended-model} and \(n\geq5f+1\).
We consider only the consensus algorithm, abstracting away transaction and block dissemination and the optimizations in the Minimmit paper.

\paragraph{Signed view agreement.}
Fix one view agreement instance; all messages below are signed and include its view number, which we omit from the notation.
There are two message types: \(\sig{\Vote,x}\) and \(\sig{\msgtag{nullify}}\).
The latter expresses support for disabling the instance.
The same votes support both preparation and commitment, at different thresholds:
\begin{center}
\begin{tabular}{@{}l@{\quad}l@{\quad}l@{}}
\toprule
Certificate & Contents (distinct signers) & Effect \\
\midrule
\(M(x)\) & \(2f+1\) votes for \(x\) & Prepare \(x\) \\
\(L(x)\) & \(n-f\) votes for \(x\) & Commit \(x\) \\
\(N\) & \(2f+1\) \msgtag{nullify} messages & Disable \\
\bottomrule
\end{tabular}
\end{center}
These are Minimmit's M-notarizations, L-notarizations, and nullifications, respectively; preparation and commitment require \(\fn{valid}(x)\).
Receiving a certificate also supplies its constituent signed messages, whose original signers are counted even when the messages are forwarded.
An \(L(x)\) contains an \(M(x)\).

\Cref{fig:va-minimmit} presents the implementation, using the state and procedures of \Cref{fig:va-shared}.
A timeout alone never causes a voter to nullify: after voting for \(x\), a party needs opposing evidence from \(2f+1\) distinct signers, each supplying either a \msgtag{nullify} message or a vote for some \(y\neq x\).
These opposing votes need not all support the same value, and this rule does not require a preceding call to \(\fn{request\_disable}()\).
Forwarding is once per value for \(M(x)\), and once per instance for \(N\), and continues after the instance is cleared or the party leaves its view.
Forwarding \(L(x)\) is unnecessary: forwarding its \(M(x)\) suffices for VA-totality of view agreement.

\begin{figure*}[!t]
    \centering
    \small
    \setlength{\fboxsep}{2pt}%
    \fbox{%
    \begin{minipage}{\dimexpr\textwidth-2\fboxsep-2\fboxrule\relax}
    \begin{minipage}[t]{0.485\linewidth}
    \begin{algorithmic}[1]
        \State \kw{import} Fig.~\ref{fig:va-shared};
        \BlankLine
        \State \textbf{Private state:}
        \State \ind $\var{vote\_value} \gets \None$;
        \State \ind $\var{null\_sent} \gets \False$;
        \BlankLine
        \State \kw{procedure} $\fn{nullify}()$:
        \State \ind \kw{if} not $\var{null\_sent}$:
        \State \indd $\var{null\_sent} \gets \True$;
        \State \indd \kw{broadcast} signed $\sig{\msgtag{nullify}}$;
        \BlankLine
        \State \kw{upon} $\fn{vote}(x)$:
        \State \ind \kw{if} $\var{vote\_value} = \None$
        \State \inddd and not $\var{null\_sent}$:
        \State \indd $\var{vote\_value} \gets x$;
        \State \indd \kw{broadcast} signed $\sig{\Vote,x}$;
        \BlankLine
        \State \kw{upon} $\fn{request\_disable}()$:
        \State \ind \kw{if} $\var{vote\_value} = \None$
        \State \inddd and $\var{prepared} = \emptyset$
        \State \inddd and not $\var{disabled}$:
        \State \indd \kw{call} $\fn{nullify}()$;
        \algstore{vaminimmit}
    \end{algorithmic}
    \end{minipage}\hfill
    \begin{minipage}[t]{0.485\linewidth}
    \begin{algorithmic}[1]
        \algrestore{vaminimmit}
        \State \kw{upon} $\var{vote\_value} = x \neq \None$
        \State \indd and $\var{prepared} = \emptyset$
        \State \indd and not $\var{disabled}$
        \State \indd and $\#\!\bigl[\,\sig{\Vote,y}: y \neq x$
        \State \inddd or $\sig{\msgtag{nullify}}\,\bigr] \geq 2f+1$:
        \State \ind \kw{call} $\fn{nullify}()$;
        \BlankLine
        \State \kw{upon} first obtaining $M(x)$:
        \State \ind \kw{broadcast} $M(x)$;
        \State \ind \kw{if} $\fn{valid}(x)$:
        \State \indd \kw{call} $\fn{clear}(x)$;
        \State \ind \kw{else}:
        \State \indd \kw{call} $\fn{nullify}()$;
        \BlankLine
        \State \kw{upon} first obtaining $N$:
        \State \ind \kw{broadcast} $N$;
        \State \ind \kw{call} $\fn{clear}(\bot)$;
        \BlankLine
        \State \kw{upon} obtaining $L(x)$:
        \State \ind \kw{call} $\fn{do\_commit}(x)$;
    \end{algorithmic}
    \end{minipage}
    \end{minipage}}
    \caption{Minimmit's signed view agreement implementation for \(f<n/5\).
    Certificates \(M(x)\), \(L(x)\), and \(N\) contain \(2f+1\) votes for \(x\), \(n-f\) votes for \(x\), and \(2f+1\) \msgtag{nullify} messages, respectively.
    All counts refer to distinct original signers, and all messages are bound to this instance's view.
    Obtaining a certificate means receiving it or collecting its constituent signed messages.
    Certificate handlers remain active after leaving the view.}%
    \label{fig:va-minimmit}
\end{figure*}

The invalid-value branch makes the implementation live even when callers supply invalid inputs to view agreement.
This convention is safe because an M-notarization for one value excludes an L-notarization for every different value, as shown below.
It does not change the normal path of voting for a valid proposal, preparing it, and committing it.

\begin{theorem}
\label{thm:va-minimmit}
The signed implementation in \Cref{fig:va-minimmit} satisfies view agreement and VA-strong-unanimity (\Cref{def:va-strong-unanimity}) for \(n\geq5f+1\).
When message delay is bounded by \(\delta\), its prepare and commit delays are \(\delta\), its totality and Byzantine-silent totality delays are at most \(\delta\), and its fallback delay is at most \(2\delta\).
\end{theorem}

\begin{proof}
\textbf{VA-validity and VA-consistency.}
An M-notarization contains at least \(f+1\) correct votes, and the helpers check validity before preparing or committing.
If \(L(x)\) exists, its \(n-f\) signers intersect the signers of any \(M(y)\) in at least
\[
    (n-f)+(2f+1)-n=f+1
\]
parties, including a correct party.
Since correct parties vote at most once, no \(M(y)\) with \(y\neq x\) exists, hence neither does \(L(y)\).
Now suppose a valid \(x\) is committed, and let \(C\) be its correct voters, with \(|C|\geq n-2f\).
No member of \(C\) can nullify before voting, or because of an invalid M-notarization, since such a certificate would conflict with \(L(x)\).
A first member of \(C\) to nullify on opposing evidence would need \(2f+1\) signers outside \(C\), but there are at most \(2f\).
Thus no member of \(C\) nullifies, no nullification exists, and no correct party disables or prepares another value.
The shared commit procedure prepares a value before committing it.

\noindent\textbf{VA-strong-unanimity and VA-totality.}
If every correct party votes for the same valid \(x\) and none requests disabling before voting, no correct party can be the first to nullify, since a request after voting has no effect.
All correct parties receive \(n-f\) votes and both prepare and commit within \(\delta\) of the last correct vote.
Whenever a correct party prepares \(x\), it forwards \(M(x)\), so all correct parties prepare \(x\) within another \(\delta\).
Likewise, a disabling party forwards a nullification, giving the same totality bound.

\noindent\textbf{VA-fallback-progress and fallback delay.}
Let \(t\) be a time by which every correct party has requested disabling or the instance is cleared there.
If it is cleared anywhere at time \(t\), VA-totality ensures that it is cleared everywhere by \(t+\delta\).
Otherwise, every correct party has already sent a vote or a \msgtag{nullify} message.
Fix \(n-f\) correct parties.
If \(2f+1\) of them voted for one value \(x\), every correct party obtains \(M(x)\) by \(t+\delta\), preparing \(x\) if valid and otherwise sending \msgtag{nullify}.
If no value has that many votes among these parties, every voter receives opposing evidence from at least \((n-f)-2f=n-3f\geq2f+1\) of them by \(t+\delta\).
Consequently, either the instance is cleared at some correct party by \(t+\delta\), and VA-totality finishes by \(t+2\delta\), or every correct party has nullified by \(t+\delta\), and all obtain a nullification by \(t+2\delta\).
\end{proof}

\paragraph{Recovering the blockchain protocol.}
Use this implementation for each \(\var{va}[v]\) in \Cref{fig:consensus-va}, with blocks as values and genesis prepared and committed in view~0.
The leader extends an M-notarized block from the greatest earlier view in which it knows one; Minimmit resolves ties between blocks in that view lexicographically.
The safety check requires an M-notarized parent and disabled intervening views, corresponding to Minimmit's parent notarization and nullifications.
The parent hash identifies the particular block being extended, since M-notarizations in the same view need not be unique.

An M-notarization for a valid block or a nullification clears the current instance, so a party advances without waiting for an L-notarization.
Generic Simplex's catch-up vote is precisely Minimmit's extra vote before advancing on an M-notarization: a party that has neither voted nor requested disabling votes for the prepared block even if the block is not yet safe locally.
An L-notarization commits its block and ancestors independently of the party's current view.
The \(2\Delta\) view timer used by Minimmit invokes \(\fn{request\_disable}()\); after a vote, this call has no effect.
Since the implementation satisfies VA-strong-unanimity and \(d_t(\Delta)\leq\Delta\), this timer satisfies the timeout condition~\eqref{eq:consensus-timeout}, so the results of~\Cref{app:consensus-proofs} apply.
In particular, in a view with a correct leader whose first correct entry occurs after GST, all correct parties vote for the leader's block within \(2\Delta\) of that entry, before any of their timers fires, and all commit it (\Cref{lem:consensus-correct-leader}).

In a good-case view, one message delay delivers the proposal and one delivers the votes, giving a proposal-to-commit latency of \(2\delta\).
The next leader can propose as soon as it collects \(2f+1\) votes, without waiting for the \(n-f\) votes needed for commitment.
Thus preparation and commitment have the same worst-case good-case delay bound, but the smaller preparation quorum can advance the chain earlier within that voting round.

\clearpage
\section{Simplex in the View Agreement Framework}
\label{app:simplex}

We obtain the core notarization and finalization mechanism of Simplex~\cite[Section~2.1]{chanSimplexConsensusSimple2023} by instantiating Generic Simplex (\Cref{fig:consensus-va}) with the signed view agreement implementation below.
We assume the PKI setting of Appendix~\ref{app:extended-model} and \(n\geq3f+1\), and use Simplex's quorum size \(q=\lceil2n/3\rceil\).
Thus \(q\leq n-f\) and \(2q-n>f\); when \(n=3f+1\), we have \(q=2f+1=n-f\).
We focus on the consensus rules, with the interface adaptations described below.

\paragraph{Signed view agreement.}
Fix one view agreement instance; every message is signed and includes its view number, omitted below.
An ordinary vote \(\sig{\Vote,x}\) supports a value, a dummy vote \(\sig{\Vote,\bot}\) supports disabling, and \(\sig{\msgtag{finalize}}\) supports finalization of the view.
There are three certificates:
\begin{center}
\begin{tabular}{@{}l@{\quad}l@{\quad}l@{}}
\toprule
Certificate & Contents (distinct signers) & Effect \\
\midrule
\(Q(x)\) & \(q\) votes for \(x\) & Prepare valid \(x\) \\
\(N\) & \(q\) votes for \(\bot\) & Disable \\
\(F\) & \(q\) \msgtag{finalize} messages & Commit valid \(x\) with \(Q(x)\) \\
\bottomrule
\end{tabular}
\end{center}
Here \(Q(x)\) is an ordinary notarization, \(N\) is a dummy-block notarization, and \(F\) is a finalization certificate.
A \msgtag{finalize} message names only the view, not a value: quorum intersection ensures that at most one value can have an ordinary notarization.

\Cref{fig:va-simplex} uses the shared state and helpers of \Cref{fig:va-shared}.
Each party makes two separate signing decisions: it casts at most one ordinary vote, and sends at most one of a dummy vote and a \msgtag{finalize} message.
In particular, an ordinary voter may still send a dummy vote on \(\fn{request\_disable}()\).
A party sends \msgtag{finalize} once the instance is cleared, provided it has not already sent a dummy vote; subsequent requests to disable cannot reverse that decision.
The shared validity checks ensure that an invalid notarized value is neither prepared nor committed, and does not by itself trigger finalization.
Forwarding notarizations gives VA-totality; forwarding \(F\) is unnecessary for the view agreement properties.

\begin{figure*}[!t]
    \centering
    \small
    \setlength{\fboxsep}{2pt}%
    \fbox{%
    \begin{minipage}{\dimexpr\textwidth-2\fboxsep-2\fboxrule\relax}
    \begin{minipage}[t]{0.485\linewidth}
    \begin{algorithmic}[1]
        \State \kw{import} Fig.~\ref{fig:va-shared};
        \BlankLine
        \State \textbf{Private state:}
        \State \ind $\var{second\_vote} \gets \None$;
        \BlankLine
        \State \kw{upon} $\fn{vote}(x)$:
        \State \ind \kw{broadcast} signed $\sig{\Vote,x}$;
        \BlankLine
        \State \kw{upon} $\fn{request\_disable}()$:
        \State \ind \kw{if} $\var{second\_vote} = \None$:
        \State \indd $\var{second\_vote} \gets \bot$;
        \State \indd \kw{broadcast} signed $\sig{\Vote,\bot}$;
        \algstore{vasimplex}
    \end{algorithmic}
    \end{minipage}\hfill
    \begin{minipage}[t]{0.485\linewidth}
    \begin{algorithmic}[1]
        \algrestore{vasimplex}
        \State \kw{upon} first obtaining $Q(x)$:
        \State \ind \kw{broadcast} $Q(x)$;
        \State \ind \kw{call} $\fn{clear}(x)$;
        \BlankLine
        \State \kw{upon} first obtaining $N$:
        \State \ind \kw{broadcast} $N$;
        \State \ind \kw{call} $\fn{clear}(\bot)$;
        \BlankLine
        \State \kw{when} $\var{prepared} \neq \emptyset$ or $\var{disabled}$:
        \State \ind \kw{if} $\var{second\_vote} = \None$:
        \State \indd $\var{second\_vote} \gets \msgtag{finalize}$;
        \State \indd \kw{broadcast} signed $\sig{\msgtag{finalize}}$;
        \BlankLine
        \State \kw{when} holding $Q(x)$ and $F$:
        \State \ind \kw{call} $\fn{do\_commit}(x)$;
    \end{algorithmic}
    \end{minipage}
    \end{minipage}}
    \caption{Signed view agreement implementing Simplex's core for \(f<n/3\).
    Each certificate contains \(q=\lceil2n/3\rceil\) messages from distinct original signers: \(Q(x)\) contains ordinary votes for \(x\), \(N\) dummy votes, and \(F\) \msgtag{finalize} messages.
    All messages are bound to this instance's view.
    Obtaining a certificate means receiving it or collecting its constituent signed messages.
    Certificate handlers remain active after leaving the view.}%
    \label{fig:va-simplex}
\end{figure*}

\begin{theorem}
\label{thm:va-simplex}
The signed implementation in \Cref{fig:va-simplex} satisfies view agreement for \(n\geq3f+1\).
When message delay is bounded by \(\delta\), its prepare delay is \(\delta\), its commit delay is \(2\delta\), and its totality, Byzantine-silent totality, and fallback delays are at most \(\delta\).
\end{theorem}

\begin{proof}
\textbf{VA-validity and VA-consistency.}
Every \(Q(x)\) contains a correct vote, and the shared helpers check \(\fn{valid}(x)\).
Two quorums intersect in at least \(2q-n>f\) parties, including a correct party.
Since correct parties cast at most one ordinary vote, two distinct values cannot both have notarizations.
Since a correct party never sends both a dummy vote and a \msgtag{finalize} message, \(N\) and \(F\) cannot both exist.
Consequently, commitment from \(Q(x)\) and \(F\) excludes every conflicting preparation or commitment and every disabling event.
The shared commit procedure prepares a value before committing it.

\noindent\textbf{VA-unanimity and delays.}
Suppose all correct parties vote for the same valid \(x\), none requests disabling, and the last correct vote occurs at time \(t\).
By \(t+\delta\), every correct party has \(q\) votes for \(x\), prepares \(x\), and sends \msgtag{finalize}.
By \(t+2\delta\), every correct party also has \(F\) and commits \(x\).
Whenever a correct party prepares or disables, it forwards the corresponding notarization, so every correct party does the same within another \(\delta\).

\noindent\textbf{VA-fallback-progress and fallback delay.}
Let \(t\) be a time by which every correct party has requested disabling or the instance is cleared there.
If it is cleared at any correct party at time \(t\), VA-totality ensures that it is cleared everywhere by \(t+\delta\).
Otherwise, every correct party has requested disabling and none has sent \msgtag{finalize}, since that requires a cleared instance.
Thus all correct parties have sent dummy votes, and every correct party obtains \(N\) by \(t+\delta\).
\end{proof}

\paragraph{Recovering the blockchain protocol.}
Instantiate each \(\var{va}[v]\) in \Cref{fig:consensus-va} with this implementation, using blocks as values and treating genesis as prepared and committed in view~0.
Simplex's iteration or height corresponds to our view number, and its notarized dummy blocks correspond to disabled views.
A proposal extending a prepared block across disabled views therefore corresponds to a Simplex proposal whose parent chain contains notarized dummy blocks at those intervening heights.
Omitting these empty blocks from the represented chain leaves the transaction sequence unchanged.

An ordinary notarization prepares the block and permits advancing immediately, while a dummy notarization permits skipping the view.
If the party has not timed out, this advancement also sends \msgtag{finalize}; these messages can arrive while the next leader proposes its block.
Once \(Q(b)\) and \(F\) are available, the block and its ancestors are committed, even if the party has already entered a later view.
A real-block notarization may coexist with a dummy notarization, but finalization excludes the latter.

Use Simplex's \(3\Delta\) view timer to invoke \(\fn{request\_disable}()\).
This implementation does not satisfy VA-strong-unanimity (\Cref{def:va-strong-unanimity}), since a dummy vote after an ordinary vote withholds \msgtag{finalize}; since \(d_t(\Delta)\leq\Delta\) and \(d_p(\Delta)=\Delta\), Simplex's \(3\Delta\) timer satisfies the timeout condition~\eqref{eq:consensus-timeout}.
In a view whose first correct entry occurs after GST, forwarding notarizations brings all correct parties into the view within \(\delta\) of that entry.
With a correct leader, the proposal and its supporting evidence arrive within one further delay, and notarization takes one more.
All correct parties therefore prepare and send \msgtag{finalize} before their timers expire.
Measured from the leader's proposal, the next view starts within \(2\delta\) and the block commits within \(3\delta\).
If the leader is faulty, the fallback bound gives advancement within \(3\Delta+\delta\) of the last correct entry.
Thus the next proposal overlaps finalization, giving Simplex's \(2\delta\) good-case block time and \(3\delta\) proposal-to-commit latency.

Generic Simplex uses our interface conventions: it stops ordinary voting after a request to disable and may cast a catch-up vote before advancing on a prepared block.
The original Simplex rules allow an ordinary vote after a dummy vote and do not require this catch-up vote.
These conventions preserve the certificate arguments and latency bounds above, while giving an implementation of our interface rather than identical executions of the original protocol.
The framework's round-robin leader schedule can also be replaced by Simplex's random leader schedule; we do not use the paper's expected latency bounds here.

\end{document}